\documentclass[12pt]{article}
\usepackage{amsmath,amssymb,amsfonts,amsthm,bm,mathrsfs}
\usepackage{graphicx,booktabs,array,multirow}
\usepackage{psfrag,epsf}
\usepackage{enumerate}
\usepackage{natbib}
\usepackage{url} % not crucial - just used below for the URL 
\usepackage{xcolor}
\usepackage{xurl}
\usepackage{algorithm}
\usepackage{algpseudocode}
\usepackage[doublespacing]{setspace}
\usepackage{enumitem}
\usepackage{microtype}
\usepackage{hyperref}
\usepackage{mathtools}

\newtheorem{theorem}{Theorem}[section]
\newtheorem{proposition}[theorem]{Proposition}
\newtheorem{lemma}[theorem]{Lemma}

\newtheorem{assumption}{Assumption}[section]
\newtheorem{remark}{Remark}[section]

\newcommand{\ind}{\mathbb I}
\newcommand{\bx}{\bm x}
\newcommand{\bz}{\bm z}
\newcommand{\btheta}{\bm{\theta}}
\newcommand{\balpha}{\bm{\alpha}}
\newcommand{\bbeta}{\bm{\beta}}
\newcommand{\bgamma}{\bm{\gamma}}
\newcommand{\bpsi}{\bm{\psi}}
\newcommand{\bxi}{\boldsymbol{\xi}}

\newcommand{\norm}[1]{\left\lVert #1\right\rVert}

\newcommand{\Appendix}{\def\thesection{Appendix~\Alph{section}}\def\thesubsection{A.\arabic{subsection}}}

\begin{document}

\def\spacingset#1{\renewcommand{\baselinestretch}%
{#1}\small\normalsize} \spacingset{1}

%%%%%%%%%%%%%%%%%%%%%%%%%%%%%%%%%%%%%%%%%%%%%%%%%%%%%%%%%%%%%%%%%%%%%%%%%%%%%%

%\if1\blind
{
  \title{\bf Scalable Nonparametric Change-Plane Regression with Random Features: Estimation and Inference}
  \author{Wei Liu\hspace{.2cm}\\
    School of Mathematics, Sichuan University, Chengdu,  China \\
    and \\
   Shaogao Lv \\
    School of Statistics and Data Science\\ Southeast University, Nanjing, China}
    \date{}
  \maketitle
} %\fi

%\if0\blind
%{
 %    \title{\bf Scalable Nonparametric Change-Plane Regression with Random Features: Estimation and Inference}
  %    \date{}
  % \maketitle
%} \fi

\bigskip
\begin{abstract}
Threshold and change-plane models provide a useful framework for capturing
regime-dependent dynamics, but most existing approaches typically rely on scalar
thresholds or linear splitting boundaries. We propose a scalable
semiparametric change-plane regression framework that allows the regime
boundary to depend nonlinearly on multivariate state variables. Random Fourier features provide a flexible and computationally efficient
approximation to the nonlinear boundary function. A two-stage procedure first
estimates the induced regime classification using a smooth surrogate and then
refits the structural parameters by hard-threshold OLS. We establish convergence
of the first-stage regime classification and show that its estimation error is
asymptotically negligible for the second-stage refit, yielding oracle equivalence.
We further develop a dependent-multiplier bootstrap Sup-Wald test for detecting the existence of regime effects when the boundary is unidentified under the null, and establish
its asymptotic validity. Simulations demonstrate accurate regime recovery, near-oracle performance of the proposed estimator, and satisfactory finite-sample behavior of the
Sup-Wald test.
An application to postwar U.S. unemployment data provides strong evidence
of regime dependence, while the estimated nonlinear change plane identifies
a distinct stress regime characterized by significant shifts in short-run
unemployment dynamics.
\end{abstract}
\noindent%
\vfill

\newpage
\spacingset{1.45} % DON'T change the spacing!

\section{Introduction}
\label{sec:intro}

Threshold regression and change-plane models are widely used  in empirical macroeconomics and finance for capturing structural changes, state-dependent dynamics, and tipping points \citep[e.g.,][]{Tong1990,Hansen2000,caner2004instrumental,SeoLinton2007}. While early theoretical developments focused primarily on univariate threshold variables where sample splitting is determined by a scalar cutoff, macroeconomic and financial phenomena are rarely governed by a single metric. Instead, regime shifts may depend on composite, multidimensional economic conditions. For instance, labor-market dynamics may depend jointly on real economic
activity, financial conditions, and macroeconomic uncertainty
\citep[e.g.,][]{bernanke1999financial,gilchrist2012credit,ludvigson2021uncertainty}.
% For instance, in macroeconomic stress monitoring and labor market dynamics, the tipping point of an unemployment surge depends nonlinearly on the joint posture of real economic activity and financial frictions \citep[e.g.,][]{bernanke1999financial,gilchrist2012credit,ludvigson2021uncertainty}.
Restricting the regime boundary to a scalar threshold or a linear hyperplane may therefore lead to systematic regime misclassification and inaccurate estimation of regime-specific dynamics.

To relax the rigidity of scalar thresholds, a growing econometric literature
has explored multivariate sample splitting. One prominent branch extends the
scalar threshold to a parametric linear change plane, where the boundary is
governed by an unknown linear combination of covariates
\citep[e.g.,][]{SeoLinton2007,FanSongLu2017,lee2021factor,yu2021threshold,
zhang2022single,KangEtAl2025}. Although these models accommodate
multi-factor triggers, they impose the restriction that the regime boundary
must be an affine hyperplane. A parallel and more flexible strand of the
literature considers semi- and fully nonparametric threshold models, allowing
the threshold function or regression surfaces to be unknown
\citep[e.g.,][]{delgado2000nonparametric,henderson2015nonparametric,
kourtellos2016structural,LeeWang2023}. More recently, flexible nonlinear
regime boundaries with multiple covariates have also been explored through
classification-based approaches \citep{chang2026algorithm}.
These developments motivate a broader question: how can a nonlinear
multivariate regime boundary be represented flexibly while retaining
scalable estimation and tractable inference for the finite-dimensional
structural parameters? This issue is particularly important in threshold
regression because approximation errors in the boundary affect structural
estimation through the induced regime classification.

% To relax the rigidity of scalar thresholds, a growing econometric literature has explored multivariate sample splitting. One prominent branch extends the scalar threshold to a parametric linear change plane, where the boundary is governed by an unknown linear combination of covariates \citep[e.g.,][]{SeoLinton2007,FanSongLu2017,lee2021factor,yu2021threshold,zhang2022single,KangEtAl2025}. Although these models accommodate multi-factor triggers, they impose the strong restriction that the regime boundary must be an affine hyper-surface. A parallel and more flexible strand of literature pursues semi- and fully nonparametric threshold modeling, allowing the threshold function or the regression surfaces to be completely unknown \citep[e.g.,][]{delgado2000nonparametric,henderson2015nonparametric,kourtellos2016structural,LeeWang2023}. However, existing nonparametric formulations face a severe methodological barrier: they are practically constrained to low-dimensional (typically scalar) threshold variables. Extending traditional local polynomials, splines, or kernel weights to multivariate state vectors $\bz \in \mathbb{R}^p$ ($p \ge 2$) triggers an acute curse of dimensionality, while implicit reproducing kernel methods incur prohibitive $\mathcal{O}(n^2)$ storage and inversion burdens \citep[e.g.,][]{chang2026algorithm}. Consequently, applied researchers are trapped in an undesirable dichotomy between misspecified linear change planes and dimensionally restricted univariate nonparametric thresholds.

In this paper, we address this problem by developing a scalable
nonparametric change-plane regression framework based on randomized kernel
approximations. Drawing on the harmonic representation of positive-definite
kernels through Bochner's theorem \citep{RahimiRecht2007}, we approximate
the unknown transition boundary $f_0(\bz)$ using an explicit
$L$-dimensional random Fourier feature (RFF) expansion. This representation
replaces an implicit kernel representation with a finite-dimensional
Euclidean parameterization, allowing nonlinear boundary estimation to scale
to multivariate state vectors without directly constructing dense kernel
matrices. While random features have been extensively studied in convex
statistical learning and kernel ridge regression
\citep[e.g.,][]{Bach2017,RudiRosasco2017}, their use in threshold regression
leads to a distinct statistical mechanism: boundary approximation affects
structural estimation through regime misclassification near the threshold
surface, rather than requiring uniform recovery of the boundary function
over the entire domain.

% In this paper, we resolve this methodological tension by proposing a scalable multivariate nonparametric change-plane regression framework enabled by randomized kernel approximations. Drawing on the harmonic analysis of positive-definite kernels via Bochner's theorem \citep{RahimiRecht2007}, we project the unknown transition boundary $f_0(\bz)$ onto an explicit $L$-dimensional random Fourier feature (RFF) expansion. This explicit spectral feature map transforms an intractable functional optimization over an infinite-dimensional RKHS into a finite-dimensional Euclidean estimation problem.  While random features have been extensively studied in convex statistical learning and kernel ridge regression \citep[e.g.,][]{Bach2017,RudiRosasco2017}, their introduction into non-differentiable indicator gates entails a fundamentally different statistical mechanism: boundary approximation errors matter only through the measure of observations misclassified across the threshold surface, rather than pointwise uniform convergence over the entire domain.

To operationalize the framework, we introduce a two-stage estimation
procedure that separates nonlinear boundary estimation from final structural estimation. In the first stage, we replace the discontinuous regime indicator with a smooth surrogate \citep{horowitz1992smoothed,SeoLinton2007} and jointly estimate the boundary function and regression parameters. We then recover the discrete regime indicators by applying the estimated hard boundary. In the second stage, we re-estimate the structural coefficients by
ordinary least squares (OLS) using the resulting hard regime classification. To establish the theoretical properties of the two-stage estimator, we develop
a regime-classification analysis that avoids requiring uniform convergence of
the estimated boundary function. Under a classification margin condition
\citep{mammen1999smooth,tsybakov2004optimal} and a quantitative
regime-separation condition, we directly derive convergence rates for the
population and empirical regime misclassification errors. When these errors
vanish sufficiently fast, their effect on the second-stage OLS refit becomes
asymptotically negligible. Consequently, the finite-dimensional structural
parameters converge at the parametric $\sqrt{T}$-rate and attain oracle equivalence and have the same first-order asymptotic distribution as the
infeasible oracle estimator that knows the true regime memberships.

Beyond estimation and inference on the structural parameters, a complete
econometric analysis also requires determining whether a regime effect is
genuinely present in the data. While the oracle properties of our two-stage estimator are established under a well-defined regime structure, the parameters governing the transition boundary are not identified under the null hypothesis of no regime effect, \(H_0:\bbeta_0=\mathbf{0}\). This gives rise to the classical nuisance-parameter problem under the null \citep{davies1977hypothesis,davies1987hypothesis,andrews1993tests,hansen1996inference}. In conventional threshold and structural-change models, this non-identification leads to nonstandard limiting distributions for supremum-type Wald, Lagrange multiplier, and related test statistics \citep[e.g.,][]{hansen1996inference,Hansen2000,caner2004instrumental,lee2021factor}. In our framework, the nuisance parameter indexes a flexible family of nonlinear regime partitions generated by random Fourier features, leading to a Gaussian-process limit whose covariance structure depends on the underlying data-generating process. We therefore develop a dependent-multiplier bootstrap Sup-Wald test, denoted by \(\mathrm{Sup}\mathcal{W}_T\), tailored to the random-feature change-plane model. Using heteroskedasticity- and autocorrelation-consistent (HAC) covariance
estimation,  the proposed procedure provides asymptotically valid inference under serial dependence without estimating the unidentified transition boundary under the null. We further establish consistency of the test against fixed alternatives for which at least one candidate partition in the random-feature family exhibits a nonzero population score drift.

We summarize the three main contributions as follows:
\begin{enumerate}[label=(\roman*)]
    \item We develop a scalable semiparametric change-plane framework that accommodates flexible nonlinear regime boundaries driven by multivariate state variables, together with a two-stage estimation procedure.
    \item We establish convergence of the estimated regime classification and
    show that the refitted structural estimator achieves $\sqrt{T}$-consistency
    and oracle equivalence under suitable rate conditions.
    \item We develop a bootstrap $\mathrm{Sup}\mathcal{W}_T$ test for the existence of regime effects when the boundary parameter is unidentified under the null,
    allowing for heteroskedasticity and serial dependence.
\end{enumerate}
To demonstrate its practical performance, we evaluate the proposed methodology
through  simulation studies and an empirical analysis of postwar
U.S. unemployment dynamics from 1960 to 2025. In the empirical analysis, the proposed Sup-Wald test rejects the null of
no regime effect ($p=0.0025$), while the estimated nonlinear boundary
reveals a distinct stress regime with substantially different short-run
unemployment dynamics. The identified stress-regime observations tend to
cluster around major U.S. economic downturns, including the 2008 financial
crisis and the 2020 COVID-19 shock.

The remainder of the paper is organized as follows.
Section~\ref{sec:methodology} formalizes the econometric framework,
establishes point identification, introduces the two-stage RFF estimator,
and develops the dependent-multiplier bootstrap
$\mathrm{Sup}\mathcal{W}_T$ test for regime existence.
Section~\ref{sec:theory} establishes the asymptotic theory, while
Section~\ref{sec:computation} discusses the tuning-parameter
selection. Sections~\ref{sec:simulation} and~\ref{sec:empirical} present
Monte Carlo evidence and an empirical application to U.S. unemployment
dynamics, respectively. Section~\ref{sec:discussion} concludes with a brief discussion of possible
extensions.
All technical proofs and auxiliary lemmas are provided in Supplementary Materials. In addition, we implement our proposed method in an efficient and user-friendly R package called \textit{TRM.RFF}.

\section{Econometric Model, Two-Stage Estimation, and Hypothesis Testing}
\label{sec:methodology}

\subsection{Nonparametric change-plane model for time-series dynamics}
\label{subsec:model}

Let $\{(Y_t, \bx_t^\top, \bz_t^\top, U_t)^\top\}_{t=1}^T$ be a strictly stationary and ergodic time-series process defined on a probability space $(\Omega, \mathcal{F}, \mathbb{P})$, where $Y_t \in \mathbb{R}$ is the scalar response variable, $\bx_t \in \mathbb{R}^d$ collects structural covariates (including autoregressive lags and exogenous indicators), $\bz_t \in \mathbb{R}^p$ contains multi-dimensional macroeconomic or financial state variables governing the regime boundary, and $U_t \in \mathbb{R}$ denotes an observed threshold-triggering variable. We specify the multivariate nonparametric change-plane regression model as
\begin{equation}
\label{eq:model}
Y_t
=
\balpha_0^\top\bx_t
+
\bbeta_0^\top\bx_t D_{0t}
+
\varepsilon_t,
\qquad
D_{0t}
=
\ind\{f_0(\bz_t)\ge U_t\},
\end{equation}
where the conditional predictive relationship shifts abruptly from $\balpha_0$ in the baseline regime ($D_{0t}=0$) to $\balpha_0+\bbeta_0$ in the alternative regime ($D_{0t}=1$). The vector $\bbeta_0$ measures the structural jump across the regime boundary. The disturbance sequence $\varepsilon_t$ satisfies the conditional moment restriction $\mathbb{E}(\varepsilon_t \mid \mathcal{F}_{t}) = 0,$
where $\mathcal{F}_{t} = \sigma(\bx_t, \bz_t, U_t, Y_{t-1}, \bx_{t-1}, \bz_{t-1}, U_{t-1}, \ldots)$ denotes the information set available prior to the realization of $Y_t$.

Model \eqref{eq:model} generalizes several prominent econometric frameworks. When $f_0(\cdot)$ is restricted to an unknown constant $\gamma$, \eqref{eq:model} encompasses the classical scalar-threshold specification underlying TAR models~\citep{Tong1990,Hansen2000}. When $f_0(\bz) = c + \bgamma^\top\bz$, it corresponds to the linear change-plane formulation of \citet{SeoLinton2007}. Our primary interest centers on the general setting where $f_0: \mathbb{R}^p \to \mathbb{R}$ is an unknown smooth nonlinear function of multivariate state variables $\bz_t \in \mathbb{R}^p$ ($p \ge 2$).

\subsection{Identification and dependence structure}
\label{subsec:identification}

To accommodate temporal persistence and dynamic feedback characteristic of macroeconomic data, we place our analysis in a weakly dependent time-series setting. For two $\sigma$-fields $\mathcal{A}$ and $\mathcal{B}$, define the $\alpha$-mixing coefficient $\alpha(\mathcal{A}, \mathcal{B}) = \sup_{A \in \mathcal{A}, B \in \mathcal{B}} |\mathbb{P}(A \cap B) - \mathbb{P}(A)\mathbb{P}(B)|$. For the process $\bxi_t = (Y_t, \bx_t^\top, \bz_t^\top, U_t)^\top$, the mixing coefficient of lag $m$ is given by $\alpha(m) = \sup_{\tau \ge 1} \alpha\Big(\sigma(\bxi_s: s \le \tau),\, \sigma(\bxi_s: s \ge \tau + m)\Big).$

\begin{assumption}[Stationarity and mixing]
\label{ass:mixing}
The sequence $\{\bxi_t\}_{t=1}^T$ is strictly stationary and $\alpha$-mixing with geometric decay: $\alpha(m) \le C \rho^m$ for some constants $C > 0$ and $\rho \in (0, 1)$.
\end{assumption}
Assumption~\ref{ass:mixing} allows for temporal dependence while requiring
the dependence between observations separated by $m$ periods to decay
geometrically with $m$. The geometric decay of the $\alpha$-mixing coefficients is standard in the time-series threshold literature \citep[e.g.,][]{Tong1990,caner2004instrumental,lee2021factor}. Combined with the conditional mean restriction
$\mathbb{E}(\varepsilon_t\mid\mathcal F_t)=0$, the framework accommodates
lagged dependent variables and predetermined covariates. Moreover, since no conditional homoskedasticity restriction is imposed, the error variance may vary with the information set $\mathcal F_t$.

To characterized the smoothness/complexity of the thresholds function $f_0$, we now introduce a reproducing kernel Hilbert space (RKHS) $\mathcal{H}_k$. The RKHS is a Hilbert space consisting of functions defined on $\mathcal{Z}$, associated with a symmetric positive definite kernel $k:\bz\times\bz\to\mathbb{R}$.
It satisfies the reproducing property $\langle f, k(\cdot, \bz) \rangle_{\mathcal{H}_k} = f(x), \quad \forall f\in\mathcal{H}_k,\,\bz\in\mathcal{Z}.$
The Moore--Aronszajn theorem states that each symmetric positive definite kernel uniquely induces an RKHS. RKHS lays the theoretical foundation for kernel methods, which realize implicit nonlinear feature mappings without explicitly evaluating high-dimensional features.

\begin{assumption}[Boundary regularity and support]
\label{ass:boundary}
The boundary function satisfies $f_0 \in \mathcal{H}_k$ over compact support $\mathcal{Z} \subset \mathbb{R}^p$. The conditional density $p_{U\mid \bx,\bz}(u\mid\bx,\bz)$ is continuously differentiable and bounded away from zero and infinity in a neighborhood of $u=f_0(\bz)$, and the regime probability satisfies $\mathbb{P}(D_{0t}=1) \in (c, 1-c)$ for some $c \in (0, 1/2)$.
\end{assumption}
Assumption~\ref{ass:boundary} imposes a smoothness condition on the regime
boundary by requiring $f_0 \in \mathcal{H}_k$, which guarantees uniform approximation by randomized Fourier features \citep{RahimiRecht2007,Bach2017}. In contrast to classical threshold regressions \citep{Hansen2000,caner2004instrumental} that restrict regime splitting to a scalar cutoff, and the nonparametric threshold approach of \citet{LeeWang2023},
which estimates a threshold function of a scalar conditioning variable
using local-constant Taylor-approximate methods, our framework handles a fully multivariate state vector $\bz_t \in \mathbb{R}^p$ while retaining a finite-dimensional random-feature representation. The density condition $p_{U\mid\bx,\bz}(u\mid\bx,\bz) \ge \underline{p} > 0$ ensures that the transition surface is populated by positive probability mass, ruling out unidentifiable empty boundary segments.
\begin{assumption}[Regime separation and non-degeneracy]
\label{ass:separation}
(i) The jump parameter satisfies $\bbeta_0 \neq \mathbf{0}$; (ii) the regime-specific design matrices $\boldsymbol{\Sigma}_0 = \mathbb{E}\big[\bx_t\bx_t^\top (1 - D_{0t})\big]
\quad\text{and}\quad
\boldsymbol{\Sigma}_1 = \mathbb{E}\big[\bx_t\bx_t^\top D_{0t}\big]$
are strictly positive definite with minimal eigenvalues bounded away from zero. (iii) For any $f\in\mathcal H_k$ such that $\mathbb{P}\{D_t(f)\neq D_{0t}\}>0$, \[ \inf_{\balpha,\bbeta\in\mathbb{R}^d} \mathbb{E}\left[ \left\{ \balpha_0^\top\bx_t+\bbeta_0^\top\bx_tD_{0t} -\balpha^\top\bx_t-\bbeta^\top\bx_tD_t(f) \right\}^2 \right]>0, \] where $D_t(f)=\ind\{f(\bz_t)\ge U_t\}$.
%For any candidate boundary function $f \in \mathcal{H}_k$ such that the symmetric difference event $\Delta(f, f_0) := \big\{ (\bz_t, U_t): \ind(f_0(\bz_t) \ge U_t) \neq \ind(f(\bz_t) \ge U_t) \big\}$ has strictly positive probability $\mathbb{P}\big(\Delta(f, f_0)\big) > 0$, the conditional expected jump energy is strictly positive:
%     \[
%     \mathbb{E}\big[ (\bx_t^\top\bbeta_0)^2 \mid \Delta(f, f_0) \big] > 0.
%     \]
\end{assumption}
Assumption~\ref{ass:separation}(ii) is standard full-rank conditions ensuring OLS invertibility within each regime \citep[e.g.,][]{Hansen2000,caner2004instrumental}.  Assumption~\ref{ass:separation}(iii) imposes point identification of the regime partition:
any candidate boundary that disagrees with the true partition with positive
probability must incur strictly positive excess quadratic risk. Importantly,
this condition does not impose uniform separation over all $f\neq f_0$ and
therefore allows the excess risk to vanish as a candidate boundary approaches the truth. Above three Assumptions are sufficient to establish point identification of model~\eqref{eq:model}, as formalized in the following
proposition.

\begin{proposition}[Point identification]
\label{prop:id}
Under Assumptions~\ref{ass:mixing}--\ref{ass:separation}, if a parameter tuple $(\balpha, \bbeta, f)$ satisfies $\mathbb{E}\!\left[
\big\{Y_t - \balpha^\top\bx_t - \bbeta^\top\bx_t\ind(f(\bz_t)\ge U_t)\big\}^2
\right]
=
\mathbb{E}(\varepsilon_t^2),$
then $\balpha=\balpha_0$, $\bbeta=\bbeta_0$, $\ind\{f(\bz_t)\ge U_t\}=D_{0t}$ almost surely, and $f(\bz)=f_0(\bz)$ for almost every $\bz \in \mathcal{Z}$.
\end{proposition}

\subsection{Scalable random Fourier feature boundary approximation}
\label{subsec:rff}
A direct kernel-based representation of $f_0\in\mathcal H_k$ typically
requires storing and manipulating a dense $T\times T$ Gram matrix, leading
to $\mathcal O(T^2)$ memory and, under direct matrix inversion,
$\mathcal O(T^3)$ computational cost. Such costs become particularly burdensome when the boundary must be
repeatedly updated during change-plane estimation. We circumvent this dimensionality bottleneck via randomized spectral expansions \citep{RahimiRecht2007}.

By Bochner's theorem, a continuous, shift-invariant positive-definite
kernel can be represented through its spectral measure $\Lambda$:
\begin{equation}
\label{eq:bochner}
k(\bz,\bz')
=
\int_{\mathbb{R}^p} e^{i\boldsymbol{\omega}^\top(\bz-\bz')} \mathrm{d}\Lambda(\boldsymbol{\omega})
=
\mathbb{E}_{\boldsymbol{\omega},b}\big\{2\cos(\boldsymbol{\omega}^\top\bz+b)\cos(\boldsymbol{\omega}^\top\bz'+b)\big\}, \nonumber
\end{equation}
where $\boldsymbol{\omega} \sim \Lambda$ and $b \sim \mathrm{Uniform}(0, 2\pi)$. For the Gaussian radial basis function (RBF) kernel $k(\bz,\bz')=\exp(-\norm{\bz-\bz'}_2^2 / (2\sigma_k^2))$, the spectral measure is Gaussian: $\Lambda = \mathcal{N}(0, \sigma_k^{-2} \mathbf{I}_p)$.

By drawing $L$ independent random spectral frequencies $\{\boldsymbol{\omega}_\ell\}_{\ell=1}^L \stackrel{i.i.d.}{\sim} \Lambda$ and uniform phase shifts $\{b_\ell\}_{\ell=1}^L \stackrel{i.i.d.}{\sim} \mathrm{Uniform}(0, 2\pi)$, we construct the explicit, finite-dimensional randomized feature mapping $\bpsi_L: \mathbb{R}^p \to \mathbb{R}^L$:
\begin{equation}
\label{eq:rff_map}
\bpsi_L(\bz)
=
\sqrt{\frac{2}{L}}
\Big(\cos(\boldsymbol{\omega}_1^\top\bz+b_1),\ldots,\cos(\boldsymbol{\omega}_L^\top\bz+b_L)\Big)^\top.
\end{equation}
The empirical inner product
$\widetilde k_L(\bz,\bz')=\bpsi_L(\bz)^\top\bpsi_L(\bz')$
is an unbiased Monte Carlo approximation of $k(\bz,\bz')$. Consequently, we approximate the unknown boundary function $f_0(\bz)$ by the linear sieve:
\begin{equation}
\label{eq:rfboundary}
f_{\bgamma}(\bz) = \bgamma^\top\bpsi_L(\bz),
\qquad \bgamma \in \mathbb{R}^L. \nonumber
\end{equation}
This formulation reduces a kernel-based boundary estimation problem to an explicit $L$-dimensional Euclidean parameter estimation, reducing the operational storage from $\mathcal{O}(T^2)$ to $\mathcal{O}(TL)$ with $L \ll T$.

\subsection{First-stage smoothed surrogate gate optimization}
\label{subsec:smooth}

The exact sample least-squares criterion for model \eqref{eq:model} is discontinuous:
\begin{equation}
\label{eq:hardcriterion}
Q_T(\balpha,\bbeta,\bgamma)
=
\frac1T\sum_{t=1}^T
\left[
Y_t-\balpha^\top\bx_t
-\bbeta^\top\bx_t
\ind\{\bgamma^\top\bpsi_L(\bz_t)\ge U_t\}
\right]^2
+
\lambda\norm{\bgamma}_2^2,
\end{equation}
where $\lambda > 0$ is a ridge penalty parameter. Because the indicator gate is discontinuous and has zero derivative almost
everywhere, direct optimization of \eqref{eq:hardcriterion} is computationally
challenging and is not amenable to standard gradient-based methods.

Following the smoothed threshold estimation approach of \cite{SeoLinton2007}, we approximate the discontinuous step with a smooth surrogate CDF $G(v)$, such as the standard logistic function. Defining $G_h(v) = G(v/h)$, where $h> 0$ is a bandwidth vanishing with $T$, we formulate the continuous surrogate objective:
\begin{equation}
\label{eq:smoothcriterion}
Q_{T,h}(\balpha,\bbeta,\bgamma)
=
\frac1T\sum_{t=1}^T
\left[
Y_t-\balpha^\top\bx_t
-\bbeta^\top\bx_t
G_{h}\{\bgamma^\top\bpsi_L(\bz_t)-U_t\}
\right]^2
+
\lambda\norm{\bgamma}_2^2.
\end{equation}
The preliminary parameters are obtained via gradient-based optimization:
\begin{equation}
(\widehat\balpha, \widehat\bbeta, \widehat\bgamma)
=
\arg\min_{\balpha,\bbeta,\bgamma} Q_{T,h}(\balpha,\bbeta,\bgamma). \nonumber
\end{equation}
The primary role of the first-stage estimator $\widehat\bgamma$ is to recover
the regime boundary and the associated discrete partition, rather than to
provide the final structural parameter estimates. Once the optimizer converges, we discard the continuous gate and construct the estimated discrete regime partition:
\begin{equation}
\label{eq:recovered_partition}
\widehat f(\bz) = \widehat\bgamma^\top\bpsi_L(\bz),
\qquad
\widehat D_t = \ind\{\widehat f(\bz_t)\ge U_t\}, \quad t=1,\ldots,T. \nonumber
\end{equation}

\subsection{Two-stage OLS refitting and robust HAC inference}
\label{subsec:refit}

While the smooth surrogate gate facilitates first-stage optimization,
retaining the smoothed regime weights would carry the effects of smoothing
and regularization into the final structural parameter estimates. We therefore
use the first stage only to recover the regime partition and re-estimate the
structural coefficients using the resulting hard classification.

To separate structural parameter estimation from the first-stage smoothing
and regularization, we perform a second-stage OLS refit conditional on the
recovered hard partition $\{\widehat D_t\}_{t=1}^T$:
\begin{equation}
\label{eq:refit}
\bar\btheta
=
\begin{pmatrix}\bar\balpha\\ \bar\bbeta\end{pmatrix}
=
\arg\min_{\balpha,\bbeta}
\frac1T\sum_{t=1}^T
\left[
Y_t - \balpha^\top\bx_t - \bbeta^\top\bx_t\widehat D_t
\right]^2. \nonumber
\end{equation}
Letting $\mathbf{W}_t(D) = (\bx_t^\top, D\bx_t^\top)^\top$ and $\widehat{\mathbf{W}}_t = \mathbf{W}_t(\widehat D_t)$, the closed-form refitted estimator is
\begin{equation}
\label{eq:refit_closed}
\bar\btheta
=
\left(\sum_{t=1}^T \widehat{\mathbf{W}}_t \widehat{\mathbf{W}}_t^\top\right)^{-1}
\sum_{t=1}^T \widehat{\mathbf{W}}_t Y_t.
\end{equation}
Let $\mathbf{W}_{0t} = \mathbf{W}_t(D_{0t})$. We show in Section~\ref{sec:theory} that $\bar\btheta$ is asymptotically equivalent to the infeasible oracle estimator $\widehat\btheta_{\mathrm{or}}$ that uses $D_{0t}$ and satisfies the asymptotical normal distribution:
\begin{equation}
\label{eq:asymptotic_normality_hac}
\sqrt{T}(\bar\btheta - \btheta_0) \ \rightsquigarrow\ \mathcal{N}\big(\mathbf{0},\, \boldsymbol{\Sigma}_{\mathbf{W}}^{-1} \boldsymbol{\Omega}_{\mathrm{LR}} \boldsymbol{\Sigma}_{\mathbf{W}}^{-1}\big), \nonumber
\end{equation}
where $\boldsymbol{\Sigma}_{\mathbf{W}} = \mathbb{E}[\mathbf{W}_{0t}\mathbf{W}_{0t}^\top]$ and  $\boldsymbol{\Omega}_{\mathrm{LR}} =  \sum_{j=-\infty}^\infty \mathbb{E}\big[(\mathbf{W}_{0t}\varepsilon_t)(\mathbf{W}_{0,t-j}\varepsilon_{t-j})^\top\big].$

To conduct robust asymptotic inference, we construct a Heteroskedasticity and Autocorrelation Consistent (HAC) plug-in estimator. Using the refitted residuals $\widehat\varepsilon_t = Y_t - \widehat{\mathbf{W}}_t^\top\bar\btheta$, define the sample Gram block $\widehat{\mathbf{A}}_T = \frac{1}{T}\sum_{t=1}^T \widehat{\mathbf{W}}_t \widehat{\mathbf{W}}_t^\top$ and the Bartlett kernel-weighted long-run variance estimator:
\begin{equation}
\label{eq:hac_omega}
\widehat{\boldsymbol{\Omega}}_{\mathrm{HAC}}
=
\widehat{\boldsymbol{\Gamma}}_0 + \sum_{j=1}^{m_T} \kappa\left(\frac{j}{m_T+1}\right) \big(\widehat{\boldsymbol{\Gamma}}_j + \widehat{\boldsymbol{\Gamma}}_j^\top\big), \nonumber
\end{equation}
where $\widehat{\boldsymbol{\Gamma}}_j = \frac{1}{T}\sum_{t=j+1}^T (\widehat{\mathbf{W}}_t \widehat\varepsilon_t)(\widehat{\mathbf{W}}_{t-j} \widehat\varepsilon_{t-j})^\top$, {$\kappa(u) = (1 - |u|) \, \mathbb{I}(|u| \le 1)$} is the Bartlett kernel, and $m_T$ is a bandwidth parameter, i.e. $m_T = \lfloor 4(T/100)^{2/9} \rfloor$~\citep{newey1994automatic}. The HAC asymptotic covariance estimator is then given by
\begin{equation}
\label{eq:hac_sandwich}
\widehat{\mathbf{V}}_{\mathrm{HAC}} = \widehat{\mathbf{A}}_T^{-1} \widehat{\boldsymbol{\Omega}}_{\mathrm{HAC}} \widehat{\mathbf{A}}_T^{-1}.
\end{equation}
The standard error for the $k$-th element of $\bar\btheta$ is computed as $\mathrm{SE}(\bar\theta_k) = T^{-1/2} [\widehat{\mathbf{V}}_{\mathrm{HAC}}]_{kk}^{1/2}$.

The full two-stage estimation pipeline summarized in Algorithm~\ref{alg:main}.
\begin{algorithm}[ht]
\caption{Two-stage multivariate change-plane estimation}
\label{alg:main}
\begin{algorithmic}[1]
\Require Time series $\{(Y_t, \bx_t, \bz_t, U_t)\}_{t=1}^T$ and tuning parameters $(L, \lambda, h)$.
\Ensure Structural coefficients $\bar\btheta$, regime series $\{\widehat D_t\}_{t=1}^T$, and HAC covariance $\widehat{\mathbf{V}}_{\mathrm{HAC}}$.
\State \textbf{Random Fourier features}: Generate spectral draws $\{\boldsymbol{\omega}_\ell\}_{\ell=1}^L$ and phases $\{b_\ell\}_{\ell=1}^L$; construct $\bpsi_L(\bz_t)$ via \eqref{eq:rff_map}.
\State \textbf{First-stage smoothed gate optimization}: Minimize the surrogate objective \eqref{eq:smoothcriterion} using L-BFGS-B with multiple starts to obtain $(\widehat\balpha, \widehat\bbeta, \widehat\bgamma)$.
\State \textbf{Hard regime recovery}: Extract the discrete partition $\widehat D_t = \ind\{\widehat\bgamma^\top\bpsi_L(\bz_t) \ge U_t\}$ for $t=1,\ldots,T$.
\State \textbf{Second-stage OLS refitting}: Compute the refitted parameter vector $\bar\btheta$ via \eqref{eq:refit_closed} on $\widehat{\mathbf{W}}_t = (\bx_t^\top, \bx_t^\top \widehat D_t)^\top$.
\State \textbf{HAC variance estimation}: Compute refitted residuals $\widehat\varepsilon_t = Y_t - \widehat{\mathbf{W}}_t^\top\bar\btheta$ and calculate $\widehat{\mathbf{V}}_{\mathrm{HAC}}$ via \eqref{eq:hac_sandwich}.
\State \Return $\bar\btheta$, $\{\widehat D_t\}_{t=1}^T$, and $\widehat{\mathbf{V}}_{\mathrm{HAC}}$.
\end{algorithmic}
\end{algorithm}
\subsection{Testing for structural regime existence}
\label{subsec:testing}

In addition to estimation, a fundamental empirical question is whether
a structural regime shift is present in the first place:
\begin{equation}
\label{eq:hypotheses}
H_0: \bbeta_0 = \mathbf{0}
\quad\text{versus}\quad
H_1: \bbeta_0 \neq \mathbf{0}.
\end{equation}
Testing \eqref{eq:hypotheses} gives rise to the classical
nuisance-parameter problem under the null
\citep{davies1977hypothesis,davies1987hypothesis,hansen1996inference}.
Under $H_0$, the boundary parameter $\bgamma$ is not identified because
the regime indicator drops out of the conditional mean model. Consequently,
supremum-type Wald, score, and likelihood-ratio statistics generally have
nonstandard limiting distributions characterized by functionals of
nuisance-indexed stochastic processes.

The role of random Fourier features in this testing context differs fundamentally from that in boundary estimation. In the estimation stage, we allow $L = L_T \to \infty$ so that the random-feature sieve increasingly approximates the unknown infinite-dimensional boundary $f_0$. For testing, however, because the boundary is unidentified under $H_0$, we keep the random-feature dimension $L$ fixed as $T \to \infty$. This allows the features to index a fixed-dimensional, flexible family of candidate regime partitions. Specifically, for $\bgamma \in \boldsymbol{\Gamma} \subset \mathbb{R}^L$, where $\boldsymbol{\Gamma}$ is a compact parameter space, we define the candidate regime indicator and regressor vector as
\[
D_t(\bgamma) := \ind\{\bgamma^\top\bpsi_L(\bz_t) \ge U_t\} \qquad \text{and} \qquad \mathbf{W}_t(\bgamma) := \begin{pmatrix} \bx_t \\ \bx_t D_t(\bgamma) \end{pmatrix}.
\]
The resulting test searches over this flexible family of nonlinear
candidate partitions for evidence of a structural regime effect.

Conditional on a fixed boundary direction $\bgamma \in \boldsymbol{\Gamma}$, the least-squares estimator of $\btheta = (\balpha^\top, \bbeta^\top)^\top$ is
\[
\widehat\btheta(\bgamma) = \begin{pmatrix} \widehat\balpha(\bgamma) \\ \widehat\bbeta(\bgamma) \end{pmatrix} := \widehat{\mathbf{A}}_T(\bgamma)^{-1} \left( \frac{1}{T}\sum_{t=1}^T \mathbf{W}_t(\bgamma) Y_t \right), \quad \text{where} \quad \widehat{\mathbf{A}}_T(\bgamma) := \frac{1}{T}\sum_{t=1}^T \mathbf{W}_t(\bgamma) \mathbf{W}_t(\bgamma)^\top.
\]
The pointwise Wald statistic for testing $\bbeta_0 = \mathbf{0}$ at a fixed $\bgamma$ is given by
\begin{equation}
\label{eq:pointwise_wald}
\mathcal{W}_T(\bgamma) := T \cdot \widehat\bbeta(\bgamma)^\top \Big[ \mathbf{R} \big( \widehat{\mathbf{A}}_T(\bgamma)^{-1} \widehat{\boldsymbol{\Omega}}_{\mathrm{HAC}}(\bgamma) \widehat{\mathbf{A}}_T(\bgamma)^{-1} \big) \mathbf{R}^\top \Big]^{-1} \widehat\bbeta(\bgamma),
\end{equation}
where $\mathbf{R} := [\mathbf{0}_{d \times d}, \mathbf{I}_d]$ is the selector matrix satisfying $\mathbf{R}\btheta = \bbeta$, and  $\widehat{\boldsymbol{\Omega}}_{\mathrm{HAC}}(\bgamma)$ is the
Bartlett HAC estimator of the long-run covariance of
$\mathbf W_t(\bgamma)\widehat\varepsilon_t(\bgamma)$. To account for the unidentified boundary parameter, we maximize the
pointwise Wald statistic over the candidate parameter space:
\begin{equation}
\label{eq:sup_wald}
\mathrm{Sup}\mathcal{W}_T := \sup_{\bgamma \in \boldsymbol{\Gamma}} \mathcal{W}_T(\bgamma).
\end{equation}

Because the asymptotic null distribution of $\mathrm{Sup}\mathcal{W}_T$ depends on the non-pivotal covariance kernel of the limit Gaussian process $\mathbb{G}_\infty(\cdot)$ and the unknown joint distribution of $(\bx_t, \bz_t, U_t)$, critical values cannot be tabulated. We approximate the null distribution using the dependent wild bootstrap,  detailed in Algorithm~\ref{alg:bootstrap}.

\begin{algorithm}[ht]
\caption{Dependent wild bootstrap test for regime existence}
\label{alg:bootstrap}
\begin{algorithmic}[1]
\Require Time series $\{(Y_t, \bx_t, \bz_t, U_t)\}_{t=1}^T$, finite evaluation grid $\boldsymbol{\Gamma}_G \subset \boldsymbol{\Gamma}$, bootstrap replications $B$, kernel bandwidth $b_T$.
\Ensure Empirical bootstrap $p$-value $p_{\mathrm{boot}}$ for testing $H_0: \bbeta_0 = \mathbf{0}$.
\State Fit the restricted linear model under $H_0$: $\widetilde\balpha := (\sum_{t=1}^T \bx_t\bx_t^\top)^{-1} \sum_{t=1}^T \bx_t Y_t$.
\State Compute restricted OLS residuals $\widetilde\varepsilon_t := Y_t - \bx_t^\top \widetilde\balpha$ for $t = 1, \dots, T$.
\State Compute the sample test statistic $\mathrm{Sup}\mathcal{W}_T = \max_{\bgamma \in \boldsymbol{\Gamma}_G} \mathcal{W}_T(\bgamma)$ via \eqref{eq:sup_wald}.
\For{$b = 1, \ldots, B$}
    \State Draw dependent multipliers $\boldsymbol{\eta}^{*(b)} \sim \mathcal{N}(\mathbf{0}, \mathbf{\Omega}^*)$ with covariance $\mathbf{\Omega}^*_{ts} = \kappa\big((t-s)/b_T\big)$.
    \State Construct pseudo-outcomes $Y_t^{*(b)} = \bx_t^\top \widetilde\balpha + \widetilde\varepsilon_t \eta_t^{*(b)}$ for $t = 1, \dots, T$.
    \State Compute the bootstrap test statistic $\mathrm{Sup}\mathcal{W}_T^{*(b)} = \max_{\bgamma \in \boldsymbol{\Gamma}_G} \mathcal{W}_T^{*(b)}(\bgamma)$.
\EndFor
\State \Return $p_{\mathrm{boot}} = \frac{1}{B} \sum_{b=1}^B \ind\{\mathrm{Sup}\mathcal{W}_T^{*(b)} \ge \mathrm{Sup}\mathcal{W}_T\}$.
\end{algorithmic}
\end{algorithm}
\begin{remark}[Comparison with classical threshold tests]
\label{rem:test_comparison}
The statistic in \eqref{eq:sup_wald} belongs to the class of
supremum-type tests with nuisance parameters unidentified under the null
\citep{davies1977hypothesis,hansen1996inference}. Unlike classical
scalar-threshold tests, the candidate regime partitions here are generated
by random Fourier features. For fixed $L$, the resulting indicator class
has finite VC dimension, providing a flexible family of nonlinear
multivariate partitions while retaining manageable empirical-process
complexity.
\end{remark}

\section{Statistical theory}
\label{sec:theory}
In this section, we establish the asymptotic properties of the proposed
two-stage change-plane procedure and Sup-Wald test. We first derive convergence rates for the
population and empirical regime misclassification errors induced by the
estimated boundary. We then establish the oracle equivalence and asymptotic
normality of the second-stage refitted estimator, together with the
consistency of the HAC covariance estimator. Finally, we derive the limiting
null distribution of the regime-existence test, establish the validity of
the dependent-multiplier bootstrap, and study the consistency of the test
under fixed alternatives.

\subsection{First-stage regime recovery}
\label{subsec:first_stage_rate}

We first introduce notation for measuring the accuracy of the recovered regime partition. Let $D_t(f) := \ind\{f(\bz_t) \ge U_t\}$, $D_{0t} := D_t(f_0)$, and define the regime-disagreement pseudo-metric $\rho(f, g) := \mathbb{P}\{D_t(f) \neq D_t(g)\}$. For any candidate boundary $f$, define the population profiled hard-threshold prediction risk
\[
\mathcal{R}(f) := \inf_{\balpha, \bbeta \in \mathbb{R}^d} \mathbb{E}\left[ \left\{ Y_t - \balpha^\top\bx_t - \bbeta^\top\bx_t D_t(f) \right\}^2 \right].
\]
Let  $\mathcal{F}_{L} := \left\{ f(\cdot) = \bgamma^\top\bpsi_{L}(\cdot) : \bgamma \in \boldsymbol{\Gamma}_{L}, \, \|\bgamma\|_2 \le C_\gamma \right\}$ denote the random-feature sieve where we allow $L=L_T$ to diverge as $T$, and let $f_{0, L}(\bz) := \bgamma_{0, L}^\top\bpsi_{L}(\bz)$ denote an $L$-dimensional RFF approximation to $f_0$.

For the smoothed criterion with bandwidth $h$, define $$m_t(f, \balpha, \bbeta) := \left[ Y_t - \balpha^\top\bx_t - \bbeta^\top\bx_t G\left\{ (f(\bz_t) - U_t)/h \right\} \right]^2,$$ and let $(\balpha(f), \bbeta(f))$ denote its population profile minimizer. For later use, define $\Delta m_t(f) := m_t(f, \balpha(f), \bbeta(f)) - m_t(f_{0, L}, \balpha(f_{0, L}), \bbeta(f_{0, L}))$, together with the centered processes $Z_t(f) := \Delta m_t(f) - \mathbb{E}[\Delta m_t(f)]$ and $H_t(f) := \ind\{D_t(f) \neq D_{0t}\} - \rho(f, f_0)$.

The following conditions govern the accuracy of first-stage regime recovery.
\begin{assumption}[Boundary margin condition]
\label{ass:margin}
There exist constants $C_m > 0$, $\kappa > 0$, and $\delta_0 > 0$ such that $\mathbb{P}\big(|U_t - f_0(\bz_t)| \le \delta\big) \le C_m \delta^\kappa$ for all $\delta \in (0, \delta_0]$.
\end{assumption}

Under Assumption~\ref{ass:boundary}, the margin condition holds with $\kappa=1$ and $C_m = 2\overline{p}$. Higher margin exponents $\kappa > 1$ characterize thinner probability mass around the transition surface and yield faster decay of the misclassification probability~\citep{mammen1999smooth,tsybakov2004optimal}.

\begin{assumption}[Spectral feature approximation]
\label{ass:rf}
The target boundary satisfies $f_0 \in \mathcal{H}_k$. For the continuous shift-invariant kernel $k$, there exist a coefficient sequence $\bgamma_{0,L} \in \mathbb{R}^L$, a deterministic rate $a_L \to 0$, and a constant $C_\gamma > 0$ such that, with probability approaching one over the spectral measure $\Lambda$,
\[
\|f_{0,L} - f_0\|_\infty \le a_L \qquad \text{and} \qquad \|\bgamma_{0,L}\|_2 \le C_\gamma.
\]
\end{assumption}

The approximation rate $a_L$ depends on the kernel, the smoothness of $f_0$, and the random-feature dimension; standard results provide conditions under which $a_L \to 0$~\citep{RudiRosasco2017}.

\begin{assumption}[Smoothness, boundedness, and moments]
\label{ass:smooth_kernel}
\leavevmode
\begin{enumerate}[label=(\roman*), leftmargin=2em]
    \item The surrogate CDF $G: \mathbb{R} \to [0, 1]$ is twice continuously differentiable with bounded derivatives $\|G'\|_\infty < \infty$, $\|G''\|_\infty < \infty$, and $\int_{\mathbb{R}} |G(v) - \ind\{v \ge 0\}|\,dv < \infty$.
    \item For some $\delta > 0$, $\mathbb{E}\|\bx_t\|_2^{4+\delta} < \infty$ and $\mathbb{E}|\varepsilon_t|^{4+\delta} < \infty$.
    \item There exist constants $\delta_M > 0$ and $C_M < \infty$ such that $\sup_{(z,u):\, |u-f_0(z)| \le \delta_M} \mathbb{E}\big( \|\bx_t\|_2^2 + \varepsilon_t^2 \mid \bz_t=z, U_t=u \big) \le C_M$.
\end{enumerate}
\end{assumption}

The local conditional moment bound in Assumption~\ref{ass:smooth_kernel}(iii) is required only in a fixed neighborhood of the true boundary, ruling out explosive conditional variability where misclassification can occur while accommodating arbitrary heterogeneity elsewhere.

\begin{assumption}[Quantitative regime separation]
\label{ass:classification_separation}
There exist constants $c_\rho > 0$ and $\delta_\rho > 0$ such that:
\begin{enumerate}[label=(\roman*), leftmargin=2em]
    \item \textbf{Local curvature:} $\mathcal{R}(f) - \mathcal{R}(f_0) \ge c_\rho \, \rho(f, f_0)$ for all $f \in \mathcal{F}_{L}$ satisfying $\rho(f, f_0) \le \delta_\rho$.
    \item \textbf{Global identification:}  For every fixed $\delta>0$,
there exists a constant $c_\delta>0$ such that, for all sufficiently
large $T$, $\inf_{f \in \mathcal{F}_{L}:\, \rho(f, f_0) \ge \delta} \{ \mathcal{R}(f) - \mathcal{R}(f_0) \} \ge c_\delta$.
\end{enumerate}
\end{assumption}
Assumption~\ref{ass:classification_separation} links excess prediction
risk to errors in regime classification. The global identification
condition rules out candidate boundaries that generate substantially
different regime partitions while attaining nearly the same population
risk as the true boundary. Once the estimator is localized around the
true regime partition, the local curvature condition provides a
quantitative lower bound on the excess risk in terms of the disagreement
probability $\rho(f,f_0)$, which is used to derive the regime-recovery rate.

\begin{assumption}[Localized long-run variance]
\label{ass:localized_lrv}
There exists a constant $C_{\mathrm{LR}} > 0$ such that, uniformly over $f \in \mathcal{F}_{L}$ with $\rho(f, f_0) \le \delta_\rho$,
\begin{align}
\sum_{j=-\infty}^{\infty} \big| \mathrm{Cov}\big(Z_t(f), Z_{t+j}(f)\big) \big| &\le C_{\mathrm{LR}} \big\{ \rho(f, f_0) + h + a_{L}^{\kappa} \big\}, \nonumber \\
 \nonumber
\sum_{j=-\infty}^{\infty} \big| \mathrm{Cov}\big(H_t(f), H_{t+j}(f)\big) \big| &\le C_{\mathrm{LR}} \, \rho(f, f_0). 
\end{align}
\end{assumption}
Assumption~\ref{ass:localized_lrv} controls the temporal dependence of
the localized criterion and disagreement processes. It requires their long-run variances to shrink as the candidate regime partition approaches the true partition, up to the smoothing and random-feature approximation errors. This localization permits the empirical-process fluctuations to
scale with the disagreement radius and is the key ingredient for obtaining
the regime-recovery rates under serial dependence.
\begin{theorem}[First-stage regime recovery]
\label{thm:regime_recovery}
Suppose Assumptions~\ref{ass:mixing}--\ref{ass:separation} and \ref{ass:margin}--\ref{ass:localized_lrv} hold. As $T \to \infty$, if $\eta_T := \frac{L\log T}{T} + h + \lambda + a_{L}^\kappa \to 0$, then
\begin{equation}
\label{eq:regime_recovery}
\rho(\widehat f, f_0) = \mathcal{O}_p(\eta_T) \qquad \text{and} \qquad \mathrm{MR}_T := \frac{1}{T}\sum_{t=1}^T \ind\{\widehat D_t \neq D_{0t}\} = \mathcal{O}_p(\eta_T).
\end{equation}
\end{theorem}

Theorem~\ref{thm:regime_recovery} directly characterizes first-stage accuracy through regime disagreement rather than uniform estimation of the boundary function. The composite rate $\eta_T$ captures stochastic complexity, smoothing approximation error, regularization, and random-feature approximation. In particular, both the population disagreement probability and the empirical misclassification rate are of order $\mathcal{O}_p(\eta_T)$, providing the key input for the oracle equivalence result below.
\subsection{Oracle asymptotic normality of the refitted estimator}
\label{subsec:theory_oracle}

We next establish the asymptotic properties of the second-stage OLS refitted estimator $\bar\btheta = (\bar\balpha^\top, \bar\bbeta^\top)^\top$ defined in \eqref{eq:refit_closed}. The key observation is that, although the first-stage boundary estimator is nonparametric, its induced regime misclassification becomes asymptotically negligible under an appropriate first-stage rate condition. To formalize this result, we first introduce the infeasible oracle estimator that uses the true regime indicator $D_{0t}$: 
\begin{equation}
\label{eq:oracle_estimator_def}
\widehat\btheta_{\mathrm{or}} = \left(\sum_{t=1}^T \mathbf{W}_{0t} \mathbf{W}_{0t}^\top\right)^{-1} \sum_{t=1}^T \mathbf{W}_{0t} Y_t. \nonumber
\end{equation}
where  $\mathbf{W}_{0t} = (\bx_t^\top, \bx_t^\top D_{0t})^\top$.

\begin{assumption}[First-stage negligibility rate]
\label{ass:rate}
Let $\nu := \frac{2+\delta}{4+\delta} \in (\frac{1}{2}, 1)$, where $\delta > 0$ is the moment exponent from Assumption~\ref{ass:smooth_kernel}(ii). As $T \to \infty$, the tuning parameters and approximation rates satisfy $\sqrt{T} \, \eta_T^\nu \longrightarrow 0,$ where $\eta_T := \frac{L \log T}{T} + h + \lambda + a_{L}^\kappa.$
\end{assumption}

Assumption~\ref{ass:rate} governs the rate at which first-stage estimation errors diminish. By Theorem~\ref{thm:regime_recovery}, the empirical misclassification rate satisfies $\mathrm{MR}_T = \mathcal{O}_p(\eta_T)$. Combined with the $(4+\delta)$-moment conditions in Assumption~\ref{ass:smooth_kernel}(ii), H\"older's inequality implies that the misclassification-weighted terms entering the second-stage Gram and score expansions scale as $\mathcal{O}_p(\mathrm{MR}_T^\nu) = \mathcal{O}_p(\eta_T^\nu)$. Consequently, the rate restriction $\sqrt{T}\eta_T^\nu \to 0$ ensures that the first-stage boundary recovery error is asymptotically negligible at the parametric $\sqrt{T}$ scale, preserving oracle equivalence for the refitted estimator.

\begin{theorem}[Oracle equivalence and asymptotic distribution]
\label{thm:oracle}
Suppose Assumptions~\ref{ass:mixing}--\ref{ass:separation} and \ref{ass:margin}--\ref{ass:rate} hold. Then the two-stage refitted estimator $\bar\btheta$ is asymptotically equivalent to the infeasible oracle estimator $\widehat\btheta_{\mathrm{or}}$:
\begin{equation}
\label{eq:oracle_diff_op1}
\sqrt{T}(\bar\btheta - \widehat\btheta_{\mathrm{or}}) = o_p(1). \nonumber
\end{equation}
Consequently, $\bar\btheta$ achieves the oracle asymptotic normal distribution:
\begin{equation}
\label{eq:oracle_clt}
\sqrt{T}(\bar\btheta - \btheta_0) \ \rightsquigarrow\ \mathcal{N}\big(\mathbf{0},\, \mathbf{V}_{\mathrm{or}}\big),
\qquad
\mathbf{V}_{\mathrm{or}} = \boldsymbol{\Sigma}_{\mathbf{W}}^{-1} \boldsymbol{\Omega}_{\mathrm{LR}} \boldsymbol{\Sigma}_{\mathbf{W}}^{-1}, \nonumber
\end{equation}
where $\boldsymbol{\Sigma}_{\mathbf{W}} = \mathbb{E}[\mathbf{W}_{0t}\mathbf{W}_{0t}^\top]$ and $\boldsymbol{\Omega}_{\mathrm{LR}} = \sum_{j=-\infty}^\infty \mathbb{E}\big[(\mathbf{W}_{0t} \varepsilon_t)(\mathbf{W}_{0,t-j} \varepsilon_{t-j})^\top\big].$
\end{theorem}

\begin{theorem}[Consistency of the HAC covariance estimator]
\label{thm:hac_consistency} 
Suppose $\mathbb{E}\|\bx_t\|_2^{8+\delta}<\infty, \mathbb{E}|\varepsilon_t|^{8+\delta}<\infty$ and the bandwidth parameter $m_T$ satisfies $m_T \to \infty$, $m_T / T^{1/2} \to 0$ and $m_T\eta_T^{\mu/2}\to 0$ as $T \to \infty$, where $ \mu=\frac{4+\delta}{8+\delta}$. Under the conditions of Theorem~\ref{thm:oracle}, it holds that $\widehat{\mathbf{V}}_{\mathrm{HAC}} \stackrel{p}{\longrightarrow} \mathbf{V}_{\mathrm{or}}.$
\end{theorem}
Theorem~\ref{thm:oracle} shows that the first-stage regime recovery error
is asymptotically negligible for the refitted estimator under
$\sqrt{T}\eta_T^\nu\to0$. Consequently, the refitted estimator has the
same first-order asymptotic distribution as the infeasible oracle estimator,
without requiring trimming near the estimated boundary.
Theorem~\ref{thm:hac_consistency} further establishes consistency of the HAC estimator for the corresponding long-run asymptotic covariance matrix,
thereby enabling valid large-sample inference.

\subsection{Asymptotic theory for the regime existence test}
\label{subsec:theory_testing}
We now establish the asymptotic properties of the global Sup-Wald statistic $\mathrm{Sup}\mathcal{W}_T$ in \eqref{eq:sup_wald} for testing regime existence.

Recall $D_t(\bgamma) = \ind\{\bgamma^\top\bpsi_L(\bz_t) \ge U_t\}$ and define $\mathbf{A}_{00} = \mathbb{E}[\bx_t\bx_t^\top]$, and $\mathbf{A}_{10}(\bgamma) = \mathbb{E}[D_t(\bgamma)\bx_t\bx_t^\top]$. The population partialled-out regressor for the structural jump is $\widetilde{\bx}_t(\bgamma) = D_t(\bgamma)\bx_t - \mathbf{A}_{10}(\bgamma)\mathbf{A}_{00}^{-1}\bx_t$. We define the empirical score process as $\mathbb{G}_T(\bgamma) := \frac{1}{\sqrt{T}} \sum_{t=1}^T \widetilde{\bx}_t(\bgamma)\varepsilon_t, \bgamma \in \boldsymbol{\Gamma},$
and denote its asymptotic covariance kernel by
\begin{equation}
\label{eq:score_cov_kernel}
\mathcal{K}(\bgamma, \bgamma') := \sum_{j=-\infty}^{\infty} \mathbb{E}\big[\widetilde{\bx}_t(\bgamma)\varepsilon_t \widetilde{\bx}_{t-j}(\bgamma')^\top\varepsilon_{t-j}\big], \qquad \mathbf{V}_0(\bgamma) := \mathcal{K}(\bgamma, \bgamma). \nonumber
\end{equation}
Let $\ell^\infty(\boldsymbol{\Gamma};\mathbb R^d)$ denote the Banach
space of bounded $\mathbb R^d$-valued functions on
$\boldsymbol{\Gamma}$ equipped with the uniform norm. We need the regularity conditions for regime-existence testing.
\begin{assumption}
\label{ass:test}
We assume
\begin{enumerate}[label=(\roman*), leftmargin=2em]
 \item  \textbf{Threshold density:} The conditional density
    $p_{U\mid\bx,\bz}(u\mid\bx,\bz)$ is uniformly bounded over
    a neighborhood of $\left\{
    \bgamma^\top\bpsi_L(\bz):
    \bgamma\in\boldsymbol{\Gamma},
    \ \bz\in\mathcal Z
    \right\}.$
    \item \textbf{Moments:} There exists $\delta > 0$ such that $\mathbb{E}\|\bx_t\|_2^{4+\delta} < \infty$ and $\mathbb{E}|\varepsilon_t|^{4+\delta} < \infty$.
    \item \textbf{Non-degeneracy:} The Gram matrix $\mathbf{A}_{00}$ is positive definite, and the asymptotic covariance satisfies $\inf_{\bgamma \in \boldsymbol{\Gamma}} \lambda_{\min}\{\mathbf{V}_0(\bgamma)\} > 0$.
    \item \textbf{Empirical process regularity:} The indexed score class $\mathcal{M} := \big\{ (\bx, \bz, u, \varepsilon) \mapsto [D(\bgamma)\mathbf{I}_{d_x} - \mathbf{A}_{10}(\bgamma)\mathbf{A}_{00}^{-1}]\bx\varepsilon : \bgamma \in \boldsymbol{\Gamma} \big\}$ satisfies the metric entropy and envelope conditions for stochastic equicontinuity and uniform HAC approximation under geometric $\alpha$-mixing.
    \item \textbf{Bandwidth rates:} As $T \to \infty$, the Bartlett bandwidth $m_T$ and the bootstrap kernel bandwidth $b_T$ satisfy $m_T \to \infty$, $b_T \to \infty$, with $m_T/\sqrt{T} \to 0$ and $b_T/\sqrt{T} \to 0$.
    \item \textbf{Continuity of limit functional:} The asymptotic distribution function $F_\infty(z) := \mathbb{P}(\mathrm{Sup}\mathcal{W}_\infty \le z)$ is continuous on $(0, \infty)$.
\end{enumerate}
\end{assumption}
\begin{theorem}[Limiting null distribution]
\label{thm:sup_wald_null}
Suppose Assumptions~\ref{ass:mixing} and \ref{ass:test}  hold. Under $H_0: \bbeta_0 = \mathbf{0}$, as $T \to \infty$:
\begin{enumerate}[label=(\roman*)]
    \item $\mathbb{G}_T(\cdot) \rightsquigarrow \mathbb{G}_\infty(\cdot)$ in $\ell^\infty(\boldsymbol{\Gamma};\mathbb R^d)$, where $\mathbb{G}_\infty(\cdot)$ is a mean-zero Gaussian process with covariance kernel $\mathcal{K}(\bgamma, \bgamma')$;
    \item $\mathrm{Sup}\mathcal{W}_T \rightsquigarrow \sup_{\bgamma \in \boldsymbol{\Gamma}} \mathbb{G}_\infty(\bgamma)^\top \mathbf{V}_0(\bgamma)^{-1} \mathbb{G}_\infty(\bgamma)$.
\end{enumerate}
\end{theorem}

Because the null limit depends on the covariance kernel $\mathcal{K}(\bgamma, \bgamma')$ and the design characteristics, it is non-pivotal. We therefore approximate critical values  using a dependent Gaussian multiplier bootstrap in Algorithm~\ref{alg:bootstrap}.

\begin{theorem}[Asymptotic validity of the bootstrap]
\label{thm:bootstrap_validity}
Suppose Assumptions~\ref{ass:mixing} and \ref{ass:test} hold. Under $H_0: \bbeta_0 = \mathbf{0}$, the conditional distribution of the bootstrap statistic consistently approximates the null distribution: $\sup_{z \in \mathbb{R}} \big| \mathbb{P}^*\big(\mathrm{Sup}\mathcal{W}_T^* \le z\big) - \mathbb{P}\big(\mathrm{Sup}\mathcal{W}_T \le z\big) \big| \stackrel{p}{\longrightarrow} 0.$
Consequently, for any nominal significance level $\alpha \in (0, 1)$, the ideal bootstrap $p$-value $p_{\mathrm{boot}}^{\mathrm{ideal}} := \mathbb{P}^*(\mathrm{Sup}\mathcal{W}_T^* \ge \mathrm{Sup}\mathcal{W}_T)$ satisfies $\mathbb{P}\big(p_{\mathrm{boot}}^{\mathrm{ideal}} \le \alpha\big) \longrightarrow \alpha \text{ as } T \to \infty.$
Moreover, if the number of bootstrap replications satisfies $B_T \to \infty$, the Monte Carlo $p$-value satisfies $\mathbb{P}(p_{\mathrm{boot}} \le \alpha) \to \alpha$.
\end{theorem}

\begin{theorem}[Consistency of the regime-existence test]
\label{thm:test_consistency}
Suppose Assumptions~\ref{ass:mixing} and \ref{ass:test} hold. Under $H_1: \bbeta_0 \neq \mathbf{0}$, define the drift function $\boldsymbol{\mu}(\bgamma):= \mathbb{E}[\widetilde{\bx}_t(\bgamma)\bx_t^\top D_{0t}]\bbeta_0$. If there exists $\bgamma_0 \in \boldsymbol{\Gamma}$ such that $\boldsymbol{\mu}(\bgamma_0) \neq \mathbf{0}$ and $\widehat{\mathbf{V}}_T(\bgamma_0) \stackrel{p}{\longrightarrow} \mathbf{V}_1(\bgamma_0)$ with $\lambda_{\min}\{\mathbf{V}_1(\bgamma_0)\} > 0$, then
\begin{equation}
\label{eq:test_divergence}
\frac{1}{T}\mathrm{Sup}\mathcal{W}_T \ge \boldsymbol{\mu}(\bgamma_0)^\top \mathbf{V}_1(\bgamma_0)^{-1} \boldsymbol{\mu}(\bgamma_0) + o_p(1), \nonumber
\end{equation}
implying $\mathrm{Sup}\mathcal{W}_T \stackrel{p}{\longrightarrow} \infty$. Consequently, if $c_{T,1-\alpha}^* = \mathcal{O}_p(1)$, then $\mathbb{P}(p_{\mathrm{boot}} \le \alpha) \to 1$ for any $\alpha \in (0, 1)$.
\end{theorem}
\section{Tuning parameters}
\label{sec:computation}

The proposed estimation procedure involves three main tuning parameters: the
random-feature dimension $L$, the smoothing bandwidth $h$, and the
ridge parameter $\lambda$. The dimension $L$ controls the flexibility
of the random-feature approximation, while $h$ determines the
smoothness of the surrogate function for the regime indicator. The ridge parameter is used primarily for numerical stabilization and, when chosen sufficiently small, has little effect on the final hard-regime refitted estimator, as demonstrated by the simulation results in Table~\ref{tab:ablation_all}. We therefore fix $\lambda=10^{-10}$ throughout. 
%For the Gaussian RBF kernel, the kernel bandwidth $\sigma_k$ is selected using the median heuristic.

For empirical implementation, $L$ and $h$ can be selected by
forward-chaining cross-validation to respect the temporal ordering of
the observations. Candidate values are evaluated according to
out-of-sample squared prediction error using the final hard regime
indicator rather than the smoothed training surrogate. In practice,
we find that moderate random-feature dimensions and relatively small
smoothing bandwidths work well; our numerical experiments use
$L\in\{15,20\}$ and $h=0.025$ as baseline choices. The sensitivity analysis in Table~\ref{tab:ablation_all} further shows that
the finite-sample results are stable across a range of choices of $L$ and $h$.

Computationally, constructing the random-feature matrix requires
$\mathcal O(TpL)$ operations and $\mathcal O(TL)$ storage. The
subsequent optimization is performed directly in this
finite-dimensional representation, avoiding the construction and
inversion of a dense $T\times T$ kernel matrix.

\section{Simulation Study}
\label{sec:simulation}

We evaluate the finite-sample performance of our proposed two-stage estimator and bootstrap test across diverse boundary specifications. Observations are generated from model~(1) with $\bx_t = (1, X_{2t})^\top$, $\balpha_0=(1, 2)^\top$, and $\bbeta_0=(-1.5, 1.5)^\top$. We draw $(U_t, \bz_t)^\top \stackrel{i.i.d.}{\sim} \mathcal{N}(\mathbf{0}, \mathbf{I}_{p+1})$ and set $X_{2t}\mid(U_t,\bz_t) \sim \mathcal{N}(0, \{1+0.5(U_t^2+\norm{\bz_t}_2^2)\}^{-1})$. To capture serial persistence, disturbances follow an AR(1) process: $\varepsilon_t = 0.3 \varepsilon_{t-1} + \nu_t$ with $\nu_t \stackrel{i.i.d.}{\sim} \mathcal{N}(0, 0.1^2)$. We examine three scenarios with distinct boundary forms: (1) \emph{Linear}: $f_0(z) = 0.3 + z$; (2) \emph{Univariate nonlinear}: $f_0(z) = \sin(z) + 0.3$; and (3) \emph{Multivariate nonlinear}: $f_0(\bz) = \sin(\bz^\top\bgamma) + 0.3$ with $p=3$ and $\bgamma = (1, -1, 0.5)^\top/\sqrt{2.25}$.

We compare six estimators across $T \in \{200, 500\}$ and $M=1{,}000$ replications: the proposed first-stage surrogate estimator (\textbf{TRM-RF}) and its refitted counterpart (\textbf{TRM-RF-refit}); the smoothed least-squares linear change-plane estimator of \citet{SeoLinton2007} (\textbf{TRM-SLS}) and its refitted version (\textbf{TRM-SLS-refit}); the univariate threshold regression of \citet{Hansen2000} (\textbf{TRM-Hansen}) using $U_t$; and the infeasible \textbf{Oracle} refit based on the true $D_{0t}$. Performance is evaluated by: (i) root of mean square error (RMSE) of parameter estimators; (ii) empirical coverage probability (CP, \%) of nominal 95\% confidence level; and (iii) the regime misclassification rate, $\mathrm{MR} = T^{-1}\sum_{t=1}^T \ind(\widehat D_t \neq D_{0t}) \times 100\%$.

\subsection{Simulation results}
\label{subsec:sim_results}

Tables~\ref{tab:perf_scen1}--\ref{tab:perf_scen3} report the finite-sample performance across Scenarios~1--3. Across all boundary geometries and dependence structures, three central empirical findings emerge:

\paragraph{Vulnerability of Flat Boundaries versus Nonparametric Flexibility.}
When the true partition entails multivariate interaction or geometric curvature, restricting the regime-switching mechanism to a scalar threshold (\textbf{TRM-Hansen}) or a flat hyperplane (\textbf{TRM-SLS}, \textbf{TRM-SLS-refit}) incurs substantial misspecification. In Scenario~1, the scalar-threshold specification of
\textbf{TRM-Hansen} cannot accommodate the state-dependent boundary, resulting in misclassification rates exceeding 24\% and empirical coverage for the structural jump $\beta_1$ falling to 0.0\% at both $T=200$ and $T=500$.

More critically, in Scenarios~2 and 3 where the true boundary is nonlinear, linear change planes cannot conform to boundary manifolds, leaving an irreducible misclassification error of roughly 2.5\%--3.0\%. While \textbf{TRM-SLS-refit} exhibits ostensibly moderate RMSEs in small samples, this pattern is consistent with persistent boundary misspecification becoming increasingly consequential for inference as sampling variability decreases. Specifically, as $n$ grows from 200 to 500, the coverage of $\beta_1$ under \textbf{TRM-SLS-refit} plummets from 60.9\% to 5.5\% in Scenario~2, and from 63.9\% to a mere 7.5\% in Scenario~3 (with $\alpha_1$ coverage deteriorating to 22.3\% and 27.1\%, respectively). In sharp contrast, the random Fourier feature approximation flexibly accommodates the nonlinear boundary geometries considered here, maintaining misclassification rates below $0.38\%$ and yielding
coverage probabilities close to the nominal level after refitting.

\paragraph{Benefits of Discrete Refitting.}
Across all three scenarios, contrasting continuous surrogate-gate estimators (\textbf{TRM-RF}, \textbf{TRM-SLS}) with their discrete refitted counterparts (\textbf{TRM-RF-refit}, \textbf{TRM-SLS-refit}) provides striking finite-sample validation of our two-stage estimation rationale. In the first stage, the smooth surrogate gate carries the effects of smoothing and regularization into the preliminary structural estimates. Under serially correlated disturbances, these effects become more consequential for coverage as sampling uncertainty decreases.

Once the discrete regime partition $\widehat D_t = \ind(U_t \le \widehat f(\bz_t))$ is reconstructed, performing OLS refitting substantially reduces these distortions. Across the three scenarios, \textbf{TRM-RF-refit} generally reduces parameter RMSEs and substantially improves empirical coverage.

\paragraph{Near-Oracle Finite-Sample Performance.}
Most prominently, \textbf{TRM-RF-refit} achieves empirical performance that closely approaches the performance of the infeasible \textbf{Oracle} estimator that knows the true partition a priori. Even in the challenging multivariate nonlinear boundary of Scenario~3, the parameter RMSEs of \textbf{TRM-RF-refit} at $T=500$ $(0.009, 0.013, 0.011, 0.017)$ mirror the Oracle benchmarks $(0.008, 0.012, 0.009, 0.015)$ almost, while tracking the Oracle's coverage (93.4\%--95.2\% vs. Oracle's 92.4\%--95.7\%). Furthermore, both the Oracle and \textbf{TRM-RF-refit} exhibit mild undercoverage at $T=200$, consistent with the known finite-sample limitations of kernel-based HAC inference under serial dependence
\citep{den199712}. As $T$ increases to 500, empirical coverage moves
closer to the nominal 95\% level.

Crucially, in Scenario~1 where the linear boundary assumption is satisfied and \textbf{TRM-SLS-refit} is correctly specified, \textbf{TRM-RF-refit} performs comparably to \textbf{TRM-SLS-refit}, suggesting little finite-sample cost from the additional boundary flexibility. These numerical findings support Theorem~\ref{thm:oracle}: under the established margin and rate conditions, the first-stage regime recovery error is asymptotically negligible for second-stage structural inference.

\begin{table}[!ht]
\centering
\small
\caption{Finite-sample performance under Scenario 1 (linear boundary, $f_0(z)=\gamma_0+\gamma_1 z$): RMSE, 95\% empirical coverage probabilities (CP, \%), and misclassification rate (MR, \%) based on 1,000 replications.}\label{tab:perf_scen1}
\setlength{\tabcolsep}{4.5pt}
\begin{tabular}{llccccccccc}
\toprule
& & \multicolumn{4}{c}{RMSE} & \multicolumn{4}{c}{95\% CP (\%)} & \\
\cmidrule(lr){3-6} \cmidrule(lr){7-10}
$n$ & Method & $\alpha_1$ & $\alpha_2$ & $\beta_1$ & $\beta_2$ & $\alpha_1$ & $\alpha_2$ & $\beta_1$ & $\beta_2$ & MR (\%) \\
\midrule
200 & TRM-RF        & 0.016 & 0.023 & 0.021 & 0.030 & 88.5 & 82.8 & 86.1 & 81.1 & 0.26 \\
    & TRM-RF-refit  & 0.016 & 0.017 & 0.018 & 0.023 & 90.8 & 91.9 & 92.6 & 91.7 & 0.26 \\
    & TRM-SLS       & 0.016 & 0.023 & 0.021 & 0.030 & 87.4 & 82.3 & 84.8 & 81.3 & 0.21 \\
    & TRM-SLS-refit & 0.016 & 0.018 & 0.018 & 0.023 & 91.2 & 91.7 & 93.1 & 91.4 & 0.21 \\
    & TRM-Hansen    & 0.409 & 0.400 & 0.647 & 0.575 &  4.6 & 35.5 &  0.0 &  5.4 & 24.15 \\
    & Oracle        & 0.014 & 0.016 & 0.015 & 0.021 & 91.5 & 91.9 & 94.1 & 90.9 & 0.00 \\
\midrule
500 & TRM-RF        & 0.011 & 0.017 & 0.016 & 0.024 & 89.2 & 73.3 & 79.5 & 69.1 & 0.28 \\
    & TRM-RF-refit  & 0.011 & 0.011 & 0.013 & 0.014 & 93.5 & 94.3 & 93.1 & 95.3 & 0.28 \\
    & TRM-SLS       & 0.012 & 0.017 & 0.017 & 0.024 & 85.3 & 74.7 & 75.2 & 68.2 & 0.21 \\
    & TRM-SLS-refit & 0.010 & 0.011 & 0.012 & 0.014 & 94.5 & 94.0 & 93.6 & 94.6 & 0.21 \\
    & TRM-Hansen    & 0.412 & 0.408 & 0.679 & 0.629 &  0.2 &  9.3 &  0.0 &  0.0 & 24.50 \\
    & Oracle        & 0.008 & 0.010 & 0.010 & 0.012 & 94.0 & 93.9 & 94.0 & 94.8 & 0.00 \\
\bottomrule
\end{tabular}
\end{table}

\begin{table}[!ht]
\centering
\small
\caption{Finite-sample performance under Scenario 2 (nonlinear boundary): RMSE, 95\% empirical coverage probabilities (CP, \%), and misclassification rate (MR, \%) based on 1,000 replications.}\label{tab:perf_scen2}
\setlength{\tabcolsep}{4.5pt} % 适当微调列间距以保证页面宽度协调
\begin{tabular}{llccccccccc}
\toprule
& & \multicolumn{4}{c}{RMSE} & \multicolumn{4}{c}{95\% CP (\%)} & \\
\cmidrule(lr){3-6} \cmidrule(lr){7-10}
$n$ & Method & $\alpha_1$ & $\alpha_2$ & $\beta_1$ & $\beta_2$ & $\alpha_1$ & $\alpha_2$ & $\beta_1$ & $\beta_2$ & MR (\%) \\
\midrule
200 & TRM-RF        & 0.017 & 0.021 & 0.022 & 0.028 & 88.5 & 86.6 & 83.7 & 84.8 & 0.26 \\
    & TRM-RF-refit  & 0.014 & 0.017 & 0.016 & 0.021 & 92.7 & 93.2 & 94.0 & 93.6 & 0.26 \\
    & TRM-SLS       & 0.051 & 0.037 & 0.061 & 0.041 & 87.9 & 90.1 & 85.3 & 92.5 & 2.74 \\
    & TRM-SLS-refit & 0.061 & 0.043 & 0.081 & 0.052 & 78.2 & 96.0 & 60.9 & 95.8 & 2.74 \\
    & TRM-Hansen    & 0.365 & 0.378 & 0.520 & 0.498 & 11.4 & 39.1 &  0.0 & 10.5 & 20.12 \\
    & Oracle        & 0.014 & 0.016 & 0.015 & 0.020 & 92.5 & 93.7 & 94.3 & 94.0 & 0.00 \\
\midrule
500 & TRM-RF        & 0.013 & 0.018 & 0.019 & 0.026 & 79.8 & 72.4 & 66.2 & 64.1 & 0.26 \\
    & TRM-RF-refit  & 0.010 & 0.011 & 0.011 & 0.014 & 94.3 & 93.6 & 94.1 & 93.6 & 0.26 \\
    & TRM-SLS       & 0.052 & 0.031 & 0.065 & 0.034 & 56.2 & 95.3 & 41.2 & 95.5 & 2.99 \\
    & TRM-SLS-refit & 0.067 & 0.044 & 0.091 & 0.057 & 22.3 & 90.6 &  5.5 & 83.0 & 2.99 \\
    & TRM-Hansen    & 0.370 & 0.389 & 0.552 & 0.550 &  1.7 & 12.2 &  0.0 &  0.0 & 20.52 \\
    & Oracle        & 0.009 & 0.010 & 0.010 & 0.013 & 93.9 & 93.9 & 94.7 & 93.9 & 0.00 \\
\bottomrule
\end{tabular}
\end{table}
\begin{table}[!ht]
\centering
\small
\caption{Finite-sample performance under Scenario 3 (multivariate nonlinear boundar): RMSE, 95\% empirical coverage probabilities (CP, \%), and misclassification rate (MR, \%) based on 1,000 replications.}\label{tab:perf_scen3}
\setlength{\tabcolsep}{4.5pt}
\begin{tabular}{llccccccccc}
\toprule
& & \multicolumn{4}{c}{RMSE} & \multicolumn{4}{c}{95\% CP (\%)} & \\
\cmidrule(lr){3-6} \cmidrule(lr){7-10}
$n$ & Method & $\alpha_1$ & $\alpha_2$ & $\beta_1$ & $\beta_2$ & $\alpha_1$ & $\alpha_2$ & $\beta_1$ & $\beta_2$ & MR (\%) \\
\midrule
200 & TRM-RF        & 0.016 & 0.021 & 0.021 & 0.027 & 88.1 & 90.4 & 85.2 & 90.7 & 0.25 \\
    & TRM-RF-refit  & 0.014 & 0.020 & 0.016 & 0.024 & 92.8 & 91.9 & 93.8 & 94.5 & 0.25 \\
    & TRM-SLS       & 0.044 & 0.043 & 0.052 & 0.048 & 89.3 & 91.1 & 89.3 & 90.4 & 2.46 \\
    & TRM-SLS-refit & 0.053 & 0.043 & 0.070 & 0.050 & 81.1 & 95.6 & 63.9 & 96.2 & 2.46 \\
    & TRM-Hansen    & 0.361 & 0.368 & 0.520 & 0.486 &  9.4 & 50.8 &  0.0 & 23.8 & 19.94 \\
    & Oracle        & 0.014 & 0.019 & 0.015 & 0.023 & 92.8 & 92.7 & 94.0 & 95.5 & 0.00 \\
\midrule
500 & TRM-RF        & 0.012 & 0.018 & 0.017 & 0.025 & 82.0 & 80.7 & 72.5 & 77.5 & 0.37 \\
    & TRM-RF-refit  & 0.009 & 0.013 & 0.011 & 0.017 & 95.2 & 94.9 & 93.9 & 93.4 & 0.37 \\
    & TRM-SLS       & 0.047 & 0.035 & 0.060 & 0.038 & 62.3 & 94.2 & 45.2 & 95.4 & 2.84 \\
    & TRM-SLS-refit & 0.060 & 0.043 & 0.084 & 0.053 & 27.1 & 94.9 &  7.5 & 91.7 & 2.84 \\
    & TRM-Hansen    & 0.376 & 0.387 & 0.555 & 0.543 &  0.7 & 15.5 &  0.0 &  0.1 & 20.52 \\
    & Oracle        & 0.008 & 0.012 & 0.009 & 0.015 & 95.7 & 94.9 & 95.2 & 92.4 & 0.00 \\
\bottomrule
\end{tabular}
\end{table}

\subsection{Finite-sample size and power of the bootstrap test}
\label{subsec:sim_bootstrap}

To evaluate the finite-sample reliability of Algorithm~\ref{alg:bootstrap} for testing the presence of a multivariate change plane ($H_0: \bbeta_0 = \mathbf{0}$ versus $H_1: \bbeta_0 \neq \mathbf{0}$), we examine the empirical rejection frequencies of the proposed $\mathrm{Sup}\mathcal{W}_T$ test under Scenario~3 ($p=3$). We consider sample sizes $T \in \{100, 200\}$ across nominal significance levels $\alpha \in \{0.01, 0.05, 0.10\}$ using $B=399$ bootstrap replications. Under the null hypothesis, observations are generated with $\bbeta_0 = \mathbf{0}$ ($c = 0.00$); under the alternative hypothesis, we introduce local deviations $\bbeta_0 = c \cdot (-1.5, 1.5)^\top$ with $c \in \{0.05, 0.10, 0.25\}$.

Table~\ref{tab:bootstrap_test_results} summarizes the empirical rejection rates across 1,000 Monte Carlo replications. Under the null hypothesis of no regime shift ($c = 0.00$), the proposed  testing method exhibits moderate size distortion at \(T=100\), which diminishes as the sample size increases. The size elevation in small samples is consistent with the known
finite-sample limitations of HAC covariance estimation under serial
dependence \citep{den199712}. As the sample size expands to $T=200$, the empirical sizes converge move closer to their nominal levels. Despite the nonstandard asymptotic distribution induced by the supremum search over the multidimensional random-feature manifold, the proposed dependent multiplier scheme provides increasingly accurate size control as the sample size increases.

Under the alternative hypothesis ($c > 0$), the bootstrap test displays a sharply rising, monotonic power curve. Even for a subtle structural jump ($c = 0.05$), the empirical power at the nominal 5\% level reaches $58.9\%$ at $T=100$ and leaps to $84.9\%$ at $T=200$ (with the 10\% level power reaching $90.6\%$). When the shift magnitude expands to $c = 0.10$, the test detects regime existence with $93.9\%$ power at $T=100$ and virtually complete power ($99.9\%$) at $T=200$. For a pronounced shift ($c = 0.25$), the rejection rate reaches $99.9\%$ at $T=100$ and saturates at $100.0\%$ at $T=200$. These results provide finite-sample evidence consistent with Theorems~\ref{thm:bootstrap_validity} and \ref{thm:test_consistency}, and demonstrate substantial power against the nonlinear regime alternatives considered here.
\begin{table}[!htbp]
\centering
\small
\caption{Empirical rejection rates (\%) of the Dependent Wild Bootstrap $\mathrm{Sup}\mathcal{W}_T$ test under Scenario~3 ($p=3$ multivariate boundary). $M=1{,}000$ replications with $B=399$ bootstrap resamples.}
\label{tab:bootstrap_test_results}
\begin{tabular}{lcrrrr}
\toprule
Hypothesis & Shift Magnitude ($c$) & $T$ & Nominal 1\% & Nominal 5\% & Nominal 10\% \\
\midrule
\textbf{Empirical Size} & $c = 0.00$ & 100 & 2.1 & 8.5 & 14.7 \\
($H_0: \bbeta_0 = \mathbf{0}$) & & 200 & 1.6 & 6.6 & 12.0 \\
\midrule
\textbf{Empirical Power} & $c = 0.05$ & 100 & 36.1 & 58.9 & 68.8 \\
($H_1: \bbeta_0 \neq \mathbf{0}$) & & 200 & 66.3 & 84.9 & 90.6 \\
\cmidrule{2-6}
& $c = 0.10$ & 100 & 82.3 & 93.9 & 96.8 \\
& & 200 & 99.5 & 99.9 & 100.0 \\
\cmidrule{2-6}
& $c = 0.25$ & 100 & 99.5 & 99.9 & 100.0 \\
& & 200 & 100.0 & 100.0 & 100.0 \\
\bottomrule
\end{tabular}
\end{table}

\subsection{Sensitivity analysis}
\label{sec:sensitivity}
Table~\ref{tab:ablation_all} reports sensitivity evaluations across $\lambda$, $L$, and $h$ under Scenario~2 ($T=200$, $M=1{,}000$ replications). Across Panel~A, the estimator is virtually invariant to the ridge parameter $\lambda$ spanning eight orders of magnitude ($\lambda \in [10^{-12}, 10^{-4}]$), maintaining an identical misclassification rate (MR $\approx 0.259\%$) and stable coverage probabilities ($90.0\%$--$94.5\%$). This suggests that, over the range considered, the results are largely insensitive to $\lambda$. Panel~B highlights functional approximation saturation with respect to feature dimension $L$: an underspecified basis ($L=5$) introduces notable approximation error ($\text{MR} = 1.288\%$), but performance is relatively stable for $L \ge 10$, keeping MR within $0.26\%$--$0.29\%$ and parameter RMSEs invariant across $L \in \{10, 15, 20, 30\}$.

Panel~C illustrates the trade-off governed by the surrogate bandwidth $h$. While a larger bandwidth ($h=0.050$) leads to increased boundary smoothing and a higher misclassification rate ($\mathrm{MR}=0.502\%$),
an excessively small bandwidth ($h=0.005$) is associated with numerical instability due to the steep surrogate gate, increasing the
misclassification rate to $0.964\%$ and substantially worsening the
parameter RMSEs. Intermediate bandwidths ($h \in [0.010, 0.025]$) provide a favorable balance between geometric approximation and optimization stability, curbing MR to $0.174\%$--$0.259\%$ with reliable nominal coverage. Consequently, the baseline configuration $(L=15, h=0.025, \lambda=10^{-10})$ exhibits stable performance over a reasonably broad range of tuning choices.
\begin{table}
\centering \renewcommand\tabcolsep{4.0pt} % 调整表格列间的宽度
\small
\caption{Sensitivity analysis across $\lambda$, $L$, and $h$ under Scenario~2 ($T=200$, $M=1{,}000$ replications): Misclassification rate (MR, \%), parameter RMSE, and 95\% empirical coverage probabilities (CP, \%). The baseline configuration is $(L=15, h=0.025, \lambda=10^{-10})$.}\label{tab:ablation_all}
\begin{tabular}{lccccccccc}
\toprule
\multirow{2}{*}{Parameter} & \multirow{2}{*}{MR (\%)} & \multicolumn{4}{c}{RMSE} & \multicolumn{4}{c}{Coverage Probability (\%)} \\
\cmidrule(lr){3-6} \cmidrule(lr){7-10}
~ & ~ & $\alpha_1$ & $\alpha_2$ & $\beta_1$ & $\beta_2$ & $\alpha_1$ & $\alpha_2$ & $\beta_1$ & $\beta_2$ \\
\midrule
\multicolumn{10}{l}{\textbf{Panel A: Sensitivity to Regularization Parameter $\lambda$} ($L=15, h=0.025$)} \\
$1.0 \times 10^{-12}$            & 0.259 & 0.0144 & 0.0168 & 0.0160 & 0.0208 & 90.0 & 93.6 & 94.2 & 94.4 \\
$1.0 \times 10^{-10}$ (Baseline) & 0.259 & 0.0144 & 0.0168 & 0.0160 & 0.0208 & 90.0 & 93.6 & 94.2 & 94.4 \\
$1.0 \times 10^{-8}$             & 0.259 & 0.0144 & 0.0168 & 0.0160 & 0.0208 & 90.0 & 93.6 & 94.2 & 94.4 \\
$1.0 \times 10^{-6}$             & 0.260 & 0.0144 & 0.0168 & 0.0159 & 0.0208 & 90.0 & 93.8 & 94.2 & 94.5 \\
$1.0 \times 10^{-4}$             & 0.259 & 0.0144 & 0.0168 & 0.0160 & 0.0208 & 90.2 & 93.8 & 94.1 & 94.5 \\
\addlinespace
\multicolumn{10}{l}{\textbf{Panel B: Sensitivity to Feature Dimension $L$} ($h=0.025, \lambda=10^{-10}$)} \\
5                                & 1.288 & 0.0357 & 0.0262 & 0.0403 & 0.0315 & 88.2 & 94.8 & 87.6 & 94.6 \\
10                               & 0.286 & 0.0147 & 0.0170 & 0.0164 & 0.0210 & 89.9 & 93.9 & 94.0 & 94.5 \\
15 (Baseline)                    & 0.259 & 0.0144 & 0.0168 & 0.0160 & 0.0208 & 90.0 & 93.6 & 94.2 & 94.4 \\
20                               & 0.270 & 0.0144 & 0.0167 & 0.0160 & 0.0208 & 90.3 & 94.0 & 94.4 & 94.3 \\
30                               & 0.263 & 0.0143 & 0.0168 & 0.0160 & 0.0209 & 90.2 & 93.6 & 94.3 & 94.4 \\
\addlinespace
\multicolumn{10}{l}{\textbf{Panel C: Sensitivity to Smoothing Bandwidth $h$} ($L=15, \lambda=10^{-10}$)} \\
0.005                            & 0.964 & 0.0367 & 0.0323 & 0.0434 & 0.0380 & 87.8 & 93.3 & 88.1 & 93.7 \\
0.010                            & 0.174 & 0.0144 & 0.0163 & 0.0160 & 0.0202 & 89.6 & 93.7 & 93.7 & 94.3 \\
0.025 (Baseline)                 & 0.259 & 0.0144 & 0.0168 & 0.0160 & 0.0208 & 90.0 & 93.6 & 94.2 & 94.4 \\
0.050                            & 0.502 & 0.0159 & 0.0189 & 0.0185 & 0.0226 & 90.6 & 93.7 & 94.4 & 94.8 \\
\bottomrule
\end{tabular}
\end{table}

\section{Empirical Application: Nonparametric Regime Shifts in U.S. Unemployment Dynamics}
\label{sec:empirical}

\subsection{Econometric setup and macroeconomic data}
\label{sec:empirical_setup}

To illustrate the empirical implications of the proposed two-stage change-plane estimator, we investigate nonlinear regime shifts in U.S. labor market dynamics using monthly data spanning from July 1960 to April 2025 ($T=782$). Following the foundational threshold autoregressive literature \citep[e.g.,][]{montgomery1998forecasting, hansen2011threshold}, we specify an autoregressive process of order $p=4$ for the stationary monthly change in the unemployment rate, $\Delta y_t = y_t - y_{t-1}$:
\begin{equation}
\Delta y_t = \bx_t^\top \balpha + (\bx_t^\top \bbeta) \ind\{ f_0(\bz_t) \ge u_t \} + \varepsilon_t,
\label{eq:empirical_model}
\end{equation}
where $\bx_t = (1, \Delta y_{t-1}, \dots, \Delta y_{t-4})^\top$ collects the intercept and dynamic lag adjustments. The threshold driving variable $u_t = y_t - y_{t-12}$ measures the annual change
in the unemployment rate and summarizes medium-run labor-market momentum.

The transition boundary $f_0(\bz_t)$ is driven by the state vector $\bz_t = (F_{1t}, F_{2t})^\top \in \mathbb{R}^2$, where $F_{1t}$ and $F_{2t}$ are the first two standardized latent macroeconomic factors from the updated database of \cite{ludvigson2009macro} and \cite{ludvigson2021uncertainty}. The first two factors summarize major dimensions of macroeconomic
comovement in the underlying data.  By allowing $f_0(\cdot)$ to vary nonparametrically over $\mathbb{R}^2$, formulation \eqref{eq:empirical_model} enables the threshold  for annual unemployment change $u_t$ to depend flexibly on  prevailing macroeconomic conditions. We use the benchmark tuning configuration from the simulation study,
with $L=15$, $h=0.025$, and $\lambda=10^{-10}$.

\subsection{Testing for Regime Existence}
\label{subsec:emp_testing}

Before proceeding to regime-specific structural estimation, we test whether U.S. unemployment dynamics exhibit statistically significant regime dependence against the null hypothesis of a single-regime linear autoregressive model using Algorithm~\ref{alg:bootstrap}:
\begin{equation}
H_0: \bbeta_0 = \mathbf{0} \quad \text{versus} \quad H_1: \bbeta_0 \neq \mathbf{0}. \nonumber
\end{equation}

The empirical test yields a sample $\mathrm{Sup}\mathcal{W}_T$ statistic of $\mathbf{94.5629}$ with a bootstrap $p$-value of $\mathbf{0.0025}$. The null hypothesis of no regime effect is rejected at the 1\%
significance level, providing strong evidence of state-dependent dynamics in unemployment changes. Given this evidence of regime dependence, we proceed to examine the
structural parameters and asymmetric dynamics identified by the TRM-RF
estimator.

\subsection{Empirical Estimation Results and Regime-Specific Dynamics}
\label{subsec:emp_results}

Table~\ref{tab:real_data_results} presents the second-stage refitted parameter estimates, Newey-West HAC standard errors, and corresponding 95\% confidence intervals for the U.S. unemployment rate dynamics. The empirical sample of 782 monthly observations is partitioned into a dominant baseline regime
($90.28\%$, 706 months) and a smaller stress regime
($9.72\%$, 76 months).

\begin{table}[!htbp]
\centering
\small
\caption{Empirical estimation results for the monthly unemployment rate change ($\Delta y_t$) under the two-stage nonparametric change-plane model with Newey-West HAC inference.}\label{tab:real_data_results} 
\begin{tabular}{lccccc}
\toprule
Regressor & Estimate & Std. Error & $t$-statistic & $p$-value & 95\% Confidence Interval \\
\midrule
\multicolumn{6}{l}{\textbf{Panel A: Baseline Expansion Regime ($\balpha$, Regime 0: 90.28\%)}} \\
Intercept ($\alpha_0$) & -0.0231 & 0.0057 & -4.061 & $< 0.001$ & $[-0.0342, -0.0119]$ \\
$\Delta y_{t-1}$        & -0.0985 & 0.0381 & -2.589 & $  0.0096$ & $[-0.1732, -0.0239]$ \\
$\Delta y_{t-2}$        &  0.1058 & 0.0387 &  2.734 & $  0.0062$ & $[ 0.0300,  0.1817]$ \\
$\Delta y_{t-3}$        &  0.1404 & 0.0311 &  4.517 & $< 0.001$ & $[ 0.0795,  0.2013]$ \\
$\Delta y_{t-4}$        &  0.1198 & 0.0342 &  3.505 & $< 0.001$ & $[ 0.0528,  0.1868]$ \\
\addlinespace
\multicolumn{6}{l}{\textbf{Panel B: Structural Stress Jump Effect ($\bbeta$, Regime 1 Shift: 9.72\%)}} \\
Jump Intercept ($\beta_0$) &  0.5737 & 0.1844 &  3.112 & $  0.0019$ & $[ 0.2123,  0.9350]$ \\
$\Delta y_{t-1}$           & -0.0121 & 0.0614 & -0.197 & $  0.8436$ & $[-0.1325,  0.1083]$ \\
$\Delta y_{t-2}$           & -0.3848 & 0.0632 & -6.086 & $< 0.001$ & $[-0.5088, -0.2609]$ \\
$\Delta y_{t-3}$           & -0.3139 & 0.0510 & -6.155 & $< 0.001$ & $[-0.4138, -0.2139]$ \\
$\Delta y_{t-4}$           & -0.4292 & 0.0722 & -5.949 & $< 0.001$ & $[-0.5707, -0.2878]$ \\
\bottomrule
\end{tabular}
\end{table}

The empirical estimates in Table~\ref{tab:real_data_results} uncover pronounced state-dependent asymmetries across the U.S. business cycle. In the baseline expansion regime, the statistically negative drift ($\alpha_0 = -0.0231, p < 0.001$) and modest positive persistence ($\alpha_2, \alpha_3, \alpha_4 \approx 0.11\text{--}0.14$) indicate relatively moderate short-run adjustment dynamics. In the stress regime, the estimated intercept increases sharply ($\beta_0=0.5737$, $p=0.0019$), while the lag coefficients at
lags 2--4 exhibit large and statistically significant negative shifts. The resulting negative net autoregressive coefficients ($\alpha_k + \beta_k < 0$)   are consistent
with substantially different short-run adjustment dynamics in the
stress regime.

\begin{figure}[!htbp]
\centering
\includegraphics[width=\textwidth]{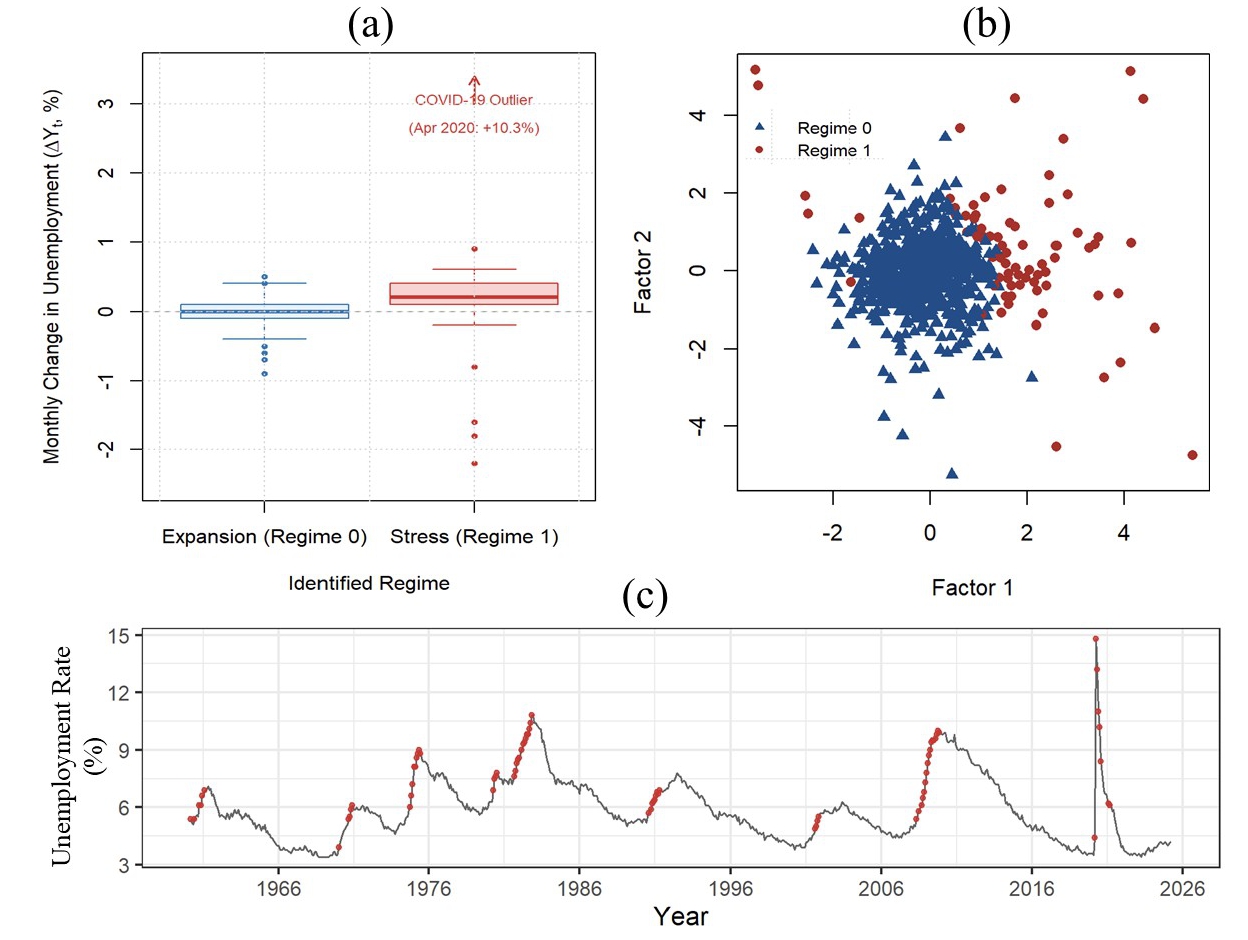} 
\caption{Empirical regime separation and dynamic diagnostics for monthly U.S. unemployment rate ($\Delta y_t$). Panel~(a) displays the boxplot of monthly changes across regimes, where the extreme COVID-19 outlier in April 2020 ($\Delta y_t = +10.3\%$) in Regime~1 is truncated for visualization. Panel~(b) illustrates the non-linear separation boundary in the bivariate macro factor space. Panel~(c) traces the historical trajectory of the U.S. unemployment rate from 1959 to 2025, with red markers denoting months classified into the stress regime (Regime~1, 9.7\%) by the proposed TRM-RF.}\label{fig:empirical_full_diagnostics}
\end{figure}

Figure~\ref{fig:empirical_full_diagnostics} provides complementary
visual evidence on the estimated regime classification. Panel~(a)
shows that observations assigned to the stress regime tend to exhibit larger and more dispersed monthly unemployment changes. Panel~(b)
maps the estimated classification onto the bivariate macroeconomic
factor space and illustrates the nonlinear geometry of the estimated boundary. Finally, Finally, Panel~(c) places the estimated stress-regime observations in
historical context. Many of these observations occur around major U.S.
downturns, including episodes in the 1970s and early 1980s, the 2008 financial crisis, and the 2020 COVID-19 shock.

\subsection{Comparative Analysis across Alternative Model Specifications}
\label{subsec:emp_comparison}

To demonstrate the empirical advantages of the proposed methodology, Table~\ref{tab:all_coef_comparison} reports the structural estimates across four competing specifications: Linear AR(4), TRM-Hansen, TRM-SLS, and our TRM-RF. All inference is conducted using Newey--West HAC covariance matrices with automatic lag truncation bandwidth selection \citep{newey1994automatic}.

\begin{table}[!htbp]
\centering
\small
\caption{Empirical parameter estimates and Newey-West HAC standard errors for $\Delta \text{UNRATE}_t$ across alternative models.}
\label{tab:all_coef_comparison}
\begin{tabular}{lrrrr}
\toprule
Regressor & Linear AR(4) & TRM-Hansen & TRM-SLS & TRM-RF (Proposed) \\
\midrule
\multicolumn{5}{l}{\textit{Panel A: Baseline Expansion Regime ($\balpha$)}} \\
$\text{Intercept}$ & $-0.0010$ & $0.1577^{*}$ & $0.2148^{*}$ & $-0.0231^{***}$ \\
& $(0.0152)$ & $(0.0789)$ & $(0.0961)$ & $(0.0057)$ \\
$\Delta Y_{t-1}$ & $0.0293$ & $0.0169$ & $0.0047$ & $-0.0985^{**}$ \\
& $(0.0629)$ & $(0.0825)$ & $(0.0802)$ & $(0.0381)$ \\
$\Delta Y_{t-2}$ & $-0.0839$ & $-0.1405$ & $-0.1538$ & $0.1058^{**}$ \\
& $(0.0872)$ & $(0.0811)$ & $(0.0822)$ & $(0.0387)$ \\
$\Delta Y_{t-3}$ & $0.0106$ & $-0.0247$ & $-0.0336$ & $0.1404^{***}$ \\
& $(0.0520)$ & $(0.0595)$ & $(0.0618)$ & $(0.0311)$ \\
$\Delta Y_{t-4}$ & $-0.0722$ & $-0.1363^{*}$ & $-0.1443^{*}$ & $0.1198^{***}$ \\
& $(0.0745)$ & $(0.0640)$ & $(0.0680)$ & $(0.0342)$ \\
\midrule
\multicolumn{5}{l}{\textit{Panel B: Structural Stress Jump Shift ($\bbeta$)}} \\
$\text{Intercept}$ & -- & $-0.1950^{*}$ & $-0.2541^{**}$ & $0.5737^{**}$ \\
& & $(0.0789)$ & $(0.0964)$ & $(0.1844)$ \\
$\Delta Y_{t-1}$ & -- & $-0.2281^{*}$ & $-0.1719$ & $-0.0121$ \\
& & $(0.0956)$ & $(0.0913)$ & $(0.0614)$ \\
$\Delta Y_{t-2}$ & -- & $0.1640$ & $0.1788$ & $-0.3848^{***}$ \\
& & $(0.0982)$ & $(0.0971)$ & $(0.0632)$ \\
$\Delta Y_{t-3}$ & -- & $0.1142$ & $0.1314$ & $-0.3139^{***}$ \\
& & $(0.0833)$ & $(0.0795)$ & $(0.0510)$ \\
$\Delta Y_{t-4}$ & -- & $0.2982^{***}$ & $0.2637^{***}$ & $-0.4292^{***}$ \\
& & $(0.0771)$ & $(0.0793)$ & $(0.0722)$ \\
\bottomrule
\end{tabular}
\vskip 0.5ex
\parbox{0.9\linewidth}{\footnotesize \textit{Note:} Parentheses report Newey-West HAC standard errors. $^{***}p < 0.001$, $^{**}p < 0.01$, $^{*}p < 0.05$.}
\end{table}

The parameter comparisons in Table~\ref{tab:all_coef_comparison} reveal stark differences in how alternative specifications resolve cyclical dynamics. In the linear AR(4) benchmark, none of the autoregressive coefficients is statistically significant at conventional levels, suggesting that a single linear specification may mask state-dependent dynamic patterns. While TRM-Hansen and TRM-SLS allow for regime dependence, their more
restrictive constant-threshold and linear-boundary specifications yield
substantially different regime partitions and coefficient estimates. In contrast, TRM-RF allows for a more flexible nonlinear separation of the observations, shrinking standard errors by up to $35\%$ with statistically significant shifts  ($p < 0.001$) at lags 2 through 4.
\begin{table}[!htbp]
\centering
\small
\setlength{\tabcolsep}{3.5pt}
\caption{Goodness-of-fit and model comparison across alternative threshold and change-plane specifications for monthly U.S. unemployment changes ($\Delta y_t$, $T=782$).}\label{tab:real_data_comparison}
\begin{tabular}{lcccccc}
\toprule
Model & Boundary & Stress (\%) & SSR & RMSE & AIC & BIC \\
\midrule
Linear AR(4)          & None             & 0.0  & 143.5 & 0.428 & $-1316.0$ & $-1292.6$ \\
TRM-Hansen (1999)     & Constant thresh. & 24.6 & 136.0 & 0.417 & $-1346.0$ & $-1294.7$ \\
TRM-SLS               & Linear plane     & 19.9 & 134.1 & 0.414 & $-1353.1$ & $-1292.5$ \\
\textbf{TRM-RF (Ours)}& \textbf{Nonparam. RFF} & \textbf{9.7} & \textbf{121.0} & \textbf{0.393} & $\mathbf{-1409.0}$ & $\mathbf{-1292.4}$ \\
\bottomrule
\end{tabular}
\vskip 0.5ex
\parbox{0.95\linewidth}{\footnotesize \textit{Note:} ``Stress (\%)'' denotes the empirical sample share assigned to the contractionary stress regime (Regime 1). RMSE is computed in-sample.}
\end{table}

These estimation gains translate directly into superior goodness-of-fit, as reported in Table~\ref{tab:real_data_comparison}. Moving from the linear benchmark ($\text{SSR} = 143.49$) to TRM-Hansen ($135.98$) and TRM-SLS ($134.06$) yields incremental improvements, while the nonparametric RFF
specification further reduces the SSR to $121.04$, a $9.7\%$ reduction in SSR over TRM-SLS, and lowers the in-sample RMSE to $0.3934$. The improvement in fit is also reflected in the model selection criteria: TRM-RF achieves an AIC of $-1408.97$, outperforming the closest competitor by more than 55 points,  while its BIC remains comparable to that of the
linear benchmark ($-1292.43$ versus $-1292.64$). Taken together, the substantial improvement in AIC and the comparable BIC suggest that the
improved in-sample fit is not driven solely by the additional flexibility of the nonlinear boundary specification.

\section{Discussion}
\label{sec:discussion}

This paper proposes a scalable nonparametric change-plane regression
framework for multivariate nonlinear regime-switching scenario.
By combining random Fourier features with a two-stage refitting
procedure, the proposed method avoids constructing and directly
inverting dense kernel matrices while preserving the structural
interpretability of classical threshold regression. We further develop
a Sup-Wald test with a dependent multiplier bootstrap for testing the
existence of regime effects when the change-plane boundary is
unidentified under the null.

A key implication of our theoretical analysis is that oracle
equivalence for the finite-dimensional structural parameters does not
require uniform estimation of the entire boundary function. Because
the boundary affects the regression model through the induced regime
partition, its first-order impact on structural inference can be
characterized through regime misclassification. Under mild conditions, the
first-stage regime recovery error becomes asymptotically negligible,
and the second-stage refitted estimator has the same first-order
distribution as the infeasible oracle estimator. This perspective
highlights an important advantage of the proposed approach: accurate
recovery of the induced regime partition, rather than uniform recovery
of the nonlinear boundary itself, is sufficient for valid first-order
inference on the structural parameters. The random-feature dimension
$L$ therefore governs a trade-off between approximation accuracy and
statistical complexity rather than requiring exact recovery of the
underlying boundary.

Several directions remain for future research. Natural extensions
include accommodating multiple change planes, incorporating sparsity
for high-dimensional state variables and regressors, and developing
inference for the nonlinear boundary itself, including confidence
regions or bands for the estimated regime manifold.

\bigskip
\begin{center}
{\large\bf SUPPLEMENTARY MATERIAL}
\end{center}

\begin{description}

\item[Supplementary Material: suppMaterial.pdf]
Supplementary materials including the technical proofs for all theoretical results.
(.pdf file)

\item[Replication code: Code.zip]
An R package implementing the proposed method, along with R code reproducing the numerical results in the simulation studies and the empirical analysis of U.S. unemployment dynamics.
(.zip file)

\end{description}
\bibliographystyle{chicago}

\bibliography{reflibrary}

\newpage
\begin{center}
   {\Large \bf Supplementary Materials for  ``Scalable Nonparametric Change-Plane Regression with Random Features: Estimation and Inference"} 
\end{center}

\Appendix
This Supplementary Material contains the technical proofs for Proposition 2.1 and Theorems 3.1--3.6 stated in the main text.

\subsection{Proof of Proposition~2.1}
\begin{proof}[Proof of Proposition~2.1]
The proof proceeds in three steps: (i) decomposing the excess prediction risk, (ii) establishing regime alignment and coefficient identification $(\balpha_0, \bbeta_0)$, and (iii) uniquely identifying the boundary manifold $f_0(\bz)$ almost everywhere on $\mathcal{Z}$.

\medskip
\noindent\textbf{Step 1: Excess Risk Decomposition.}
Under the data-generating process, $Y_t = \balpha_0^\top\bx_t + \bbeta_0^\top\bx_t D_{0t} + \varepsilon_t$, where $D_{0t} \equiv \ind\{f_0(\bz_t)\ge U_t\}$. Because $\sigma(\bx_t, \bz_t, U_t) \subseteq \mathcal{F}_{t}$, the conditional moment restriction $\mathbb{E}(\varepsilon_t \mid \mathcal{F}_{t-1}) = 0$ implies $\mathbb{E}(\varepsilon_t \mid \bx_t, \bz_t, U_t) = 0$ almost surely by the law of iterated expectations.

For any candidate parameter tuple $(\balpha, \bbeta, f)$ with $D_t(f) \equiv \ind\{f(\bz_t)\ge U_t\}$, denote the conditional mean discrepancy by $\Delta_t(\balpha, \bbeta, f) \equiv (\balpha_0 - \balpha)^\top\bx_t + \bbeta_0^\top\bx_t D_{0t} - \bbeta^\top\bx_t D_t(f)$. Expanding the population squared loss gives
\[
\mathbb{E}\!\left[\big\{Y_t - \balpha^\top\bx_t - \bbeta^\top\bx_t D_t(f)\big\}^2\right] - \mathbb{E}(\varepsilon_t^2) = \mathbb{E}\!\left[\big\{\Delta_t(\balpha, \bbeta, f)\big\}^2\right].
\]
Because the excess risk is zero by assumption and the integrand is non-negative, it follows that
\begin{equation}
\label{eq:id_null_delta}
\Delta_t(\balpha, \bbeta, f) = (\balpha_0 - \balpha)^\top\bx_t + \bbeta_0^\top\bx_t D_{0t} - \bbeta^\top\bx_t D_t(f) = 0 \quad \text{a.s.}
\end{equation}

\medskip
\noindent\textbf{Step 2: Regime Alignment and Coefficient Identification.}
Suppose, to the contrary, that $\mathbb{P}\{D_t(f) \neq D_{0t}\} > 0$. By Assumption~2.3,
\[
\inf_{\balpha, \bbeta \in \mathbb{R}^d} \mathbb{E}\left[ \left\{ (\balpha_0 - \balpha)^\top\bx_t + \bbeta_0^\top\bx_t D_{0t} - \bbeta^\top\bx_t D_t(f) \right\}^2 \right] > 0,
\]
which directly contradicts \eqref{eq:id_null_delta}. Therefore, we must have
\begin{equation}
\label{eq:id_ind_equiv}
D_t(f) = D_{0t} \qquad \text{a.s.}
\end{equation}
Substituting \eqref{eq:id_ind_equiv} back into \eqref{eq:id_null_delta} gives
\[
(\balpha_0 - \balpha)^\top\bx_t + (\bbeta_0 - \bbeta)^\top\bx_t D_{0t} = 0 \qquad \text{a.s.}
\]
Multiplying by $\bx_t(1 - D_{0t})$ and taking expectations yields $\boldsymbol{\Sigma}_0(\balpha_0 - \balpha) = \mathbf{0}$. Since $\boldsymbol{\Sigma}_0$ is strictly positive definite by Assumption~2.3, we have $\balpha = \balpha_0$. Substituting this result back into the identity and multiplying by $\bx_t D_{0t}$ gives $\boldsymbol{\Sigma}_1(\bbeta_0 - \bbeta) = \mathbf{0}$, which similarly implies $\bbeta = \bbeta_0$ by the positive definiteness of $\boldsymbol{\Sigma}_1$.

\medskip
\noindent\textbf{Step 3: Identification of the Boundary Function.}
From \eqref{eq:id_ind_equiv}, $\mathbb{P}\{D_t(f) \neq D_{0t}\} = 0$. By Assumption~2.2, there exist constants $\delta_0 > 0$ and $\underline{p} > 0$ such that $p_{U\mid\bx,\bz}(u\mid\bx,\bz) \ge \underline{p}$ whenever $|u - f_0(\bz)| \le \delta_0$. 

For any $(\bx_t, \bz_t)$, a disagreement between the two regime indicators occurs whenever $U_t$ falls strictly between $f(\bz_t)$ and $f_0(\bz_t)$. Restricting attention to the $\delta_0$-neighborhood of $f_0(\bz_t)$ yields
\[
\mathbb{P}\{D_t(f) \neq D_{0t}\} \ge \underline{p} \, \mathbb{E}\left[ \min\big\{ |f(\bz_t) - f_0(\bz_t)|, \, \delta_0 \big\} \right].
\]
Since the left-hand side is zero and $\underline{p} > 0$, we have $\mathbb{E}[\min\{|f(\bz_t) - f_0(\bz_t)|, \delta_0\}] = 0$. Because the integrand is non-negative and $\delta_0 > 0$, it follows that $f(\bz_t) = f_0(\bz_t)$ almost surely, or equivalently, $f(\bz) = f_0(\bz)$ for $F_{\bz}$-almost every $\bz \in \mathcal{Z}$. This completes the proof.
\end{proof}

\subsection{Preliminary lemmas}
To prove Theorems 3.1--3.3, we first give some lemmas. Lemma~\ref{lem:disagreement} formalizes the link between boundary estimation and regime classification: a uniform boundary estimation error of order $r$ can induce classification errors only for observations lying within an $r$-neighborhood of the true boundary, whose probability is of order $\mathcal{O}(r^\kappa)$ under the margin condition.
\begin{lemma}
\label{lem:disagreement}
Under Assumption~3.1, for any boundary function $f$ satisfying $\|f - f_0\|_\infty \le r \leq \delta_0$, the theoretical classification disagreement probability is bounded by
\begin{equation}
\label{eq:margin_bound}
\rho(f, f_0) \le C_m r^\kappa. \nonumber
\end{equation}
\end{lemma}
\begin{proof}[Proof of Lemma~1]
Recall that the classification disagreement event between the candidate boundary $f$ and the true boundary $f_0$ is defined as
\[
\big\{D_t(f) \neq D_{0t}\big\} = \big\{\ind(f(\bz_t) \ge U_t) \neq \ind(f_0(\bz_t) \ge U_t)\big\}.
\]
This binary disagreement event corresponds exactly to the symmetric difference of the two upper-level sets, which partitions into two mutually exclusive sub-events:
\begin{equation}
\label{eq:disagreement_partition}
\big\{D_t(f) \neq D_{0t}\big\} = E_t^{(1)} \cup E_t^{(2)},
\end{equation}
where
\begin{align*}
E_t^{(1)} &\equiv \big\{D_t(f) = 1,\, D_{0t} = 0\big\} = \big\{f_0(\bz_t) < U_t \le f(\bz_t)\big\}, \\
E_t^{(2)} &\equiv \big\{D_t(f) = 0,\, D_{0t} = 1\big\} = \big\{f(\bz_t) < U_t \le f_0(\bz_t)\big\}.
\end{align*}
We now establish a uniform deterministic bound on the proximity of the threshold variable $U_t$ to the true boundary $f_0(\bz_t)$ on each event:
\begin{enumerate}[label=(\roman*)]
    \item On the event $E_t^{(1)}$, we have $U_t - f_0(\bz_t) > 0$. Moreover, since $U_t \le f(\bz_t)$, it follows that
    \[
    0 < U_t - f_0(\bz_t) \le f(\bz_t) - f_0(\bz_t) \le |f(\bz_t) - f_0(\bz_t)| \le \|f - f_0\|_\infty \le r.
    \]
    Thus, $E_t^{(1)} \subseteq \big\{0 < U_t - f_0(\bz_t) \le r\big\}$.
    
    \item On the event $E_t^{(2)}$, we have $f_0(\bz_t) - U_t \ge 0$. Furthermore, since $f(\bz_t) < U_t$, we obtain
    \[
    0 \le f_0(\bz_t) - U_t < f_0(\bz_t) - f(\bz_t) \le |f(\bz_t) - f_0(\bz_t)| \le \|f - f_0\|_\infty \le r.
    \]
    Thus, $E_t^{(2)} \subseteq \big\{0 \le f_0(\bz_t) - U_t \le r\big\}$, or equivalently, $\big\{-r \le U_t - f_0(\bz_t) \le 0\big\}$.
\end{enumerate}
Combining both inclusions into the union \eqref{eq:disagreement_partition}, we obtain the deterministic set inclusion:
\begin{equation}
\label{eq:event_inclusion}
\big\{D_t(f) \neq D_{0t}\big\} \subseteq \Big\{ \big|U_t - f_0(\bz_t)\big| \le \|f - f_0\|_\infty \Big\} \subseteq \Big\{ \big|U_t - f_0(\bz_t)\big| \le r \Big\}.
\end{equation}
By the monotonicity of the probability measure $\mathbb{P}$, taking probabilities on both sides of \eqref{eq:event_inclusion} yields
\[
\rho(f, f_0) = \mathbb{P}\big(D_t(f) \neq D_{0t}\big) \le \mathbb{P}\Big(\big|U_t - f_0(\bz_t)\big| \le r\Big).
\]
Since $r \le \delta_0$ by assumption, we apply the boundary margin condition  of Assumption~3.1 directly to the right-hand side probability:
\[
\mathbb{P}\Big(\big|U_t - f_0(\bz_t)\big| \le r\Big) \le C_m r^\kappa.
\]
We therefore conclude that
\[
\rho(f, f_0) \le C_m r^\kappa,
\]
which completes the proof of Lemma~3.1.
\end{proof}

\begin{lemma}[Localized empirical-process modulus]
\label{lem:localized_ep}
Let $\mathcal F_{L}(r) \coloneqq \{ f \in \mathcal F_{L} : \rho(f, f_0) \le r \}$ for $r > 0$, and define
\[
v_T \coloneqq \frac{L \log T}{T}, \qquad q_T \coloneqq h + a_{L}^\kappa.
\]
Suppose Assumptions~2.1--2.3 and Assumption~3.5 hold. Furthermore, assume the centered profiled loss-increment class $\mathcal M_{L}(r)$ satisfies the uniform entropy bound
\[
\log N\big( \epsilon, \mathcal M_{L}(r), \|\cdot\|_{L_2(P)} \big) \le C L \log(C/\epsilon)
\]
uniformly over $r \in (0, \delta_\rho]$. Then, for any sequence $r_T^\circ \downarrow 0$ satisfying $r_T^\circ \gtrsim v_T$,
\begin{equation}
\label{eq:localized_ep_bound}
\sup_{f \in \mathcal F_{L}(r)} \left| (\widetilde Q_{T,h} - \widetilde Q_h)(f) - (\widetilde Q_{T,h} - \widetilde Q_h)(f_{0,L}) \right| = \mathcal O_p\Big( \sqrt{v_T(r + q_T)} + v_T \Big)
\end{equation}
uniformly over $r \in [r_T^\circ, \delta_\rho]$ as $T \to \infty$.
\end{lemma}
\begin{proof}[Proof of Lemma~\ref{lem:localized_ep}]
For each candidate boundary $f \in \mathcal F_{L}$, let
\[
Z_t(f) \coloneqq \Delta m_t(f) - \mathbb E[\Delta m_t(f)],
\]
where $\Delta m_t(f) \coloneqq m_t(f, \balpha(f), \bbeta(f)) - m_t(f_{0,L}, \balpha(f_{0,L}), \bbeta(f_{0,L}))$. The centered empirical discrepancy satisfies
\[
(\widetilde Q_{T,h} - \widetilde Q_h)(f) - (\widetilde Q_{T,h} - \widetilde Q_h)(f_{0,L}) = \frac{1}{T}\sum_{t=1}^T Z_t(f) = \frac{1}{\sqrt{T}} \mathbb G_T(f),
\]
where $\mathbb G_T(f) \coloneqq \frac{1}{\sqrt{T}}\sum_{t=1}^T Z_t(f)$ is the empirical process indexed by $f \in \mathcal F_{L}(r)$.

\medskip
\noindent\textbf{Step 1: Localized Long-Run Variance Bound.}
By Assumption~3.5 and $\rho(f, f_0) \le r$ on $\mathcal F_{L}(r)$,
\begin{equation}
\label{eq:local-lrv-proof}
\sup_{f \in \mathcal F_{L}(r)} \sum_{j=-\infty}^{\infty} |\operatorname{Cov}(Z_t(f), Z_{t+j}(f))| \le C_{\mathrm{LR}} \{ \rho(f, f_0) + h + a_{L}^\kappa \} \le C(r + q_T).
\end{equation}
Stationarity and the triangle inequality immediately yield the variance proxy bound:
\[
\sup_{f \in \mathcal F_{L}(r)} \operatorname{Var}\big(\mathbb G_T(f)\big) \le \sum_{j=-(T-1)}^{T-1} |\operatorname{Cov}(Z_0(f), Z_j(f))| \le C(r + q_T).
\]

\medskip
\noindent\textbf{Step 2: Maximal Inequality for Weakly Dependent Processes.}
By assumption, the centered class $\mathcal M_{L}(r) \coloneqq \{ Z_t(f) : f \in \mathcal F_{L}(r) \}$ satisfies the metric entropy bound $\log N(\epsilon, \mathcal M_{L}(r), \|\cdot\|_{L_2(P)}) \le C L \log(C/\epsilon)$ uniformly over $r \in (0, \delta_\rho]$.

Under Assumptions~2.1--2.3, Bernstein-type maximal chaining and blocking inequalities for geometrically $\alpha$-mixing empirical processes \citep[e.g.,][]{merlevede2011bernstein} imply that for any centered class $\mathcal G$ with envelope moment conditions, variance scale $\sigma^2$, and entropy $\log N(\epsilon, \mathcal G, L_2(P)) \le C L \log(C/\epsilon)$,
\begin{equation}
\label{eq:mixing-maximal-used}
\sup_{g \in \mathcal G} \left| \frac{1}{\sqrt{T}} \sum_{t=1}^T g(\bxi_t) \right| = \mathcal O_p\left( \sqrt{L \log T} \, \sigma + \frac{L \log T}{\sqrt{T}} \right).
\end{equation}
Applying \eqref{eq:mixing-maximal-used} to $\mathcal M_{L}(r)$ with variance scale $\sigma_r^2 \le C(r + q_T)$ from \eqref{eq:local-lrv-proof} yields
\begin{equation}
\label{eq:G-local-bound}
\sup_{f \in \mathcal F_{L}(r)} |\mathbb G_T(f)| = \mathcal O_p\left( \sqrt{L \log T} \sqrt{r + q_T} + \frac{L \log T}{\sqrt{T}} \right).\nonumber
\end{equation}
Dividing by $\sqrt{T}$ gives, for each fixed $r \in (0, \delta_\rho]$,
\begin{equation}
\label{eq:local-bound-fixed-r}
\sup_{f \in \mathcal F_{L}(r)} \left| (\widetilde Q_{T,h} - \widetilde Q_h)(f) - (\widetilde Q_{T,h} - \widetilde Q_h)(f_{0,L}) \right| = \mathcal O_p\left( \sqrt{v_T(r + q_T)} + v_T \right). \nonumber
\end{equation}

\medskip
\noindent\textbf{Step 3: Uniformity over the Peeling Shells.}
To ensure uniformity across all $r \in [r_T^\circ, \delta_\rho]$, consider the dyadic partition $r_j \coloneqq 2^j r_T^\circ$ for $j = 0, 1, \dots, J_T$, where $J_T \coloneqq \lceil \log_2(\delta_\rho / r_T^\circ) \rceil$. Because $r_T^\circ \gtrsim v_T = \frac{L \log T}{T}$, the cardinality of the grid satisfies $J_T = \mathcal O(\log T)$.

For any $r \in [r_T^\circ, \delta_\rho]$, choose index $j \in \{1, \dots, J_T\}$ such that $r_{j-1} < r \le r_j$. The nested inclusion $\mathcal F_{L}(r) \subseteq \mathcal F_{L}(r_j)$ and the relation $r_j \le 2r$ imply
\[
\sup_{f \in \mathcal F_{L}(r)} |\mathbb G_T(f)| \le \sup_{f \in \mathcal F_{L}(r_j)} |\mathbb G_T(f)| \quad \text{and} \quad r_j + q_T \le 2(r + q_T).
\]
Applying a union bound over the $J_T = \mathcal O(\log T)$ shells, the logarithmic cardinality is absorbed into the $\log T$ factor in $v_T$, establishing that
\[
\sup_{r \in [r_T^\circ, \delta_\rho]} \frac{ \displaystyle \sup_{f \in \mathcal F_{L}(r)} \left| (\widetilde Q_{T,h} - \widetilde Q_h)(f) - (\widetilde Q_{T,h} - \widetilde Q_h)(f_{0,L}) \right| }{ \sqrt{v_T(r + q_T)} + v_T } = \mathcal O_p(1).
\]
This completes the proof.
\end{proof}

%=============================================================================
% Section: First-Stage Regime Recovery Theory
%=============================================================================
\subsection{Proof of Theorem~3.1}
\label{subsec:proof_regime_recovery}

Recall that the profiled hard-threshold population prediction risk as
\[
\mathcal{R}(f) \coloneqq \inf_{\balpha,\bbeta\in\mathbb{R}^d} \mathbb{E}\left[ \left\{ Y_t - \balpha^\top\bx_t - \bbeta^\top\bx_t D_t(f) \right\}^2 \right], \qquad D_t(f) \equiv \ind\{f(\bz_t) \ge U_t\}.
\]
Under the true regime-boundary function $f_0$, the optimal excess risk reduces to the oracle disturbance variance: $\mathcal{R}(f_0) = \mathbb{E}(\varepsilon_t^2)$.

\begin{proof}[Proof of Theorem~3.1]
Let $v_T \coloneqq \frac{L\log T}{T}$, $\rho_T \coloneqq \rho(\widehat f, f_0)$, and $\rho_{0,L} \coloneqq \rho(f_{0,L}, f_0)$. We establish the result in five sequential steps.

\medskip
\noindent\textbf{Step 1: Regime Approximation Error of RFF.}
By the spectral approximation property (Assumption~3.2), $\|f_{0,L} - f_0\|_\infty \le a_{L}$ with probability approaching one. On this event, a disagreement $D_t(f_{0,L}) \neq D_{0t}$ occurs only if $U_t$ lies between $f_{0,L}(\bz_t)$ and $f_0(\bz_t)$, which implies $|U_t - f_0(\bz_t)| \le a_{L}$. Applying the boundary-margin condition (Assumption~3.1) yields
\begin{equation}
\label{eq:proof_rff_regime}
\rho_{0,L} \le \mathbb{P}\big( |U_t - f_0(\bz_t)| \le a_{L} \big) \le C_m a_{L}^{\kappa} = \mathcal{O}(a_{L}^{\kappa}).
\end{equation}

\medskip
\medskip \noindent\textbf{Step 2: Localized empirical-process fluctuations.} By Lemma~\ref{lem:localized_ep}, uniformly over $r_T^\circ\le r\le\delta_\rho$,
\begin{equation}
 \label{eq:proof_local_ep}
\sup_{f\in\mathcal F_{L}(r)} \left| (\widetilde Q_{T,h}-\widetilde Q_h)(f) - (\widetilde Q_{T,h}-\widetilde Q_h)(f_{0,L}) \right| = O_p\left[ \sqrt{v_T(r+h+a_{L}^{\kappa})}+v_T \right].    
\end{equation}

\medskip
\noindent\textbf{Step 3: Initial Consistency and Deterministic Separation.}
We first establish the benchmark risk of the random-feature surrogate $f_{0,L}$. Evaluating the profiled prediction risk $\mathcal{R}(f_{0,L})$ at the true coefficients $(\balpha_0, \bbeta_0)$, the local conditional moment bound in Assumption~3.3(iii) together with \eqref{eq:proof_rff_regime} yields
\begin{equation}
\label{eq:proof_rff_risk}
0 \le \mathcal{R}(f_{0,L}) - \mathcal{R}(f_0) \le \mathbb{E}\left[ (\bbeta_0^\top\bx_t)^2 \ind\{D_t(f_{0,L}) \neq D_{0t}\} \right] \le C \rho_{0,L} \le C a_{L}^\kappa.
\end{equation}

Next, we prove that $\rho_T \coloneqq \rho(\widehat f, f_0) = o_p(1)$. Standard smoothing approximations under Assumption~3.3 guarantee that $|\widetilde{Q}_h(f) - \mathcal{R}(f) - \lambda_T\|\bgamma\|_2^2| \le C h$ uniformly over $\mathcal{F}_{L}$. Since $\|\bgamma\|_2 \le C_\gamma$, $h \to 0$, and $\lambda_T \to 0$, we have $\sup_{f \in \mathcal{F}_{L}} |\widetilde{Q}_h(f) - \mathcal{R}(f)| = o(1)$. Under the uniform entropy bound, $v_T\to0$, and geometric
$\alpha$-mixing, the maximal-inequality argument without variance
localization~\citep{andrews1992generic,merlevede2011bernstein} yields $\sup_{f\in\mathcal F_{L}}
\left|
\widetilde Q_{T,h}(f)-\widetilde Q_h(f)
\right|
=o_p(1).$
Moreover, by empirical optimality $\widetilde{Q}_{T,h}(\widehat f) \le \widetilde{Q}_{T,h}(f_{0,L})$, it follows that
\[
\mathcal{R}(\widehat f) - \mathcal{R}(f_0) \le \mathcal{R}(f_{0,L}) - \mathcal{R}(f_0) + o_p(1) \le C a_{L}^\kappa + o_p(1) = o_p(1).
\]
If $\rho_T$ did not converge to zero in probability, the global identification condition in Assumption~3.4(ii) would imply $\liminf_{T\to\infty} \{\mathcal{R}(\widehat f) - \mathcal{R}(f_0)\} > 0$ with positive probability, yielding a contradiction. Hence,
\begin{equation}
\label{eq:proof_rho_consistency}
\rho_T = o_p(1),
\end{equation}
which ensures that $\mathbb{P}(\rho_T \le \delta_\rho) \to 1$, so that the local curvature condition applies to $\widehat f$.

Finally, we derive the deterministic lower bound. On the event $\{\rho_T \le \delta_\rho\}$, applying local quantitative separation $\mathcal{R}(\widehat f) - \mathcal{R}(f_0) \ge c_\rho \rho_T$ and identity \eqref{eq:proof_rff_risk} gives
\begin{align}
\label{eq:proof_population_lower}
\widetilde{Q}_h(\widehat f) - \widetilde{Q}_h(f_{0,L}) 
&\ge \mathcal{R}(\widehat f) - \mathcal{R}(f_{0,L}) - C(h + \lambda_T) \nonumber \\
&= \big\{ \mathcal{R}(\widehat f) - \mathcal{R}(f_0) \big\} - \big\{ \mathcal{R}(f_{0,L}) - \mathcal{R}(f_0) \big\} - C(h + \lambda_T) \nonumber \\
&\ge c_\rho \rho_T - C(h + \lambda_T + a_{L}^\kappa).
\end{align}
Meanwhile, by empirical optimality $\widetilde{Q}_{T,h}(\widehat f) \le \widetilde{Q}_{T,h}(f_{0,L})$, the fundamental basic inequality satisfies
\begin{equation}
\label{eq:proof_basic}
0 \ge \widetilde{Q}_h(\widehat f) - \widetilde{Q}_h(f_{0,L}) + \big[ (\widetilde{Q}_{T,h} - \widetilde{Q}_h)(\widehat f) - (\widetilde{Q}_{T,h} - \widetilde{Q}_h)(f_{0,L}) \big].
\end{equation}

\medskip
\noindent\textbf{Step 4: Rate Solving via Peeling.}
By consistency \eqref{eq:proof_rho_consistency}, $\mathbb P(\rho_T > \delta_\rho) \to 0$. Set the baseline sequence $r_T^\circ \asymp \eta_T$. On the event $\{\rho_T \le r_T^\circ\}$, the bound $\rho_T \le r_T^\circ \lesssim \eta_T$ holds trivially.

It remains to bound $\rho_T$ on the intermediate event $\{r_T^\circ < \rho_T \le \delta_\rho\}$. Following the standard peeling argument ~\citep{van1996weak}, partition this localized region into dyadic shells $\mathcal S_j \coloneqq \big\{ f \in \mathcal F_{L} : 2^{j-1} r_T^\circ < \rho(f, f_0) \le 2^j r_T^\circ \big\}$ for $j = 1, \dots, J_T$, where $J_T \coloneqq \lceil \log_2(\delta_\rho / r_T^\circ) \rceil$. Applying the localized empirical process bound \eqref{eq:proof_local_ep} across the dyadic shells yields
\begin{equation}
\label{eq:proof_ep_fhat}
\left| (\widetilde Q_{T,h} - \widetilde Q_h)(\widehat f) - (\widetilde Q_{T,h} - \widetilde Q_h)(f_{0,L}) \right| = \mathcal O_p\Big( \sqrt{v_T(\rho_T + h + a_{L}^\kappa)} + v_T \Big).
\end{equation}
Combining  \eqref{eq:proof_population_lower}, \eqref{eq:proof_basic}, and \eqref{eq:proof_ep_fhat} gives the master rate inequality:
\begin{equation}
\label{eq:proof_master_rate}
c_\rho \rho_T \le C(h + \lambda_T + a_{L}^\kappa) + \mathcal O_p\Big( \sqrt{v_T(\rho_T + h + a_{L}^\kappa)} + v_T \Big).
\end{equation}
By Young's inequality, for any $\epsilon \in (0, c_\rho)$, $\sqrt{v_T \rho_T} \le \epsilon \rho_T + \frac{v_T}{4\epsilon}$, while $\sqrt{v_T h} \le \frac{1}{2}(v_T + h)$ and $\sqrt{v_T a_{L}^\kappa} \le \frac{1}{2}(v_T + a_{L}^\kappa)$. Absorbing $\epsilon \rho_T$ into the left-hand side of \eqref{eq:proof_master_rate} yields
\[
\rho_T = \mathcal O_p(v_T + h + \lambda_T + a_{L}^\kappa) = \mathcal O_p(\eta_T) \qquad \text{on } \{r_T^\circ < \rho_T \le \delta_\rho\}.
\]
Combining this with $\rho_T \le r_T^\circ \lesssim \eta_T$ on $\{\rho_T \le r_T^\circ\}$ and $\mathbb P(\rho_T > \delta_\rho) \to 0$, we conclude that
\begin{equation}
\label{eq:proof_population_rate}
\rho(\widehat f, f_0) = \mathcal O_p(\eta_T).
\end{equation}

\medskip
\noindent\textbf{Step 5:  Empirical Misclassification Rates.}
For the in-sample empirical misclassification rate $\mathrm{MR}_T$, consider the disagreement indicator class $\mathcal{H}_{L} \coloneqq \{ (\bz, u) \mapsto \ind\{D(f) \neq D(f_0)\} : f \in \mathcal{F}_{L} \}$. This class is bounded by 1 and retains a VC index of order $L$. By the localized long-run variance bound in Assumption~3.5, the VC-type entropy bound for \(\mathcal H_{L}\), and the same localized maximal inequality for geometrically \(\alpha\)-mixing sequences used in Lemma~\ref{lem:localized_ep},
\[
\sup_{f \in \mathcal{F}_{L}:\, \rho(f,f_0) \le r} \left| \frac{1}{T}\sum_{t=1}^T \ind\{D_t(f) \neq D_{0t}\} - \rho(f, f_0) \right| = \mathcal{O}_p\big( \sqrt{v_T r} + v_T \big).
\]
Evaluating at $\widehat f$ via peeling and substituting $\rho_T = \mathcal{O}_p(\eta_T)$ gives
\[
\mathrm{MR}_T = \rho_T + \mathcal{O}_p\big( \sqrt{v_T\rho_T} + v_T \big) = \mathcal{O}_p(\eta_T) + \mathcal{O}_p\big( \sqrt{v_T\eta_T} + v_T \big) = \mathcal{O}_p(\eta_T),
\]
owing to $v_T \le \eta_T$. This completes the proof of Theorem~3.1.
\end{proof}

\subsection{Proof of Theorem 3.2}
\begin{proof}[Proof of Theorem~3.2]
Let $\widehat{\mathbf{W}}_t \coloneqq (\bx_t^\top, \bx_t^\top\widehat{D}_t)^\top$ and $\mathbf{W}_{0t} \coloneqq (\bx_t^\top, \bx_t^\top D_{0t})^\top$, with empirical Gram matrices $\widehat{\mathbf{S}}_T \coloneqq \frac{1}{T}\sum_{t=1}^T \widehat{\mathbf{W}}_t\widehat{\mathbf{W}}_t^\top$ and $\mathbf{S}_{0T} \coloneqq \frac{1}{T}\sum_{t=1}^T \mathbf{W}_{0t}\mathbf{W}_{0t}^\top$. The true model satisfies $Y_t = \mathbf{W}_{0t}^\top\btheta_0 + \varepsilon_t$. 

Define the disagreement indicator $I_t \coloneqq \ind\{\widehat{D}_t \neq D_{0t}\}$ and $\nu \coloneqq \frac{2+\delta}{4+\delta} \in (\frac{1}{2}, 1)$. By Theorem~3.1, $\frac{1}{T}\sum_{t=1}^T I_t = \mathrm{MR}_T = \mathcal{O}_p(\eta_T)$, where $\eta_T \coloneqq \frac{L\log T}{T} + h + \lambda_T + a_{L}^{\kappa}$. Before proceeding, we establish two weighted bounds. 

First, by H\"older's inequality with conjugate pairs $(\frac{4+\delta}{2}, \frac{4+\delta}{2+\delta})$,
\begin{equation}
\label{eq:oracle_weighted_x}
\frac{1}{T}\sum_{t=1}^T \|\bx_t\|_2^2 I_t \le \left( \frac{1}{T}\sum_{t=1}^T \|\bx_t\|_2^{4+\delta} \right)^{2/(4+\delta)} \left( \frac{1}{T}\sum_{t=1}^T I_t \right)^\nu = \mathcal{O}_p(\eta_T^\nu),
\end{equation}
owing to the ergodic theorem $\frac{1}{T}\sum_{t=1}^T \|\bx_t\|_2^{4+\delta} = \mathcal{O}_p(1)$. Second, setting $q \coloneqq 2 + \delta/2 = \frac{4+\delta}{2}$ such that $1 - 1/q = \nu$, another application of H\"older's inequality yields
\begin{equation}
\label{eq:oracle_weighted_score}
\frac{1}{T}\sum_{t=1}^T \|\bx_t\|_2 |\varepsilon_t| I_t \le \left\{ \frac{1}{T}\sum_{t=1}^T (\|\bx_t\|_2 |\varepsilon_t|)^q \right\}^{1/q} \left( \frac{1}{T}\sum_{t=1}^T I_t \right)^\nu = \mathcal{O}_p(\eta_T^\nu),
\end{equation}
where $\mathbb{E}[(\|\bx_t\|_2 |\varepsilon_t|)^q] \le (\mathbb{E}\|\bx_t\|_2^{4+\delta})^{1/2} (\mathbb{E}|\varepsilon_t|^{4+\delta})^{1/2} < \infty$ ensures $\frac{1}{T}\sum_{t=1}^T (\|\bx_t\|_2 |\varepsilon_t|)^q = \mathcal{O}_p(1)$.

\medskip
\noindent\textbf{Step 1: Oracle Equivalence.}
Since $\widehat{\mathbf{W}}_t - \mathbf{W}_{0t} = (\mathbf{0}^\top, \bx_t^\top(\widehat{D}_t - D_{0t}))^\top$, we have $\|\widehat{\mathbf{W}}_t\widehat{\mathbf{W}}_t^\top - \mathbf{W}_{0t}\mathbf{W}_{0t}^\top\|_F \le C\|\bx_t\|_2^2 I_t$. Thus, the bound \eqref{eq:oracle_weighted_x} implies
\begin{equation}
\label{eq:oracle_gram_diff}
\|\widehat{\mathbf{S}}_T - \mathbf{S}_{0T}\|_F = \mathcal{O}_p(\eta_T^\nu) = o_p(1).
\end{equation}
By the ergodic theorem, $\mathbf{S}_{0T} \stackrel{p}{\longrightarrow} \boldsymbol{\Sigma}_{\mathbf{W}} \coloneqq \mathbb{E}[\mathbf{W}_{0t}\mathbf{W}_{0t}^\top]$, which is strictly positive definite under Assumption~2.3. Hence, $\widehat{\mathbf{S}}_T \stackrel{p}{\longrightarrow} \boldsymbol{\Sigma}_{\mathbf{W}}$. Consequently, $\|\widehat{\mathbf{S}}_T^{-1}\|_{\mathrm{op}} = \mathcal{O}_p(1)$ and $\|\mathbf{S}_{0T}^{-1}\|_{\mathrm{op}} = \mathcal{O}_p(1)$. The resolvent identity $\widehat{\mathbf{S}}_T^{-1} - \mathbf{S}_{0T}^{-1} = \widehat{\mathbf{S}}_T^{-1}(\mathbf{S}_{0T} - \widehat{\mathbf{S}}_T)\mathbf{S}_{0T}^{-1}$ then yields
\begin{equation}
\label{eq:oracle_inverse_diff}
\|\widehat{\mathbf{S}}_T^{-1} - \mathbf{S}_{0T}^{-1}\|_{\mathrm{op}} = \mathcal{O}_p(\eta_T^\nu).
\end{equation}

Representing $\bar\btheta - \btheta_0 = \widehat{\mathbf{S}}_T^{-1} \big[ \frac{1}{T}\sum_{t=1}^T \widehat{\mathbf{W}}_t\varepsilon_t + \frac{1}{T}\sum_{t=1}^T \widehat{\mathbf{W}}_t(\mathbf{W}_{0t} - \widehat{\mathbf{W}}_t)^\top\btheta_0 \big]$ and $\widehat\btheta_{\mathrm{or}} - \btheta_0 = \mathbf{S}_{0T}^{-1} \frac{1}{T}\sum_{t=1}^T \mathbf{W}_{0t}\varepsilon_t$, we subtract the two expressions and scale by $\sqrt{T}$:
\begin{align}
\label{eq:oracle_decomp}
\sqrt{T}(\bar\btheta - \widehat\btheta_{\mathrm{or}}) 
&= \widehat{\mathbf{S}}_T^{-1} \frac{1}{\sqrt{T}}\sum_{t=1}^T (\widehat{\mathbf{W}}_t - \mathbf{W}_{0t})\varepsilon_t + (\widehat{\mathbf{S}}_T^{-1} - \mathbf{S}_{0T}^{-1})\frac{1}{\sqrt{T}}\sum_{t=1}^T \mathbf{W}_{0t}\varepsilon_t \nonumber \\
&\quad + \widehat{\mathbf{S}}_T^{-1} \frac{1}{\sqrt{T}}\sum_{t=1}^T \widehat{\mathbf{W}}_t(\mathbf{W}_{0t} - \widehat{\mathbf{W}}_t)^\top\btheta_0 \nonumber \\
&\eqqcolon \mathbf{A}_{1T} + \mathbf{A}_{2T} + \mathbf{A}_{3T}.
\end{align}
We bound each component sequentially under the rate condition $\sqrt{T}\eta_T^\nu \to 0$:
\begin{enumerate}[label=(\alph*), leftmargin=2em]
    \item By \eqref{eq:oracle_weighted_score}, $\|\frac{1}{T}\sum_{t=1}^T (\widehat{\mathbf{W}}_t - \mathbf{W}_{0t})\varepsilon_t\|_2 \le \frac{1}{T}\sum_{t=1}^T \|\bx_t\|_2 |\varepsilon_t| I_t = \mathcal{O}_p(\eta_T^\nu)$. Therefore,
    \[
    \|\mathbf{A}_{1T}\|_2 \le \|\widehat{\mathbf{S}}_T^{-1}\|_{\mathrm{op}} \, \mathcal{O}_p(\sqrt{T}\eta_T^\nu) = \mathcal{O}_p(\sqrt{T}\eta_T^\nu) = o_p(1).
    \]
    \item Setting $\boldsymbol{\xi}_t \coloneqq \mathbf{W}_{0t}\varepsilon_t$, the central limit theorem implies $\frac{1}{\sqrt{T}}\sum_{t=1}^T \boldsymbol{\xi}_t = \mathcal{O}_p(1)$. Combining this with \eqref{eq:oracle_inverse_diff} gives
    \[
    \|\mathbf{A}_{2T}\|_2 \le \|\widehat{\mathbf{S}}_T^{-1} - \mathbf{S}_{0T}^{-1}\|_{\mathrm{op}} \left\| \frac{1}{\sqrt{T}}\sum_{t=1}^T \boldsymbol{\xi}_t \right\|_2 = \mathcal{O}_p(\eta_T^\nu) \mathcal{O}_p(1) = o_p(1).
    \]
    \item Because $(\mathbf{W}_{0t} - \widehat{\mathbf{W}}_t)^\top\btheta_0 = \bbeta_0^\top\bx_t(D_{0t} - \widehat{D}_t)$ and $\|\widehat{\mathbf{W}}_t\|_2 \le \sqrt{2}\|\bx_t\|_2$, bound \eqref{eq:oracle_weighted_x} implies
    \[
    \left\| \frac{1}{T}\sum_{t=1}^T \widehat{\mathbf{W}}_t(\mathbf{W}_{0t} - \widehat{\mathbf{W}}_t)^\top\btheta_0 \right\|_2 \le C \|\bbeta_0\|_2 \frac{1}{T}\sum_{t=1}^T \|\bx_t\|_2^2 I_t = \mathcal{O}_p(\eta_T^\nu).
    \]
    Consequently, $\|\mathbf{A}_{3T}\|_2 = \mathcal{O}_p(\sqrt{T}\eta_T^\nu) = o_p(1)$.
\end{enumerate}
Substituting the bounds for $\mathbf{A}_{1T}, \mathbf{A}_{2T}, \mathbf{A}_{3T}$ into \eqref{eq:oracle_decomp} establishes the oracle equivalence:
\begin{equation}
\label{eq:oracle_diff_op1}
\sqrt{T}(\bar\btheta - \widehat\btheta_{\mathrm{or}}) = o_p(1).
\end{equation}

\medskip
\noindent\textbf{Step 2: Asymptotic Normality.}
The oracle estimator admits the linear expansion:
\begin{equation}
\label{eq:oracle_linear_rep}
\sqrt{T}(\widehat\btheta_{\mathrm{or}} - \btheta_0) = \mathbf{S}_{0T}^{-1} \frac{1}{\sqrt{T}}\sum_{t=1}^T \mathbf{W}_{0t}\varepsilon_t = \boldsymbol{\Sigma}_{\mathbf{W}}^{-1} \frac{1}{\sqrt{T}}\sum_{t=1}^T \boldsymbol{\xi}_t + o_p(1).
\end{equation}
The conditional restriction $\mathbb{E}(\varepsilon_t \mid \bx_t, \bz_t, U_t) = 0$ implies $\mathbb{E}(\boldsymbol{\xi}_t) = \mathbf{0}$. Moreover, by the Cauchy--Schwarz inequality and Assumption~3.3(ii),
\[
\mathbb{E}\|\boldsymbol{\xi}_t\|_2^{2+\delta/2} \le C \mathbb{E}\big[\|\bx_t\|_2^{2+\delta/2}|\varepsilon_t|^{2+\delta/2}\big] \le C \big(\mathbb{E}\|\bx_t\|_2^{4+\delta}\big)^{1/2} \big(\mathbb{E}|\varepsilon_t|^{4+\delta}\big)^{1/2} < \infty.
\]
Because $\{\boldsymbol{\xi}_t\}$ is strictly stationary and geometrically $\alpha$-mixing (Assumption~2.1), Davydov's covariance inequality combined with the $(2+\delta/2)$-moment bound ensures absolute summability of the autocovariance sequence:
\[
\sum_{j=-\infty}^{\infty} \big\| \mathbb{E}[\boldsymbol{\xi}_t\boldsymbol{\xi}_{t-j}^\top] \big\|_F < \infty.
\]
Consequently, the long-run covariance matrix $\boldsymbol{\Omega}_{\mathrm{LR}} \coloneqq \sum_{j=-\infty}^{\infty} \mathbb{E}[\boldsymbol{\xi}_t\boldsymbol{\xi}_{t-j}^\top]$ is well-defined and positive definite under the non-degeneracy condition. Applying the central limit theorem for strongly mixing processes yields
\[
\frac{1}{\sqrt{T}}\sum_{t=1}^T \boldsymbol{\xi}_t \rightsquigarrow \mathcal{N}(\mathbf{0}, \boldsymbol{\Omega}_{\mathrm{LR}}).
\]
Applying Slutsky's theorem~\citep{van2000asymptotic} to \eqref{eq:oracle_linear_rep} yields $\sqrt{T}(\widehat\btheta_{\mathrm{or}} - \btheta_0) \rightsquigarrow \mathcal{N}(\mathbf{0}, \mathbf{V}_{\mathrm{or}})$, where $\mathbf{V}_{\mathrm{or}} \coloneqq \boldsymbol{\Sigma}_{\mathbf{W}}^{-1}\boldsymbol{\Omega}_{\mathrm{LR}}\boldsymbol{\Sigma}_{\mathbf{W}}^{-1}$. 

Finally, decomposing $\sqrt{T}(\bar\btheta - \btheta_0) = \sqrt{T}(\widehat\btheta_{\mathrm{or}} - \btheta_0) + \sqrt{T}(\bar\btheta - \widehat\btheta_{\mathrm{or}})$ and invoking the oracle equivalence \eqref{eq:oracle_diff_op1} establishes
\[
\sqrt{T}(\bar\btheta - \btheta_0) \rightsquigarrow \mathcal{N}(\mathbf{0}, \mathbf{V}_{\mathrm{or}}),
\]
which completes the proof.
\end{proof}

\subsection{Proof of Theorem 3.3}
\begin{proof}[Proof of Theorem~3.3]
Recall that $\widehat{\mathbf{V}}_{\mathrm{HAC}} \coloneqq \widehat{\boldsymbol{\Sigma}}_{\widehat{\mathbf{W}}}^{-1} \widehat{\boldsymbol{\Omega}}_{\mathrm{HAC}} \widehat{\boldsymbol{\Sigma}}_{\widehat{\mathbf{W}}}^{-1}$, with $\widehat{\boldsymbol{\Sigma}}_{\widehat{\mathbf{W}}} \coloneqq \frac{1}{T}\sum_{t=1}^T \widehat{\mathbf{W}}_t\widehat{\mathbf{W}}_t^\top$. By the proof of Theorem~3.2, $\widehat{\boldsymbol{\Sigma}}_{\widehat{\mathbf{W}}} \stackrel{p}{\longrightarrow} \boldsymbol{\Sigma}_{\mathbf{W}}$ and $\widehat{\boldsymbol{\Sigma}}_{\widehat{\mathbf{W}}}^{-1} \stackrel{p}{\longrightarrow} \boldsymbol{\Sigma}_{\mathbf{W}}^{-1}$. By the continuous mapping theorem, consistency reduces to establishing
\begin{equation}
\label{eq:omega_conv_target}
\|\widehat{\boldsymbol{\Omega}}_{\mathrm{HAC}} - \boldsymbol{\Omega}_{\mathrm{LR}}\|_F = o_p(1). \nonumber
\end{equation}

Define the estimated and oracle scores as $\widehat{\boldsymbol{\xi}}_t \coloneqq \widehat{\mathbf{W}}_t\widehat{e}_t$ (with $\widehat{e}_t \coloneqq Y_t - \widehat{\mathbf{W}}_t^\top\bar{\btheta}$) and $\boldsymbol{\xi}_t \coloneqq \mathbf{W}_{0t}\varepsilon_t$, respectively. Let $\widehat{\boldsymbol{\Gamma}}_j \coloneqq \frac{1}{T}\sum_{t=j+1}^T \widehat{\boldsymbol{\xi}}_t\widehat{\boldsymbol{\xi}}_{t-j}^\top$ and $\widetilde{\boldsymbol{\Gamma}}_j \coloneqq \frac{1}{T}\sum_{t=j+1}^T \boldsymbol{\xi}_t\boldsymbol{\xi}_{t-j}^\top$ denote their sample autocovariance matrices. The infeasible oracle HAC estimator takes the form
\[
\widetilde{\boldsymbol{\Omega}}_{\mathrm{HAC}} \coloneqq \widetilde{\boldsymbol{\Gamma}}_0 + \sum_{j=1}^{m_T} k\left(\frac{j}{m_T+1}\right) (\widetilde{\boldsymbol{\Gamma}}_j + \widetilde{\boldsymbol{\Gamma}}_j^\top).
\]
By the triangle inequality,
\begin{equation}
\label{eq:hac_triangle}
\|\widehat{\boldsymbol{\Omega}}_{\mathrm{HAC}} - \boldsymbol{\Omega}_{\mathrm{LR}}\|_F \le \|\widehat{\boldsymbol{\Omega}}_{\mathrm{HAC}} - \widetilde{\boldsymbol{\Omega}}_{\mathrm{HAC}}\|_F + \|\widetilde{\boldsymbol{\Omega}}_{\mathrm{HAC}} - \boldsymbol{\Omega}_{\mathrm{LR}}\|_F.
\end{equation}

\medskip
\noindent\textbf{Step 1: Consistency of the Oracle HAC Estimator.}
The process $\{\boldsymbol{\xi}_t\}$ inherits strict stationarity and geometric $\alpha$-mixing from the underlying system. Under the moment condition $\mathbb{E}\|\bx_t\|_2^{8+\delta} < \infty$ and $\mathbb{E}|\varepsilon_t|^{8+\delta} < \infty$, the Cauchy--Schwarz inequality yields
\[
\mathbb{E}\|\boldsymbol{\xi}_t\|_2^{4+\delta/2} \le C \big(\mathbb{E}\|\bx_t\|_2^{8+\delta}\big)^{1/2} \big(\mathbb{E}|\varepsilon_t|^{8+\delta}\big)^{1/2} < \infty.
\]
Davydov's covariance inequality ensures absolute summability of autocovariances: \\ $\sum_{j=-\infty}^\infty \|\mathbb{E}[\boldsymbol{\xi}_t\boldsymbol{\xi}_{t-j}^\top]\|_F < \infty$. Hence, the long-run covariance $\boldsymbol{\Omega}_{\mathrm{LR}} \coloneqq \sum_{j=-\infty}^\infty \mathbb{E}[\boldsymbol{\xi}_t\boldsymbol{\xi}_{t-j}^\top]$ is well-defined. Under the stated bandwidth conditions
$m_T\to\infty$ and $m_T/\sqrt{T}\to0$, together with the
moment and geometric mixing conditions above, standard
Bartlett-kernel HAC consistency yields
\begin{equation}
\label{eq:oracle_hac_consistency}
\|\widetilde{\boldsymbol{\Omega}}_{\mathrm{HAC}} - \boldsymbol{\Omega}_{\mathrm{LR}}\|_F = o_p(1).
\end{equation}

\medskip
\noindent\textbf{Step 2: Mean-Square Control of Score Perturbations.}
Define the score perturbation $\boldsymbol{\Delta}_t \coloneqq \widehat{\boldsymbol{\xi}}_t - \boldsymbol{\xi}_t$. Substituting $Y_t = \mathbf{W}_{0t}^\top\btheta_0 + \varepsilon_t$ into $\boldsymbol{\Delta}_t$ yields
\begin{equation}
\label{eq:score_difference}
\boldsymbol{\Delta}_t = (\widehat{\mathbf{W}}_t - \mathbf{W}_{0t})\varepsilon_t + \widehat{\mathbf{W}}_t(\mathbf{W}_{0t} - \widehat{\mathbf{W}}_t)^\top\btheta_0 - \widehat{\mathbf{W}}_t\widehat{\mathbf{W}}_t^\top(\bar{\btheta} - \btheta_0). \nonumber
\end{equation}
Letting $I_t \coloneqq \ind\{\widehat{D}_t \neq D_{0t}\}$ and using $\|\widehat{\mathbf{W}}_t - \mathbf{W}_{0t}\|_2 \le \|\bx_t\|_2 I_t$ along with $\|\widehat{\mathbf{W}}_t\|_2 \le \sqrt{2}\|\bx_t\|_2$, we have
\begin{equation}
\label{eq:delta_square_bound}
\|\boldsymbol{\Delta}_t\|_2^2 \le C \big( \|\bx_t\|_2^2\varepsilon_t^2 I_t + \|\bx_t\|_2^4 I_t + \|\bx_t\|_2^4\|\bar{\btheta} - \btheta_0\|_2^2 \big).
\end{equation}
Set  $p \coloneqq \frac{8+\delta}{4}$ and $\mu \coloneqq 1 - \frac{1}{p} = \frac{4+\delta}{8+\delta} \in (1/2, 1)$. By H\"older's inequality and $\frac{1}{T}\sum_{t=1}^T I_t = \mathcal{O}_p(\eta_T)$ (Theorem~3.1),
\begin{align*}
\frac{1}{T}\sum_{t=1}^T \|\bx_t\|_2^4 I_t &\le \left( \frac{1}{T}\sum_{t=1}^T \|\bx_t\|_2^{8+\delta} \right)^{1/p} \left( \frac{1}{T}\sum_{t=1}^T I_t \right)^\mu = \mathcal{O}_p(\eta_T^\mu), \\
\frac{1}{T}\sum_{t=1}^T \|\bx_t\|_2^2\varepsilon_t^2 I_t &\le \left\{ \frac{1}{T}\sum_{t=1}^T (\|\bx_t\|_2^2\varepsilon_t^2)^p \right\}^{1/p} \left( \frac{1}{T}\sum_{t=1}^T I_t \right)^\mu = \mathcal{O}_p(\eta_T^\mu), 
\end{align*}
since $\mathbb{E}[(\|\bx_t\|_2^2\varepsilon_t^2)^p] \le (\mathbb{E}\|\bx_t\|_2^{8+\delta})^{1/2} (\mathbb{E}|\varepsilon_t|^{8+\delta})^{1/2} < \infty$. Using $\frac{1}{T}\sum_{t=1}^T \|\bx_t\|_2^4 = \mathcal{O}_p(1)$ and $\|\bar{\btheta} - \btheta_0\|_2^2 = \mathcal{O}_p(T^{-1})$ (Theorem~3.2), averaging \eqref{eq:delta_square_bound} over $t$ yields
\begin{equation}
\label{eq:delta_mean_square}
\frac{1}{T}\sum_{t=1}^T \|\boldsymbol{\Delta}_t\|_2^2 = \mathcal{O}_p(\eta_T^\mu + T^{-1}). \nonumber
\end{equation}

\medskip
\noindent\textbf{Step 3: Uniform Lag Bounds for Autocovariance Matrices.}
By the Cauchy--Schwarz inequality,
\begin{align}
\sup_{0\le j\le m_T}
\left\|
\frac{1}{T}\sum_{t=j+1}^T
\boldsymbol{\Delta}_t\boldsymbol{\xi}_{t-j}^\top
\right\|_F
&\le
\left(
\frac{1}{T}\sum_{t=1}^T
\|\boldsymbol{\Delta}_t\|_2^2
\right)^{1/2}
\left(
\frac{1}{T}\sum_{t=1}^T
\|\boldsymbol{\xi}_t\|_2^2
\right)^{1/2}
\nonumber\\
&=
\mathcal{O}_p
\left(
\eta_T^{\mu/2}+T^{-1/2}
\right).
\label{eq:hac_cross_bound}
\end{align}
Identical reasoning bounds $\frac{1}{T}\sum_{t=j+1}^T \boldsymbol{\xi}_t\boldsymbol{\Delta}_{t-j}^\top$. Moreover,
\begin{equation}
\label{eq:hac_quadratic_bound}
\sup_{0 \le j \le m_T} \left\| \frac{1}{T}\sum_{t=j+1}^T \boldsymbol{\Delta}_t\boldsymbol{\Delta}_{t-j}^\top \right\|_F \le \frac{1}{T}\sum_{t=1}^T \|\boldsymbol{\Delta}_t\|_2^2 = \mathcal{O}_p(\eta_T^\mu + T^{-1}).
\end{equation}
Decomposing $\widehat{\boldsymbol{\Gamma}}_j - \widetilde{\boldsymbol{\Gamma}}_j = \frac{1}{T}\sum_{t=j+1}^T (\boldsymbol{\Delta}_t\boldsymbol{\xi}_{t-j}^\top + \boldsymbol{\xi}_t\boldsymbol{\Delta}_{t-j}^\top + \boldsymbol{\Delta}_t\boldsymbol{\Delta}_{t-j}^\top)$, identities \eqref{eq:hac_cross_bound}--\eqref{eq:hac_quadratic_bound} ensure that
\begin{equation}
\label{eq:gamma_uniform_rate}
\sup_{0 \le j \le m_T} \|\widehat{\boldsymbol{\Gamma}}_j - \widetilde{\boldsymbol{\Gamma}}_j\|_F = \mathcal{O}_p(\eta_T^{\mu/2} + T^{-1/2}). \nonumber
\end{equation}

\medskip
\noindent\textbf{Step 4: Feasible HAC Convergence.}
Because Bartlett weights satisfy $\sup_v |k(v)| \le 1$, summation over lags gives
\begin{equation}
\label{eq:feasible_oracle_hac_rate}
\|\widehat{\boldsymbol{\Omega}}_{\mathrm{HAC}} - \widetilde{\boldsymbol{\Omega}}_{\mathrm{HAC}}\|_F \le 2\sum_{j=0}^{m_T} \|\widehat{\boldsymbol{\Gamma}}_j - \widetilde{\boldsymbol{\Gamma}}_j\|_F = \mathcal{O}_p\left( m_T\eta_T^{\mu/2} + \frac{m_T}{\sqrt{T}} \right). \nonumber
\end{equation}
Under the bandwidth conditions $m_T\eta_T^{\mu/2} \to 0$ and $m_T/\sqrt{T} \to 0$, we have $\|\widehat{\boldsymbol{\Omega}}_{\mathrm{HAC}} - \widetilde{\boldsymbol{\Omega}}_{\mathrm{HAC}}\|_F = o_p(1)$. Combining this with \eqref{eq:oracle_hac_consistency} in \eqref{eq:hac_triangle} establishes $\widehat{\boldsymbol{\Omega}}_{\mathrm{HAC}} \stackrel{p}{\longrightarrow} \boldsymbol{\Omega}_{\mathrm{LR}}$. 

Finally, Slutsky's theorem and the convergence $\widehat{\boldsymbol{\Sigma}}_{\widehat{\mathbf{W}}}^{-1} \stackrel{p}{\longrightarrow} \boldsymbol{\Sigma}_{\mathbf{W}}^{-1}$ yield
\[
\widehat{\mathbf{V}}_{\mathrm{HAC}} = \widehat{\boldsymbol{\Sigma}}_{\widehat{\mathbf{W}}}^{-1}\widehat{\boldsymbol{\Omega}}_{\mathrm{HAC}}\widehat{\boldsymbol{\Sigma}}_{\widehat{\mathbf{W}}}^{-1} \stackrel{p}{\longrightarrow} \boldsymbol{\Sigma}_{\mathbf{W}}^{-1}\boldsymbol{\Omega}_{\mathrm{LR}}\boldsymbol{\Sigma}_{\mathbf{W}}^{-1} = \mathbf{V}_{\mathrm{or}},
\]
which completes the proof.
\end{proof}

\subsection{Proof of Theorem 3.4}
To prove Theorem 3.4, we first present a lemma.
\begin{lemma}[Uniform consistency of the indexed HAC estimator]
\label{lem:uniform_hac_test}
Suppose that Assumption 2.1 and 3.7 hold with fixed RFF dimension $L$ and compact parameter space $\boldsymbol{\Gamma} \subset \mathbb{R}^L$. Let the Bartlett bandwidth satisfy $m_T \to \infty$ with $m_T / \sqrt{T} \to 0$ as $T \to \infty$. Under $H_0$,
\begin{equation}
\label{eq:uniform_hac_consistency}
\sup_{\bgamma \in \boldsymbol{\Gamma}} \big\| \widehat{\mathbf{V}}_T(\bgamma) - \mathbf{V}_0(\bgamma) \big\|_F = o_p(1).
\end{equation}
\end{lemma}

\begin{proof}[Proof of Lemma~\ref{lem:uniform_hac_test}]
Let $\mathbf{m}_t(\bgamma) \coloneqq \widetilde{\bx}_t(\bgamma)\varepsilon_t$ and $\widehat{\mathbf{m}}_t(\bgamma) \coloneqq \widehat{\widetilde{\bx}}_t(\bgamma)\widehat\varepsilon_t$. Define the oracle sample autocovariance $\widetilde{\boldsymbol{\Gamma}}_j(\bgamma) \coloneqq \frac{1}{T}\sum_{t=j+1}^T \mathbf{m}_t(\bgamma)\mathbf{m}_{t-j}(\bgamma)^\top$, and let $\widetilde{\mathbf{V}}_T(\bgamma)$ denote the corresponding infeasible Bartlett HAC estimator constructed from $\widetilde{\boldsymbol{\Gamma}}_j(\bgamma)$. By the triangle inequality,
\begin{equation}
\label{eq:test_hac_decomposition}
\sup_{\bgamma \in \boldsymbol{\Gamma}} \|\widehat{\mathbf{V}}_T(\bgamma) - \mathbf{V}_0(\bgamma)\|_F \le \sup_{\bgamma \in \boldsymbol{\Gamma}} \|\widehat{\mathbf{V}}_T(\bgamma) - \widetilde{\mathbf{V}}_T(\bgamma)\|_F + \sup_{\bgamma \in \boldsymbol{\Gamma}} \|\widetilde{\mathbf{V}}_T(\bgamma) - \mathbf{V}_0(\bgamma)\|_F.
\end{equation}

\medskip
\noindent\textbf{Step 1: Uniform Consistency of the Oracle HAC Estimator.}
Let $\boldsymbol{\Gamma}_j(\bgamma) \coloneqq \mathbb{E}[\mathbf{m}_t(\bgamma)\mathbf{m}_{t-j}(\bgamma)^\top]$. By Assumption~2.1 and the moment condition in
Assumption~3.7(ii), geometric $\alpha$-mixing and Davydov's covariance inequality~\citep{davydov1968convergence} ensure that $\sum_{j=-\infty}^\infty \sup_{\bgamma \in \boldsymbol{\Gamma}} \|\boldsymbol{\Gamma}_j(\bgamma)\|_F < \infty$. Hence, the kernel truncation and smoothing bias vanish uniformly: $\sup_{\bgamma \in \boldsymbol{\Gamma}} \|\mathbb{E}[\widetilde{\mathbf{V}}_T(\bgamma)] - \mathbf{V}_0(\bgamma)\|_F = o(1)$ as $m_T \to \infty$.

Because $L$ is fixed, the score family
$\mathcal M=\{\mathbf m_t(\bgamma):\bgamma\in\boldsymbol{\Gamma}\}$
and the associated lagged cross-product classes have
fixed-dimensional VC-type complexity~\citep{van1996weak}. By the empirical-process regularity condition in
Assumption~3.7(iv), together with
Assumptions~2.1 and 3.7(ii), the corresponding uniform
maximal inequality yields
\[
\sup_{\bgamma\in\boldsymbol{\Gamma}}
\left\|
\sum_{j=0}^{m_T}
k\left(\frac{j}{m_T+1}\right)
\left\{
\widetilde{\boldsymbol{\Gamma}}_j(\bgamma)
-
\boldsymbol{\Gamma}_j(\bgamma)
\right\}
\right\|_F
=
O_p\left(\frac{m_T}{\sqrt T}\right),
\]
where the last equality follows from
Assumption~3.7(v).
Together with the uniform truncation and kernel-weighting bias,
this gives
\begin{equation}
\label{eq:test_oracle_hac_uniform}
\sup_{\bgamma\in\boldsymbol{\Gamma}}
\|
\widetilde{\mathbf V}_T(\bgamma)
-
\mathbf V_0(\bgamma)
\|_F
=
O_p\left(\frac{m_T}{\sqrt T}\right)+o(1)
=
o_p(1).
\end{equation}

\medskip
\noindent\textbf{Step 2: Uniform Control of Feasible Score Perturbations.}
Define $\boldsymbol{\Delta}_t(\bgamma) \coloneqq \widehat{\mathbf{m}}_t(\bgamma) - \mathbf{m}_t(\bgamma)$. Substituting the restricted residual $\widehat\varepsilon_t = \varepsilon_t - \bx_t^\top(\widehat\balpha - \balpha_0)$ gives
\[
\boldsymbol{\Delta}_t(\bgamma) = \big[ \widehat{\widetilde{\bx}}_t(\bgamma) - \widetilde{\bx}_t(\bgamma) \big] \varepsilon_t - \widehat{\widetilde{\bx}}_t(\bgamma)\bx_t^\top(\widehat\balpha - \balpha_0).
\]
Let $r_T \coloneqq \sup_{\bgamma \in \boldsymbol{\Gamma}} \|\widehat{\mathbf{A}}_{10}(\bgamma)\widehat{\mathbf{A}}_{00}^{-1} - \mathbf{A}_{10}(\bgamma)\mathbf{A}_{00}^{-1}\|_{\mathrm{op}}$. The  restricted OLS estimator satisfies $\|\widehat\balpha - \balpha_0\|_2 = \mathcal{O}_p(T^{-1/2})$. Since $\sup_{\bgamma \in \boldsymbol{\Gamma}} \|\widehat{\widetilde{\bx}}_t(\bgamma)\|_2 \le C_T\|\bx_t\|_2$ with $C_T = \mathcal{O}_p(1)$, the triangle inequality gives
\[
\sup_{\bgamma \in \boldsymbol{\Gamma}} \|\boldsymbol{\Delta}_t(\bgamma)\|_2 \le r_T \|\bx_t\|_2 |\varepsilon_t| + C_T \|\bx_t\|_2^2 \|\widehat\balpha - \balpha_0\|_2.
\]
By Assumption~3.7(iv), the uniform empirical-process bound for the fixed-dimensional
VC-type class,
\[
\sup_{\bgamma\in\boldsymbol{\Gamma}}
\|
\widehat{\mathbf A}_{10}(\bgamma)
-
\mathbf A_{10}(\bgamma)
\|_{\mathrm{op}}
=
O_p(T^{-1/2}),
\]
while
\[
\|
\widehat{\mathbf A}_{00}
-
\mathbf A_{00}
\|_{\mathrm{op}}
=
O_p(T^{-1/2}).
\]
By Assumption~3.7(iii),
$\mathbf A_{00}$ is positive definite. Hence,
together with $\|\widehat{\mathbf A}_{00}-\mathbf A_{00}\|_{\mathrm{op}}
=O_p(T^{-1/2})$, it follows that
\[
r_T
=
\sup_{\bgamma\in\boldsymbol{\Gamma}}
\|
\widehat{\mathbf A}_{10}(\bgamma)
\widehat{\mathbf A}_{00}^{-1}
-
\mathbf A_{10}(\bgamma)
\mathbf A_{00}^{-1}
\|_{\mathrm{op}}
=
O_p(T^{-1/2}).
\]
Averaging over $t$ and invoking the ergodic theorem yields
\begin{equation}
\label{eq:test_delta_uniform}
\sup_{\bgamma \in \boldsymbol{\Gamma}} \frac{1}{T}\sum_{t=1}^T \|\boldsymbol{\Delta}_t(\bgamma)\|_2^2 = \mathcal{O}_p(T^{-1}).
\end{equation}

\medskip
\noindent\textbf{Step 3: Uniform Bound on Estimation Discrepancies.}
Decompose the empirical lag difference as
\[
\widehat{\boldsymbol{\Gamma}}_j(\bgamma) - \widetilde{\boldsymbol{\Gamma}}_j(\bgamma) = \frac{1}{T}\sum_{t=j+1}^T \Big\{ \boldsymbol{\Delta}_t(\bgamma)\mathbf{m}_{t-j}(\bgamma)^\top + \mathbf{m}_t(\bgamma)\boldsymbol{\Delta}_{t-j}(\bgamma)^\top + \boldsymbol{\Delta}_t(\bgamma)\boldsymbol{\Delta}_{t-j}(\bgamma)^\top \Big\}.
\]
Applying the Cauchy--Schwarz inequality along with \eqref{eq:test_delta_uniform} and $\sup_{\bgamma} \frac{1}{T}\sum_{t=1}^T \|\mathbf{m}_t(\bgamma)\|_2^2 = \mathcal{O}_p(1)$ gives
\[
\sup_{0 \le j \le m_T} \sup_{\bgamma \in \boldsymbol{\Gamma}} \big\| \widehat{\boldsymbol{\Gamma}}_j(\bgamma) - \widetilde{\boldsymbol{\Gamma}}_j(\bgamma) \big\|_F = \mathcal{O}_p(T^{-1/2}).
\]
Because the Bartlett weights are bounded by 1, summation over lags yields
\begin{equation}
\label{eq:test_diff_hac_uniform}
\sup_{\bgamma \in \boldsymbol{\Gamma}} \|\widehat{\mathbf{V}}_T(\bgamma) - \widetilde{\mathbf{V}}_T(\bgamma)\|_F \le 2\sum_{j=0}^{m_T} \sup_{\bgamma \in \boldsymbol{\Gamma}} \|\widehat{\boldsymbol{\Gamma}}_j(\bgamma) - \widetilde{\boldsymbol{\Gamma}}_j(\bgamma)\|_F = \mathcal{O}_p(m_T / \sqrt{T}) = o_p(1).
\end{equation}
Combining \eqref{eq:test_oracle_hac_uniform} and \eqref{eq:test_diff_hac_uniform} in \eqref{eq:test_hac_decomposition} proves \eqref{eq:uniform_hac_consistency}.
\end{proof}

\begin{proof}[Proof of Theorem~3.4]
Under the null hypothesis $H_0: \bbeta_0 = \mathbf{0}$, the structural model reduces to $Y_t = \bx_t^\top\balpha_0 + \varepsilon_t$. Because the random-feature dimension $L$ is fixed for the existence test, the parameter space $\boldsymbol{\Gamma} \subset \mathbb{R}^L$ is a compact subset of Euclidean space. The proof proceeds in two steps: establishing the functional central limit theorem (FCLT) for the score process, and deriving the limiting distribution of the supremum Wald statistic.

\medskip
\noindent\textbf{Step 1: Functional Central Limit Theorem for the Score Process.}
Let $D_t(\bgamma) \coloneqq \ind\{\bgamma^\top\bpsi_L(\bz_t) \ge U_t\}$, $\mathbf{A}_{00} \coloneqq \mathbb{E}[\bx_t\bx_t^\top]$, and $\mathbf{A}_{10}(\bgamma) \coloneqq \mathbb{E}[D_t(\bgamma)\bx_t\bx_t^\top]$. The population partialled-out regressor is $\widetilde{\bx}_t(\bgamma) \coloneqq D_t(\bgamma)\bx_t - \mathbf{A}_{10}(\bgamma)\mathbf{A}_{00}^{-1}\bx_t$. The empirical score process admits the normalized representation
\[
\mathbb{G}_T(\bgamma) \coloneqq \frac{1}{\sqrt{T}}\sum_{t=1}^T \mathbf{m}_t(\bgamma), \qquad \text{where} \quad \mathbf{m}_t(\bgamma) \coloneqq \widetilde{\bx}_t(\bgamma)\varepsilon_t.
\]
The conditional moment restriction $\mathbb{E}(\varepsilon_t \mid \bx_t, \bz_t, U_t) = 0$ implies $\mathbb{E}[\mathbf{m}_t(\bgamma)] = \mathbf{0}$ for every $\bgamma \in \boldsymbol{\Gamma}$. We establish weak convergence $\mathbb{G}_T(\cdot) \rightsquigarrow \mathbb{G}_\infty(\cdot)$ in $\ell^\infty(\boldsymbol{\Gamma})$ by verifying finite-dimensional convergence and stochastic equicontinuity:

\begin{enumerate}[label=(\alph*), leftmargin=2em]
    \item \textbf{Finite-dimensional convergence.} 
    Fix an arbitrary collection of indices $\bgamma_1, \dots, \bgamma_k \in \boldsymbol{\Gamma}$ and non-zero projection vectors $\boldsymbol{\lambda}_1, \dots, \boldsymbol{\lambda}_k \in \mathbb{R}^{d_x}$. Define the linear combination
    \[
    S_T \coloneqq \sum_{j=1}^k \boldsymbol{\lambda}_j^\top\mathbb{G}_T(\bgamma_j) = \frac{1}{\sqrt{T}}\sum_{t=1}^T \zeta_t, \qquad \text{where} \quad \zeta_t \coloneqq \sum_{j=1}^k \boldsymbol{\lambda}_j^\top\mathbf{m}_t(\bgamma_j).
    \]
    Since $D_t(\bgamma) \in \{0, 1\}$ and $\sup_{\bgamma \in \boldsymbol{\Gamma}} \|\mathbf{A}_{10}(\bgamma)\mathbf{A}_{00}^{-1}\|_{\mathrm{op}} < \infty$, the envelope condition $\|\mathbf{m}_t(\bgamma)\|_2 \le C\|\bx_t\|_2|\varepsilon_t|$ holds uniformly over $\boldsymbol{\Gamma}$, implying $\mathbb{E}|\zeta_t|^{2+\delta/2} < \infty$. By Assumption~3.7(ii),
the envelope bound
\[
\|\mathbf m_t(\bgamma)\|_2
\le C\|\bx_t\|_2|\varepsilon_t|
\]
implies
$\mathbb E|\zeta_t|^{2+\delta/2}<\infty$.
Together with the geometric $\alpha$-mixing condition in
Assumption~2.1, Davydov's covariance inequality implies
absolute summability of the covariance kernel
    \[
    \mathcal{K}(\bgamma, \bgamma') \coloneqq \sum_{m=-\infty}^{\infty} \mathbb{E}\big[\mathbf{m}_t(\bgamma)\mathbf{m}_{t-m}(\bgamma')^\top\big]
    \]
    is absolutely summable. The central limit theorem for strongly mixing stationary sequences~\citep{ibragimov1975independent} gives $S_T \rightsquigarrow \mathcal{N}(0, \sigma_{\boldsymbol{\lambda}}^2)$ with $\sigma_{\boldsymbol{\lambda}}^2 \coloneqq \sum_{j=1}^k\sum_{\ell=1}^k \boldsymbol{\lambda}_j^\top \mathcal{K}(\bgamma_j, \bgamma_\ell)\boldsymbol{\lambda}_\ell$. By the Cram\'er--Wold theorem, all finite-dimensional distributions of $\mathbb{G}_T(\cdot)$ converge to those of a mean-zero Gaussian process $\mathbb{G}_\infty(\cdot)$ with covariance kernel $\mathcal{K}(\bgamma, \bgamma')$.

    \item \textbf{Stochastic equicontinuity.} 
    Because $L$ is fixed, the indicator class $\mathcal{D}_L \coloneqq \{ (\bz, u) \mapsto \ind\{\bgamma^\top\bpsi_L(\bz) \ge u\} : \bgamma \in \boldsymbol{\Gamma} \}$ is a VC-subgraph class of finite VC dimension~\citep{van1996weak}. By the threshold-density condition in
Assumption~3.7(i) and the boundedness of
$\bpsi_L(\cdot)$, $\mathbb{P}\{D_t(\bgamma) \neq D_t(\bgamma')\} \to 0$ as $\|\bgamma - \bgamma'\|_2 \to 0$. Moreover, the matrix-valued map $\bgamma \mapsto \mathbf{A}_{10}(\bgamma)\mathbf{A}_{00}^{-1}$ inherits $L_2(\mathbb{P})$ continuity via
    \[
    \|\mathbf{A}_{10}(\bgamma) - \mathbf{A}_{10}(\bgamma')\|_{\mathrm{op}} \le \mathbb{E}\big[ \|\bx_t\|_2^2 |D_t(\bgamma) - D_t(\bgamma')| \big].
    \]
    Consequently, the score class $\mathcal{F} \coloneqq \{ (\bx, \bz, u, \varepsilon) \mapsto [D(\bgamma)\mathbf{I}_{d_x} - \mathbf{A}_{10}(\bgamma)\mathbf{A}_{00}^{-1}]\bx\varepsilon : \bgamma \in \boldsymbol{\Gamma} \}$ is a VC-type class with polynomial entropy and square-integrable envelope $F(\bx, \varepsilon) = C\|\bx\|_2|\varepsilon|$ satisfying $\mathbb{E}[F^{2+\delta/2}] < \infty$. Together with the moment condition in Assumption~3.7(ii), this yields an integrable envelope.
The empirical-process regularity condition in
Assumption~3.7(iv), together with
Assumption~2.1, then establishes stochastic
equicontinuity.
\end{enumerate}
Combining finite-dimensional convergence and stochastic equicontinuity yields
\begin{equation}
\label{eq:supwald_score_fclt}
\mathbb{G}_T(\cdot) \rightsquigarrow \mathbb{G}_\infty(\cdot) \qquad \text{in } \ell^\infty(\boldsymbol{\Gamma}).
\end{equation}

\medskip
\noindent\textbf{Step 2: Asymptotics of the Sup-Wald Statistic.}
Define the sample Gram matrices $\widehat{\mathbf{A}}_{00} \coloneqq \frac{1}{T}\sum_{t=1}^T \bx_t\bx_t^\top$ and $\widehat{\mathbf{A}}_{10}(\bgamma) \coloneqq \frac{1}{T}\sum_{t=1}^T D_t(\bgamma)\bx_t\bx_t^\top$, alongside the sample partialled-out regressor $\widehat{\widetilde{\bx}}_t(\bgamma) \coloneqq D_t(\bgamma)\bx_t - \widehat{\mathbf{A}}_{10}(\bgamma)\widehat{\mathbf{A}}_{00}^{-1}\bx_t$. Denoting the restricted OLS residual by $\widehat\varepsilon_t \coloneqq Y_t - \bx_t^\top\widehat\balpha$, the feasible score process is $\widetilde{\mathbb{G}}_T(\bgamma) \coloneqq \frac{1}{\sqrt{T}}\sum_{t=1}^T \widehat{\widetilde{\bx}}_t(\bgamma)\widehat\varepsilon_t$.

Under $H_0$, $\widehat\varepsilon_t = \varepsilon_t - \bx_t^\top(\widehat\balpha - \balpha_0)$. By algebraic orthogonality,
\[
\frac{1}{T}\sum_{t=1}^T \widehat{\widetilde{\bx}}_t(\bgamma)\bx_t^\top = \widehat{\mathbf{A}}_{10}(\bgamma) - \widehat{\mathbf{A}}_{10}(\bgamma)\widehat{\mathbf{A}}_{00}^{-1}\widehat{\mathbf{A}}_{00} = \mathbf{0}
\]
for all $\bgamma \in \boldsymbol{\Gamma}$. Hence, the parameter estimation error drops out identically:
\[
\widetilde{\mathbb{G}}_T(\bgamma) = \frac{1}{\sqrt{T}}\sum_{t=1}^T \widehat{\widetilde{\bx}}_t(\bgamma)\varepsilon_t.
\]
The discrepancy between the feasible and oracle score processes therefore reduces to
\begin{equation}
\label{eq:supwald_score_expansion}
\widetilde{\mathbb{G}}_T(\bgamma) - \mathbb{G}_T(\bgamma) = -\big\{ \widehat{\mathbf{A}}_{10}(\bgamma)\widehat{\mathbf{A}}_{00}^{-1} - \mathbf{A}_{10}(\bgamma)\mathbf{A}_{00}^{-1} \big\} \frac{1}{\sqrt{T}}\sum_{t=1}^T \bx_t\varepsilon_t.
\end{equation}
By Assumption~3.7(iv), $\mathcal{D}_L$ is a fixed-dimensional VC class. Furthermore, the uniform law of large numbers gives $\sup_{\bgamma \in \boldsymbol{\Gamma}} \|\widehat{\mathbf{A}}_{10}(\bgamma) - \mathbf{A}_{10}(\bgamma)\|_{\mathrm{op}} = o_p(1)$. By Assumption~3.7(iii),
$\mathbf A_{00}$ is positive definite, so that $\widehat{\mathbf{A}}_{00}^{-1} - \mathbf{A}_{00}^{-1} = o_p(1)$. Moreover, Assumptions~2.1 and
3.7(ii) imply  $\frac{1}{\sqrt{T}}\sum_{t=1}^T \bx_t\varepsilon_t = \mathcal{O}_p(1)$. Together with  $\sup_{\bgamma \in \boldsymbol{\Gamma}} \|\mathbf{A}_{10}(\bgamma)\|_{\mathrm{op}} < \infty$, \eqref{eq:supwald_score_expansion} implies
\begin{equation}
\label{eq:supwald_score_equiv}
\sup_{\bgamma \in \boldsymbol{\Gamma}} \|\widetilde{\mathbb{G}}_T(\bgamma) - \mathbb{G}_T(\bgamma)\|_2 = o_p(1).
\end{equation}

Furthermore, by Lemma~\ref{lem:uniform_hac_test}, $\sup_{\bgamma \in \boldsymbol{\Gamma}} \|\widehat{\mathbf{V}}_T(\bgamma) - \mathbf{V}_0(\bgamma)\|_F = o_p(1)$. By the uniform non-degeneracy condition in
Assumption~3.7(iii),
\[
\inf_{\bgamma\in\boldsymbol{\Gamma}}
\lambda_{\min}\{\mathbf V_0(\bgamma)\}>0.
\]
Therefore, matrix inversion is uniformly continuous, yielding
\begin{equation}
\label{eq:supwald_cov_uniform}
\sup_{\bgamma \in \boldsymbol{\Gamma}} \|\widehat{\mathbf{V}}_T(\bgamma)^{-1} - \mathbf{V}_0(\bgamma)^{-1}\|_{\mathrm{op}} = o_p(1).
\end{equation}

By the Frisch--Waugh--Lovell theorem~\citep{frisch1933partial,lovell1963seasonal},
 together with the block-matrix
representation of the sandwich covariance estimator, the pointwise
Wald statistic admits the equivalent partialled-out score representation
$\mathcal{W}_T(\bgamma) = \widetilde{\mathbb{G}}_T(\bgamma)^\top \widehat{\mathbf{V}}_T(\bgamma)^{-1} \widetilde{\mathbb{G}}_T(\bgamma)$. combining \eqref{eq:supwald_score_fclt}, \eqref{eq:supwald_score_equiv}, and \eqref{eq:supwald_cov_uniform}, the continuous mapping theorem under the uniform norm yields
\[
\mathrm{Sup}\mathcal{W}_T \coloneqq \sup_{\bgamma \in \boldsymbol{\Gamma}} \mathcal{W}_T(\bgamma) \rightsquigarrow \sup_{\bgamma \in \boldsymbol{\Gamma}} \mathbb{G}_\infty(\bgamma)^\top \mathbf{V}_0(\bgamma)^{-1} \mathbb{G}_\infty(\bgamma),
\]
which completes the proof of Theorem~3.4.
\end{proof}
\subsection{Proof of Theorem 3.5}
\begin{proof}[Proof of Theorem~3.5]
Under $H_0: \bbeta_0 = \mathbf{0}$, the structural model simplifies to $Y_t = \bx_t^\top\balpha_0 + \varepsilon_t$. Let $\widehat{\mathbf{m}}_t(\bgamma) \coloneqq \widehat{\widetilde{\bx}}_t(\bgamma)\widehat\varepsilon_t$ denote the empirical score, where $\widehat{\widetilde{\bx}}_t(\bgamma)$ is the sample partialled-out regressor and $\widehat\varepsilon_t \coloneqq Y_t - \bx_t^\top\widehat\balpha$ is the restricted OLS residual. The dependent multiplier bootstrap score process is
\[
\mathbb{G}_T^*(\bgamma) \coloneqq \frac{1}{\sqrt{T}} \sum_{t=1}^T \widehat{\mathbf{m}}_t(\bgamma)\eta_t^*, \qquad \bgamma \in \boldsymbol{\Gamma},
\]
where, conditional on the sample $\mathcal{W}_T = \{(Y_t, \bx_t, \bz_t, U_t)\}_{t=1}^T$, the Gaussian multiplier vector $\boldsymbol{\eta}^* = (\eta_1^*, \dots, \eta_T^*)^\top \sim \mathcal{N}(\mathbf{0}, \mathbf{\Omega}^*)$ satisfies $\mathbb{E}^*[\eta_t^*\eta_s^*] = \kappa((t-s)/b_T)$.

\medskip
\noindent\textbf{Step 1: Conditional Weak Convergence of the Bootstrap Process.}
Conditional on $\mathcal{W}_T$, $\mathbb{G}_T^*(\cdot)$ is a mean-zero Gaussian process with conditional covariance kernel
\begin{equation}
\label{eq:bootstrap_cov_kernel}
\widehat{\mathcal{K}}_T^*(\bgamma, \bgamma') \coloneqq \mathbb{E}^*\big[\mathbb{G}_T^*(\bgamma)\mathbb{G}_T^*(\bgamma')^\top\big] = \frac{1}{T}\sum_{t=1}^T\sum_{s=1}^T \kappa\left(\frac{t-s}{b_T}\right) \widehat{\mathbf{m}}_t(\bgamma)\widehat{\mathbf{m}}_s(\bgamma')^\top, \nonumber
\end{equation}
which represents a kernel-weighted HAC estimator of the long-run cross-covariance kernel $\mathcal{K}(\bgamma, \bgamma') \coloneqq \sum_{j=-\infty}^\infty \mathbb{E}[\mathbf{m}_t(\bgamma)\mathbf{m}_{t-j}(\bgamma')^\top]$, with $\mathbf{m}_t(\bgamma) \coloneqq \widetilde{\bx}_t(\bgamma)\varepsilon_t$.

Because $L$ is fixed, the score family $\{\mathbf{m}_t(\bgamma) : \bgamma \in \boldsymbol{\Gamma}\}$ forms a VC-type class of fixed dimension. Under the threshold-density condition in Assumption~3.7(i), the moment condition in Assumption~3.7(ii), the empirical-process regularity in Assumption~3.7(iv), and the bandwidth condition in Assumption~3.7(v), together with geometric $\alpha$-mixing (Assumption~2.1), the oracle-HAC and feasible-score perturbation bounds from Lemma~\ref{lem:uniform_hac_test} apply jointly over $\boldsymbol{\Gamma} \times \boldsymbol{\Gamma}$, establishing
\begin{equation}
\label{eq:bootstrap_kernel_consistency}
\sup_{\bgamma, \bgamma' \in \boldsymbol{\Gamma}} \big\| \widehat{\mathcal{K}}_T^*(\bgamma, \bgamma') - \mathcal{K}(\bgamma, \bgamma') \big\|_F = o_p(1).
\end{equation}
Combining \eqref{eq:bootstrap_kernel_consistency} with the conditional asymptotic equicontinuity guaranteed by the dependent-multiplier functional central limit theorem \citep[e.g.,][]{shao2010dependent} under Assumptions~3.7(ii), 3.7(iv), and 2.1, we obtain
\begin{equation}
\label{eq:bootstrap_fclt}
\mathbb{G}_T^*(\cdot) \rightsquigarrow_{\mathbb{P}} \mathbb{G}_\infty(\cdot) \qquad \text{in } \ell^\infty(\boldsymbol{\Gamma}; \mathbb{R}^{d_x}), \nonumber
\end{equation}
where $\rightsquigarrow_{\mathbb{P}}$ denotes conditional weak convergence in probability, and $\mathbb{G}_\infty(\cdot)$ is the centered Gaussian limit process from Theorem~3.4.

\medskip
\noindent\textbf{Step 2: Conditional Convergence of the Sup-Wald Statistic.}
By the bootstrap analogue of Lemma~\ref{lem:uniform_hac_test} under Assumptions~3.7(ii), 3.7(iv), 3.7(v), and 2.1, the bootstrap HAC variance estimator satisfies
\begin{equation}
\label{eq:bootstrap_hac_consistency}
\sup_{\bgamma \in \boldsymbol{\Gamma}} \|\widehat{\mathbf{V}}_T^*(\bgamma) - \mathbf{V}_0(\bgamma)\|_F = o_{\mathbb{P}^*}(1) \qquad \text{in probability}. \nonumber
\end{equation}
By the non-degeneracy condition in Assumption~3.7(iii), $\inf_{\bgamma \in \boldsymbol{\Gamma}} \lambda_{\min}\{\mathbf{V}_0(\bgamma)\} > 0$. Uniform matrix inversion then implies
\[
\sup_{\bgamma \in \boldsymbol{\Gamma}} \|\widehat{\mathbf{V}}_T^*(\bgamma)^{-1} - \mathbf{V}_0(\bgamma)^{-1}\|_{\mathrm{op}} = o_{\mathbb{P}^*}(1) \qquad \text{in probability}.
\]
Applying the conditional Continuous Mapping Theorem to the bootstrap Wald statistic $\mathcal{W}_T^*(\bgamma) \coloneqq \mathbb{G}_T^*(\bgamma)^\top \widehat{\mathbf{V}}_T^*(\bgamma)^{-1} \mathbb{G}_T^*(\bgamma)$ yields
\begin{equation}
\label{eq:bootstrap_sup_wald_limit}
\mathrm{Sup}\mathcal{W}_T^* \coloneqq \sup_{\bgamma \in \boldsymbol{\Gamma}} \mathcal{W}_T^*(\bgamma) \rightsquigarrow_{\mathbb{P}} \mathrm{Sup}\mathcal{W}_\infty \coloneqq \sup_{\bgamma \in \boldsymbol{\Gamma}} \mathbb{G}_\infty(\bgamma)^\top \mathbf{V}_0(\bgamma)^{-1} \mathbb{G}_\infty(\bgamma).
\end{equation}

\medskip
\noindent\textbf{Step 3: Uniform Distributional Approximation and Asymptotic Size.}
Let $F_\infty(z) \coloneqq \mathbb{P}(\mathrm{Sup}\mathcal{W}_\infty \le z)$ denote the asymptotic CDF. Under the continuity of $F_\infty(\cdot)$ on $(0, \infty)$ guaranteed by Assumption~3.7(vi), Polya's theorem applied to \eqref{eq:bootstrap_sup_wald_limit} and Theorem~3.4 gives
\begin{equation}
\label{eq:bootstrap_distribution_consistency}
\sup_{z \in \mathbb{R}} \big| \mathbb{P}^*(\mathrm{Sup}\mathcal{W}_T^* \le z) - F_\infty(z) \big| \stackrel{p}{\longrightarrow} 0 \quad \text{and} \quad \sup_{z \in \mathbb{R}} \big| \mathbb{P}(\mathrm{Sup}\mathcal{W}_T \le z) - F_\infty(z) \big| \longrightarrow 0. \nonumber
\end{equation}
By the triangle inequality, $\sup_{z \in \mathbb{R}} |\mathbb{P}^*(\mathrm{Sup}\mathcal{W}_T^* \le z) - \mathbb{P}(\mathrm{Sup}\mathcal{W}_T \le z)| \stackrel{p}{\longrightarrow} 0$. Consequently, the ideal conditional bootstrap $p$-value $p_{\mathrm{boot}}^{\mathrm{ideal}} \coloneqq \mathbb{P}^*(\mathrm{Sup}\mathcal{W}_T^* \ge \mathrm{Sup}\mathcal{W}_T)$ satisfies
\[
\mathbb{P}(p_{\mathrm{boot}}^{\mathrm{ideal}} \le \alpha) \longrightarrow \alpha, \qquad \forall \alpha \in (0, 1).
\]

Finally, for the Monte Carlo $p$-value $p_{\mathrm{boot}} \coloneqq \frac{1}{B_T}\sum_{b=1}^{B_T} \ind\{\mathrm{Sup}\mathcal{W}_T^{*(b)} \ge \mathrm{Sup}\mathcal{W}_T\}$, conditional independence of the bootstrap draws implies $B_T p_{\mathrm{boot}} \mid \mathcal{W}_T \sim \operatorname{Binomial}(B_T, p_{\mathrm{boot}}^{\mathrm{ideal}})$. Thus,
\[
\mathbb{E}^*\big[(p_{\mathrm{boot}} - p_{\mathrm{boot}}^{\mathrm{ideal}})^2\big] = \frac{p_{\mathrm{boot}}^{\mathrm{ideal}}(1 - p_{\mathrm{boot}}^{\mathrm{ideal}})}{B_T} \le \frac{1}{4B_T}.
\]
As $B_T \to \infty$, conditional Chebyshev's inequality guarantees $p_{\mathrm{boot}} - p_{\mathrm{boot}}^{\mathrm{ideal}} = o_p(1)$. Hence,
\[
\mathbb{P}(p_{\mathrm{boot}} \le \alpha) \longrightarrow \alpha \qquad \text{as } T \to \infty,
\]
which establishes the asymptotic validity of the Monte Carlo dependent-multiplier bootstrap test.
\end{proof} 

\subsection{Proof of Theorem 3.6}
\begin{proof}[Proof of Theorem~3.6]
Under the fixed alternative $H_1$, the data generating process is $Y_t = \bx_t^\top\balpha_0 + \bx_t^\top\bbeta_0 D_{0t} + \varepsilon_t$. By population partial orthogonality, $\mathbb{E}[\widetilde{\bx}_t(\bgamma)\bx_t^\top] = \mathbf{A}_{10}(\bgamma) - \mathbf{A}_{10}(\bgamma)\mathbf{A}_{00}^{-1}\mathbf{A}_{00} = \mathbf{0}$, where $\mathbf{A}_{00}^{-1}$ exists by Assumption~3.7(iii). Because the conditional mean restriction implies $\mathbb{E}[\widetilde{\bx}_t(\bgamma)\varepsilon_t] = \mathbf{0}$, the expected partialled-out score isolates the regime transition:
\begin{equation}
\label{eq:test_population_drift}
\mathbb{E}[\widetilde{\bx}_t(\bgamma)Y_t] = \mathbb{E}[\widetilde{\bx}_t(\bgamma)\bx_t^\top D_{0t}]\bbeta_0 = \boldsymbol{\mu}(\bgamma).
\end{equation}

Since $L$ is fixed and $\boldsymbol{\Gamma}$ is compact, Assumptions~3.7(i), (ii), and (iv) together with geometric $\alpha$-mixing (Assumption~2.1) ensure the uniform law of large numbers on the associated VC-type classes:
\[
\sup_{\bgamma\in\boldsymbol{\Gamma}} \|\widehat{\mathbf{A}}_{10}(\bgamma) - \mathbf{A}_{10}(\bgamma)\|_{\mathrm{op}} = o_p(1), \qquad \sup_{\bgamma\in\boldsymbol{\Gamma}} \left\| \frac{1}{T}\sum_{t=1}^T D_t(\bgamma)\bx_t Y_t - \mathbb{E}[D_t(\bgamma)\bx_t Y_t] \right\|_2 = o_p(1).
\]
Furthermore, $\frac{1}{T}\sum_{t=1}^T \bx_t Y_t \stackrel{p}{\longrightarrow} \mathbb{E}[\bx_t Y_t]$ and $\widehat{\mathbf{A}}_{00} \stackrel{p}{\longrightarrow} \mathbf{A}_{00}$, which implies $\widehat{\mathbf{A}}_{00}^{-1} \stackrel{p}{\longrightarrow} \mathbf{A}_{00}^{-1}$ by Assumption~3.7(iii). Expanding the sample partialled-out regressor,
\[
\frac{1}{T}\sum_{t=1}^T \widehat{\widetilde{\bx}}_t(\bgamma)Y_t = \frac{1}{T}\sum_{t=1}^T D_t(\bgamma)\bx_t Y_t - \widehat{\mathbf{A}}_{10}(\bgamma)\widehat{\mathbf{A}}_{00}^{-1}\left( \frac{1}{T}\sum_{t=1}^T \bx_t Y_t \right),
\]
Slutsky's theorem and \eqref{eq:test_population_drift} yield the uniform drift approximation
\begin{equation}
\label{eq:test_empirical_drift}
\sup_{\bgamma\in\boldsymbol{\Gamma}} \left\| \frac{1}{T}\sum_{t=1}^T \widehat{\widetilde{\bx}}_t(\bgamma)Y_t - \boldsymbol{\mu}(\bgamma) \right\|_2 = o_p(1).
\end{equation}
Specializing to $\bgamma_0$ gives $\frac{1}{T}\sum_{t=1}^T \widehat{\widetilde{\bx}}_t(\bgamma_0)Y_t = \boldsymbol{\mu}(\bgamma_0) + o_p(1)$.

Scaling the pointwise Wald statistic by $T^{-1}$ and using the partial-score representation,
\begin{align*}
\label{eq:test_pointwise_divergence}
\frac{1}{T}\mathcal{W}_T(\bgamma_0) 
&= \left\{ \frac{1}{T}\sum_{t=1}^T \widehat{\widetilde{\bx}}_t(\bgamma_0)Y_t \right\}^\top \widehat{\mathbf{V}}_T(\bgamma_0)^{-1} \left\{ \frac{1}{T}\sum_{t=1}^T \widehat{\widetilde{\bx}}_t(\bgamma_0)Y_t \right\} \nonumber \\
&= \boldsymbol{\mu}(\bgamma_0)^\top \mathbf{V}_1(\bgamma_0)^{-1} \boldsymbol{\mu}(\bgamma_0) + o_p(1),
\end{align*}
where the second equality follows from \eqref{eq:test_empirical_drift}, $\widehat{\mathbf{V}}_T(\bgamma_0) \stackrel{p}{\longrightarrow} \mathbf{V}_1(\bgamma_0)$, and $\lambda_{\min}\{\mathbf{V}_1(\bgamma_0)\} > 0$. Because $\boldsymbol{\mu}(\bgamma_0) \neq \mathbf{0}$, the quadratic form satisfies $c_0 \coloneqq \boldsymbol{\mu}(\bgamma_0)^\top \mathbf{V}_1(\bgamma_0)^{-1} \boldsymbol{\mu}(\bgamma_0) > 0$. Taking the supremum over $\boldsymbol{\Gamma}$,
\begin{equation*}
\label{eq:proof_sup_diverge}
\frac{1}{T}\mathrm{Sup}\mathcal{W}_T \ge \frac{1}{T}\mathcal{W}_T(\bgamma_0) = c_0 + o_p(1) > 0,
\end{equation*}
with probability approaching one, whence $\mathrm{Sup}\mathcal{W}_T \stackrel{p}{\longrightarrow} \infty$.

Under Assumptions~2.1 and
3.7(ii), 3.7(iv), and 3.7(v),
the bootstrap score process is conditionally tight with $\sup_{\bgamma\in\boldsymbol{\Gamma}} \|\mathbb{G}_T^*(\bgamma)\|_2 = \mathcal{O}_{\mathbb{P}^*}(1)$ in probability, while
Assumption~3.7(iii) ensures uniform non-degeneracy of the
bootstrap covariance estimator, i.e., $\sup_{\bgamma\in\boldsymbol{\Gamma}} \|\widehat{\mathbf{V}}_T^*(\bgamma)^{-1}\|_{\mathrm{op}} = \mathcal{O}_{\mathbb{P}^*}(1)$ in probability. Consequently,
\[
\mathrm{Sup}\mathcal{W}_T^* \le \left( \sup_{\bgamma\in\boldsymbol{\Gamma}} \|\mathbb{G}_T^*(\bgamma)\|_2^2 \right) \left( \sup_{\bgamma\in\boldsymbol{\Gamma}} \|\widehat{\mathbf{V}}_T^*(\bgamma)^{-1}\|_{\mathrm{op}} \right) = \mathcal{O}_{\mathbb{P}^*}(1)
\]
in probability. Because $\mathrm{Sup}\mathcal{W}_T \stackrel{p}{\longrightarrow} \infty$, the conditional tail probability satisfies $q_T^* \coloneqq \mathbb{P}^*(\mathrm{Sup}\mathcal{W}_T^* \ge \mathrm{Sup}\mathcal{W}_T) \stackrel{p}{\longrightarrow} 0$. Conditional on the sample $\mathcal{W}_T$, the $B$ independent bootstrap replications yield $B p_{\mathrm{boot}} \mid \mathcal{W}_T \sim \operatorname{Binomial}(B, q_T^*)$, so that $\mathbb{E}^*[p_{\mathrm{boot}}] = q_T^* \stackrel{p}{\longrightarrow} 0$. Applying the conditional Markov inequality gives $p_{\mathrm{boot}} \stackrel{p}{\longrightarrow} 0$, which ensures $\mathbb{P}(p_{\mathrm{boot}} \le \alpha) \longrightarrow 1$ for any nominal level $\alpha \in (0, 1)$, completing the proof.
\end{proof}

\end{document}